\documentclass[11pt,letterpaper]{article}

\usepackage[T1]{fontenc}
\usepackage[utf8]{inputenc}

\usepackage[margin=1in]{geometry}
\usepackage{setspace}
\usepackage{amsmath,amssymb,amsthm,mathtools}
\usepackage{thmtools}
\usepackage{bbm}
\usepackage{bm}
\usepackage[normalem]{ulem}
\usepackage{mathtools}
\usepackage{extpfeil}
\usepackage{pifont}

\usepackage{algorithm}
\usepackage[noend]{algpseudocode}
\usepackage{graphicx}
\usepackage{booktabs}
\usepackage{caption}
\usepackage{subcaption}
\usepackage{lineno}
\usepackage{tikz}
\usepackage{booktabs}
\usepackage{multirow}
\usepackage{makecell}
\usepackage{threeparttable}
\usepackage{makecell}

\usetikzlibrary{calc,automata,positioning,arrows.meta,decorations.pathreplacing}

\usepackage[numbers,sort&compress]{natbib}
\usepackage[colorlinks=true,
            linkcolor=blue,
            citecolor=blue,
            urlcolor=blue]{hyperref}
\usepackage[nameinlink,capitalise]{cleveref}

\usepackage{microtype}

\usepackage{enumitem}
\setlist{itemsep=2pt,topsep=4pt}

\declaretheorem[numberwithin=section]{theorem}
\declaretheorem[sibling=theorem]{lemma}
\declaretheorem[sibling=theorem]{corollary}
\declaretheorem[sibling=theorem]{proposition}
\declaretheorem[sibling=theorem]{claim}

\declaretheorem[sibling=theorem,style=definition]{definition}
\declaretheorem[sibling=theorem,style=remark]{remark}

\title{3-VASS Reachability is in EXPSPACE}

\author{   
     \bgroup\def\arraystretch{1.5}
     \begin{tabular}{lr}
         \multicolumn{2}{l}{
             \begin{tabular}{l@{\hspace{2.5em}}c@{\hspace{2.5em}}c@{\hspace{2.5em}}r}
                 Weijun Chen\footnote{BASICS, Shanghai Jiao Tong University, China, \href{mailto:cwj2018@sjtu.edu.cn}{cwj2018@sjtu.edu.cn}. ORCID: 0009-0007-2611-0338.} & 
                 Bo Fu\footnote{BASICS, Shanghai Jiao Tong University, China, \href{mailto:fubo0970@sjtu.edu.cn}{fubo0970@sjtu.edu.cn}. ORCID: 0009-0002-6513-8935.} & 
                 Yuxi Fu\footnote{Corresponding author. BASICS, Shanghai Jiao Tong University, China, \href{mailto:fu-yx@cs.sjtu.edu.cn}{fu-yx@cs.sjtu.edu.cn}. ORCID: 0000-0001-6429-7550.} & 
                 Huan Long\footnote{BASICS, Shanghai Jiao Tong University, China, \href{mailto:longhuan@sjtu.edu.cn}{longhuan@sjtu.edu.cn}. ORCID: 0000-0002-1328-6197. }
             \end{tabular}
         } 
         \\ 
         \multicolumn{2}{l}{
             \begin{tabular}{l@{\hspace{2.5em}}c@{\hspace{2.5em}}r}
                 Chengfeng Xue\footnote{BASICS, Shanghai Jiao Tong University, China, \href{mailto:xcf123@sjtu.edu.cn}{xcf123@sjtu.edu.cn}.} & 
                 Qizhe Yang\footnote{Shanghai Normal University, China, \href{mailto:qzyang@shnu.edu.cn}{qzyang@shnu.edu.cn}.} & 
                 Yangluo Zheng\footnote{BASICS, Shanghai Jiao Tong University, China, \href{mailto:wunschunreif@sjtu.edu.cn}{wunschunreif@sjtu.edu.cn}. ORCID: 0009-0000-1028-5458.}
             \end{tabular}
         }
     \end{tabular}
     \egroup
}

\date{}   

\newenvironment{claimproof}[1][Proof of Claim]{
  \pushQED{\qed}
  
  \begin{proof}[#1]
}{
  \end{proof}
}

\crefname{subfigure}{Fig.}{Figs.}
\Crefname{subfigure}{Fig.}{Figs.}

\newcommand{\OpName}[1]{\mathop{\operatorname{\textup{#1}}}\nolimits}
\newcommand{\OpNameSC}[1]{\mathop{\operatorname{\textup{\textsc{#1}}}}\nolimits}

\renewcommand{\Vec}[1]{\bm{#1}}
\newcommand{\norm}[1]{\left\Vert{#1}\right\Vert}
\newcommand{\NORM}[1]{\OpNameSC{norm}(#1)}

\newcommand{\Eff}[1]{\OpNameSC{eff}(#1)}

\newcommand{\NN}{\mathbb{N}}
\newcommand{\ZZ}{\mathbb{Z}}
\newcommand{\QQ}{\mathbb{Q}}

\newcommand{\Supp}{\OpName{supp}}
\newcommand{\CycleSpace}{\OpNameSC{cyc}}
\newcommand{\Len}{\OpNameSC{len}}
\newcommand{\Span}{\OpName{span}}
\newcommand{\Drop}{\OpNameSC{drop}}
\newcommand{\Cone}{\OpNameSC{cone}}
\newcommand{\SeqCone}{\OpNameSC{SeqCone}}

\newcommand{\PolyF}[1]{{\OpName{poly}({#1})}}

\newcommand{\dimcyc}{\dim_{\textup{cyc}}}

\newcommand{\Reach}{\OpName{Reach}}
\newcommand{\ReachqVs}[3]{\OpName{Reach}_{#1}({#2},{#3})}
\newcommand{\ReachVs}[2]{\OpName{Reach}({#1},{#2})}

\newcommand{\Size}{\OpNameSC{size}}
\newcommand{\Reverse}[1]{{#1}^{\mathrm{rev}}}

\newcommand{\PSPACE}{{\normalfont\textsf{PSPACE}}}
\newcommand{\EXPSPACE}{{\normalfont\textsf{EXPSPACE}}}
\newcommand{\TOWER}{{\normalfont\textsf{TOWER}}}
\newcommand{\ACKERMANN}{{\normalfont\textsf{ACKERMANN}}}
\newcommand{\NP}{{\normalfont\textsf{NP}}}
\newcommand{\NL}{{\normalfont\textsf{NL}}}

\newcommand{\RotToEq}{\mathbin{\underline{\curvearrowright}}}
\newcommand{\NotRotToEq}{\mathbin{\not\kern -0.3em \RotToEq}}
\newcommand{\Inner}[1]{\langle {#1} \rangle}
\newcommand{\BDCounters}{\OpNameSC{bd}}

\newcommand{\ResPeriodic}[3]{{#1}^{\Inner{{#2}, \cdot} \le {#3}}}

\newcommand{\bigO}{\mathcal{O}}

\newcommand{\BRetTo}{{\mathop{\;\raisebox{-0.17em}{$\searrow$}\kern-0.88em\raisebox{0.17em}{$\nwarrow$}}}}

\newcommand{\VASSPro}{\mathbb{VASS}^3}

\newcommand{\diagVASSPro}{\rm{Diag}\mathbb{VASS}^3}
\newcommand{\diagVASSProd}{\mathrm{Diag}\mathbb{VASS}^d}

\newcommand{\semidiagVASSPro}{\rm{Semi\text{-}diag}\mathbb{VASS}^3}

\newcommand{\pumpVASSPro}{\rm{Pump}\mathbb{VASS}^3}
\newcommand{\semipumpVASSPro}{\rm{Semi\text{-}pump}\mathbb{VASS}^3}

\newcommand{\seqVASSPro}{\mathrm{Seq}\mathbb{VASS}^3}

\begin{document}
\maketitle
\vspace{-2em}

\begin{abstract}
A VASS can be viewed as a finite-state automaton manipulating a fixed number (called its dimension) of counters holding non-negative values. The reachability problem, asking whether there is a run from one configuration, defined by a state and values of the counters, to another configuration, has been a long-standing algorithmic challenge in theoretical computer science. 
When the dimension is part of the input, the problem has been shown to be \textsf{ACKERMANN}-complete in 2021. For fixed dimension greater than 2, and in particular for dimension 3, the exact complexity of the reachability problem remains unclear. For a long time the known algorithms for the $3$-dimensional VASS reachability problem had been non-elementary, while the best known lower bound is merely \PSPACE{} hardness inherited from dimension 2.
A recent breakthrough in (Czerwi\'nski, Jecker, Lasota, Orlikowski, ICALP 2025) gave the first elementary upper bound for the problem, namely $2$-$\EXPSPACE$.
In this paper it is shown that the reachability problem in 3-VASS belongs to $\EXPSPACE$.
The proof is based on a hierarchical pumpability analysis, yielding a doubly-exponential length bound on the shortest runs between two configurations.
\end{abstract}

\thispagestyle{empty}

\newpage

\section{Introduction}
Petri nets are a fundamental model of concurrency with numerous applications across theoretical computer science. Algorithmically, they are equivalent to vector addition systems with states (VASS). A $d$-dimensional VASS ($d$-VASS) is a finite automaton equipped with $d$ counters whose values range over nonnegative integers and can be updated by applying transition rules. 
In the terminology of Petri net theory a counter value represents the number of tokens in a place; it is never negative.
Let $\NN$ be the set of natural numbers and $\ZZ$ be the set of integers. 
Formally, a $d$-dimensional \emph{vector addition system with states} ($d$-VASS) $V$ is a tuple $(Q,T)$, where $Q$ is a finite set of \emph{states} and $T \subseteq Q \times \ZZ^d \times Q$ is a finite set of \emph{transitions}.
A {\em $\mathbb{Z}$-configuration} $p(\Vec{x})$ is specified by a state $p$ and a vector $\Vec{x}\in\ZZ^d$. 
A transition $(p,\Vec{t},q) \in T$, also denoted by $p\xrightarrow{\Vec{t}}q$, induces a transition relation on $\ZZ$-configurations defined by $p(\Vec{x}) \stackrel{\mathbf{t}}{\longrightarrow}_{\ZZ} q(\Vec{x}\,{+}\,\Vec{t})$ for all $\Vec{x} \in \ZZ^d$. 
A $\mathbb{Z}$-configuration $p(\Vec{a})$ is a {\em configuration} if $\Vec{a} \in \NN^d$.
We write $p(\Vec{a})\xrightarrow{\Vec{t}}q(\Vec{a}\,{+}\,\Vec{t})$ if $\Vec{a},\Vec{a}\,{+}\,\Vec{t}\in\NN^d$.
The following directed graph defines a 3-VASS~\cite{DBLP:journals/tcs/HopcroftP79}.
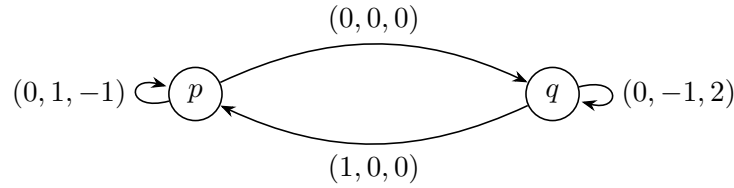
\begin{center}
\begin{tikzpicture}[
  node distance=4cm,
  auto,
  semithick,
  >=Stealth,
  every state/.style={inner sep=1pt, minimum size=20pt},
  every edge/.style={draw, ->}
]
  \node[state] (p) {$p$};
  \node[state] (q) [right=of p] {$q$};

  \path
  (p) edge[bend left=25]
      node[above, align=center]
      {$\left(0,0,0\right)$}
      (q)

  (q) edge[bend left=25]
      node[below, align=center]
      {$\left(1,0,0\right)$}
      (p)

  (p) edge[loop left]
      node[left, align=center]
      {$\left(0,1,-1\right)$}
      (p)

  (q) edge[loop right]
      node[right, align=center]
      {$\left(0,-1,2\right)$}
      (q);
\end{tikzpicture}
\end{center}
It contains 2 states labeled by $p,q$ and 4 transition rules labeled by $\Vec{t}_0:=(0,1,-1)$, $\Vec{t}_1:=(0,0,0)$, $\Vec{t}_2:=(0,-1,2)$ and $\Vec{t}_3:=(1,0,0)$ respectively.
A {\em ($\ZZ$)-run} is a finite sequence of ($\ZZ$)-configurations defined by a sequence of transition instances. 
The following is a particular run.
\[p(0,0,1) \stackrel{\Vec{t}_0}{\longrightarrow} p(0,1,0) \stackrel{\Vec{t}_1}{\longrightarrow} q(0,1,0) \stackrel{\Vec{t}_2}{\longrightarrow} q(0,0,2) \stackrel{\Vec{t}_3}{\longrightarrow} p(1,0,2).\]
We write $p(0,0,1) \rightarrow p(0,1,0) \rightarrow q(0,1,0) \rightarrow q(0,0,2) \rightarrow p(1,0,2)$ even $p(0,0,1) \xrightarrow{*} p(1,0,2)$, where $\xrightarrow{*}$ is the reflexive and transitive closure of $\rightarrow$, if we are not concerned with the labels.
By repeating the transition sequence, one gets
$p(0, 0, 1) \xrightarrow{*} p(n, 0, 2^n)$. 

A central algorithmic question for VASS is the {\em reachability problem}: Given a VASS $V$ and two configurations $s=p(\Vec{a}),t=q(\Vec{b})$, determine whether $s\xrightarrow{*}t$.
We shall call the tuple $(V,s,t)$ a {\em reachability instance}. In the general VASS reachability problem, denoted by $\mathbb{VASS}$, the number of counters, or the dimension of VASS, is unfixed.
The complexity of $\mathbb{VASS}$ has been studied for nearly five decades. 
Lipton established $\EXPSPACE{}$-hardness already in 1976 \cite{Lipton76}, while decidability was first proved by Mayr \cite{DBLP:journals/siamcomp/Mayr84}. Subsequent works simplified and refined the original proof techniques \cite{DBLP:conf/stoc/Kosaraju82, DBLP:journals/tcs/Lambert92}. Their decidability algorithm is now known as the KLM algorithm. 
Later, Leroux and Schmitz showed that the problem admits an Ackermannian upper bound \cite{DBLP:conf/lics/LerouxS19}. 
On the lower-bound end, the complexity was successively improved from $\TOWER{}$-hardness \cite{CzerwinskiTowerLB19} to $\ACKERMANN{}$-hardness \cite{DBLP:conf/focs/CzerwinskiO21,DBLP:conf/focs/Leroux21}, thereby matching the upper bound.

In the {\em $d$-dimensional VASS reachability problem}, denoted by $\mathbb{VASS}^d$, the number of counters $d$ is fixed.
The most recent refinement of the KLM algorithm~\cite{DBLP:conf/icalp/FuYZ24} has revealed that $\mathbb{VASS}^d$ is in $\mathsf{F}_{d}$ for all $d>2$, which improves the previous upper bound $\mathsf{F}_{d+4}$~\cite{DBLP:conf/lics/LerouxS19}. The complexity classes $\mathsf{F}_3,\mathsf{F}_4,\ldots,\mathsf{F}_i,\ldots$ are non-elementary complexity classes rendering true the inclusion sequence $\TOWER{}=\mathsf{F}_3\subseteq\mathsf{F}_4\subseteq\ldots\subseteq\mathsf{F}_i\subseteq\ldots \ACKERMANN{}$. For more on these classes, consult~\cite{schmitz2016complexity}.

The {\em size} of a reachability instance $(V,s,t)$, denoted by $\Size(V,s,t)$, is the encoding of $(V,s,t)$ in either binary or unary.
The difference matters for $\mathbb{VASS}^1$ and $\mathbb{VASS}^2$.
It is immaterial for the current upper bound results about $\mathbb{VASS}^d$ with $d>2$.
$\mathbb{VASS}^1$ is $\NP{}$-complete~\cite{DBLP:conf/concur/HaaseKOW09} under binary encoding and $\NL{}$-complete under unary encoding~\cite{10.1016/S0022-0000(75)80005-5}. $\mathbb{VASS}^2$ is $\PSPACE{}$-complete~\cite{DBLP:journals/jacm/BlondinEFGHLMT21} under binary encoding and 
$\NL{}$-complete~\cite{englert2016reachabilitytwodimensionalunaryvector} under unary encoding. 
The insights into 2-dimensional runs are revealing. All $2$-dimensional runs are short in that they can be described by regular expressions, called {\em linear path schemes}, irrelevant to the source and target configurations~\cite{DBLP:journals/jacm/BlondinEFGHLMT21}.
When the encoding scheme is in binary, the linear path schemes are of exponential size and contain a polynomial number of $*$-expressions.
This algebraic view helps see that $2$-dimensional runs can be characterized by linear Diophantine systems with a polynomial number of variables. 
Further studies show that all $2$-dimensional runs essentially take a zigzag shape in the first quadrant~\cite{Chistikov_2024}. 

In recent years, a special class of VASS has been shown to play an important role in the studies of VASS reachability.
A geometrically $2$-dimensional VASS is one in which all $\mathbb{Z}$-runs stay between two parallel 2-dimensional planes~\cite{DBLP:conf/icalp/FuYZ24,DBLP:conf/concur/Zheng25}.
It has been shown that the reachability problem of geometrically 2-dimensional VASS is also $\PSPACE{}$-complete under binary encoding~\cite{DBLP:conf/concur/Zheng25} and $\NP{}$-complete under unary encoding~\cite{chen2026improvingreachabilityvectoraddition}. 
Further properties of the geometrically $2$-dimensional runs have been unveiled in the recent work by Czerwinski et~al.~\cite{DBLP:conf/icalp/CzerwinskiJ0O25}.
Using these properties, the authors have come up with a proof that the 3-VASS reachability problem is in $2$-$\EXPSPACE{}$. 

\begin{table}[t]
\centering
\renewcommand{\arraystretch}{1.2}
\begin{tabular}{c|c|c}
\Xhline{1.1pt}
\multirow{2}{*}{\textbf{Dimension $d$}}
  & \multicolumn{2}{c}{\textbf{Reachability in $d$-dimensional VASS}} \\
\cline{2-3}
  & \textbf{Unary encoding} & \textbf{Binary encoding} \\
\hline
$d=1$
  & $\NL{}$-complete~\cite{10.1016/S0022-0000(75)80005-5}
  & $\NP{}$-complete~\cite{DBLP:conf/concur/HaaseKOW09} \\
\hline
$d=2$
  & $\NL{}$-complete~\cite{englert2016reachabilitytwodimensionalunaryvector}
  & $\PSPACE{}$-complete~\cite{DBLP:journals/jacm/BlondinEFGHLMT21} \\
\hline
\multirow{2}{*}{$d=3$}
  & $\NP{}$-hard~\cite{Chistikov_2024} & $\PSPACE{}$-hard \\
  \cline{2-3}
  & \multicolumn{2}{c}{
      $2$-\EXPSPACE{}~\cite{DBLP:conf/icalp/CzerwinskiJ0O25}, 
      \EXPSPACE{} [{\color{red} this paper}]
    }
   \\
\hline
\multirow{2}{*}{$d\geq 4$}
  & \multicolumn{2}{c}{
      $\mathsf{F}_{\lfloor(d-3)/2\rfloor}$-hard
      for $d\geq 9$~\cite{CzerwinskiNewLB23}
    } \\
  \cline{2-3}
  & \multicolumn{2}{c}{
      $\mathsf{F}_d$~\cite{DBLP:conf/icalp/FuYZ24}
    } \\
\Xhline{1.1pt}
\end{tabular}
\caption{Complexity landscape for $\mathbb{VASS}^d$.}
\label{tab:reachability_complexity}
\end{table}

From the complexity theoretical viewpoint, $\mathbb{VASS}^1$ and $\mathbb{VASS}^2$ have been settled, confer \cref{tab:reachability_complexity}. 
Lower bounds for dimensions 4--8 are studied in~\cite{CzerwinskiO22fixed-lower-bound}.
No completeness result however is known for $\mathbb{VASS}^d$ for any $d>2$.
The priority of study is naturally given to $\mathbb{VASS}^3$.
Can the lower bound $\PSPACE{}$ and the upper bound $2$-$\EXPSPACE{}$ be improved?
This paper partly answers this question. 

\paragraph*{Main Contribution.}
Our main result improves the $2$-$\EXPSPACE{}$ upper bound of \cite{DBLP:conf/icalp/CzerwinskiJ0O25}. 
\begin{restatable}{theorem}{VASSTINEXP}
\label{thm:3-vass-in-exp-intro}
The $3$-VASS reachability problem is in \EXPSPACE{}.
\end{restatable}
It is worth pointing out that \cref{thm:3-vass-in-exp-intro} is valid for both unary encoding and binary encoding. 
To state the main proposition from which the theorem is derived, we need to introduce some notations. 
Let $\min\Len(V,s,t)$ denote the length of the shortest runs from $s$ to $t$ admitted by $V$.
A VASS is \emph{strongly connected} if it is strongly connected as a directed graph. 
A run in a VASS traverses a series of distinct strongly connected components (SCC). 
To capture such a structure, {\em sequential VASS} is introduced in \cite{DBLP:conf/icalp/CzerwinskiJ0O25}.
Let $(V_i)_{i \in [k]}$ be a sequence of strongly connected VASS, where $V_i = (Q_i, T_i)$ for each $i \in [k]$, and $(u_i)_{i \in [k-1]}$ be transitions, where $u_i=(q_i, \Vec{a}_i, p_{i+1})$ for some $q_i \in Q_i$ and $p_{i+1} \in Q_{i+1}$.
A sequential VASS, presented as $(V_1) u_1 (V_2) u_2 \dots u_{k-1} (V_k)$, is a VASS $V=(Q,T)$, where 
\begin{equation*}
    Q=\bigcup_{i \in [k]}Q_i\quad\text{and}\quad T=\bigcup_{i \in [k]} T_i \cup \{u_i\}_{i \in [k-1]}.
\end{equation*}
Such a VASS $V$ is a \emph{$k$-component} (sequential) VASS in which $V_1,\ldots,V_k$ are \emph{components} and $u_1,\ldots,u_{k-1}$ are \emph{bridges}. 
Let $\mathrm{Seq}\mathbb{VASS}^d$ be the subset of $\mathbb{VASS}^d$ defined by the sequential $d$-VASS.
Not every VASS is sequential of course. 
But the transition sequence of a run in $V$ induces a sequence of SCC components of $V$ connected by transitions. Such a sub-VASS is essentially sequential.
From the viewpoint of reachability, $\mathrm{Seq}\mathbb{VASS}^d$ is equivalent to $\mathbb{VASS}^d$.
We can now state the key proposition.
\begin{restatable}{proposition}{propThreeVASSlenboundIntro}
\label{prop:3-VASS-lenbound-intro}
If $(V, s, t)\in\mathrm{Seq}\mathbb{VASS}^3$ has $k$ components, $\min\Len(V,s,t)\le\Size(V, s, t)^{2^\PolyF{k}}$.
\end{restatable}
\cref{thm:3-vass-in-exp-intro} is an immediate consequence of \cref{prop:3-VASS-lenbound-intro}.
For completeness, we include the standard proof in \cref{sec:proofofmaintheorem}.
In the rest of the section, we take a look at the existing analyses for $\mathbb{VASS}^3$, and outline how \cref{prop:3-VASS-lenbound-intro}, as the foundation of the new complexity bound, is proved.

\paragraph*{Existing Approaches to $\mathbb{VASS}^3$.}
Attempts to obtain a complexity upper bound for $\mathbb{VASS}^3$ date back to the KLM algorithm. In~\cite{DBLP:conf/lics/LerouxS19}, Leroux and Schmitz obtained an $\mathsf{F}_{d+4}$ upper bound for $\mathbb{VASS}$, yielding an $\mathsf{F}_7$ upper bound for $\mathbb{VASS}^3$. The core philosophy of the KLM algorithm is to reduce high-dimensional VASS instances into lower-dimensional ones. 
It is subsequently optimized by Fu, Yang, and Zheng~\cite{DBLP:conf/icalp/FuYZ24} by sharpening the characterization of low-dimensional base cases. They utilized a more tractable equation system for geometrically $2$-VASS, thereby lowering the complexity to $\mathsf{F}_d$ for $\mathbb{VASS}^d$. Particularly, for $\mathbb{VASS}^3$, the upper bound specializes to $\mathsf{F}_3 = \TOWER{}$.

It is commonly accepted that the KLM algorithm is merely a brute-force construction of a possible run when the number of runs is finite. 
Upon receiving an input $(V,s,t)$, the KLM algorithm assumes that there is a run from $s$ to $t$, and then carries out a series of transformations on the VASS, while maintaining the reachability.
This is done without considering any special properties of the input VASS.
For $3$-VASS, the KLM algorithm probably reaches its efficiency bottleneck, because instances can be found that admit finitely many but $\TOWER{}$-large runs.

The next refinement is due to Czerwiński, Jecker, Lasota and Orlikowski \cite{DBLP:conf/icalp/CzerwinskiJ0O25}, who obtained the first elementary upper bound for $\mathbb{VASS}^3$.
To achieve that, they developed a number of techniques, chief among them is an efficient representation of the {\em reachability set}, which is the set of all vectors reachable from an initial configuration in a geometrically $2$-dimensional VASS. 
Quite a few properties about the efficient representation of the reachability set are established in~\cite{DBLP:conf/icalp/CzerwinskiJ0O25}.
The second main technique of~\cite{DBLP:conf/icalp/CzerwinskiJ0O25} allows us to bound the length of shortest runs by induction on the number of components of a sequential VASS.
Let $h_k$ be an upper bound function for the length of the shortest runs of the $k$-component sequential VASS.
Suppose $V$ is $(V_1) u_1 (V_2) u_2 \dots u_{k-1} (V_k)$.
Then $\min\Len(V,s,t)\le h_k(\Size(V,s,t))$.
Suppose $s\stackrel{\pi_1}{\longrightarrow}s_1\stackrel{\pi'}{\longrightarrow}t$ is a {\em shortest} run from $s$ to $t$, where $s\stackrel{\pi_1}{\longrightarrow}s_1$ is a run admitted by the VASS $(V_1)u_1$ and $s_1\stackrel{\pi'}{\longrightarrow}t$ is a run admitted by the $(k{-}1)$-component VASS $V':=(V_2)u_2\dots u_{k-1}(V_k)$.
The idea is to seek a suitable induction that bounds the length of $s\stackrel{\pi_1}{\longrightarrow}s_1$ using the inductively obtained length bound of $s_1\stackrel{\pi'}{\longrightarrow}t$. 
The trouble is that the size of $s_1$ is unknown.
The solution of~\cite{DBLP:conf/icalp/CzerwinskiJ0O25} is to rewrite $s\stackrel{\pi_1}{\longrightarrow}s_1\stackrel{\pi'}{\longrightarrow}t$ to some $s\stackrel{\pi_1'}{\longrightarrow}s_1'\stackrel{\pi''}{\longrightarrow}t$ such that $s_1'$ is small.
Using a sophisticated argument, they are able to bound the length of $\pi''$, the size of $s_1'$, the size of $s_1$, and the length of $\pi'$ in turn, producing an upper bound of $h_k$ by an expression of the form $h_{k-1}(h_{k-1}(\Size(V',s_1,t)))$.
The nested occurrence of $h_{k-1}$ explains the triple-exponential upper bound $\Size(V,s,t)^{2^{2^{\mathcal{O}(k)}}}$ for $h_{k}(\Size(V,s,t))$.

\paragraph*{New Approach to $\mathbb{VASS}^3$.}  
The second invocation of $h_{k-1}$ must be dispensed with in order to improve the upper bound from $2$-$\EXPSPACE{}$ to $\EXPSPACE{}$.
We will prove some global properties about reachability sets in \cref{sec:repr-2vass} and adopt a separation-of-concern proof strategy in \cref{sec:DP3VASS} through
\cref{sec:general-case}.
Using the finer-tuned approach, we can do away with the second reference to $h_k$.
Our ultimate goal is to show that for each $k$-component $(V,s,t)\in\mathrm{Seq}\mathbb{VASS}^3$, 
\begin{equation}\label{2026-06-26-airport}
\min\Len(V,s,t)\le\Size(V,s,t)^{2^{\PolyF{k}}}. 
\end{equation}
The inequality~(\ref{2026-06-26-airport}) will be proved in a step-by-step fashion.
The following is an outline of our method. 
\begin{enumerate}
\item \label{2026-06-30-1240}
We start with a special subset $\mathrm{Diag}\mathbb{VASS}^3$ of $\mathrm{Seq}\mathbb{VASS}^3$.
A sequential $d$-VASS is in $\mathrm{Diag}\mathbb{VASS}^d$ if it is {\em diagonal}.
The definition of diagonal $d$-VASS will be given in~\cref{sec:overview}. 
It is shown in~\cite{DBLP:conf/icalp/CzerwinskiJ0O25} that $\mathrm{Diag}\mathbb{VASS}^3\in$ $2$-$\EXPSPACE{}$. 
We will strengthen this result to $\mathrm{Diag}\mathbb{VASS}^3\in\EXPSPACE{}$. 
By enriching our knowledge about the efficient representation of the reachability sets of the geometrically $2$-dimensional VASS, and by sharpening the analysis of~\cite{DBLP:conf/icalp/CzerwinskiJ0O25}, we are able to establish the inequality~(\ref{2026-06-26-airport}) for the diagonal VASS.
This is explained in \cref{sec:DP3VASS}.
\item 
Once there is an exponential space algorithm for $\mathrm{Diag}\mathbb{VASS}^3$, it can be used as a subroutine to help solve $\mathrm{Seq}\mathbb{VASS}^3$.
However $\mathrm{Seq}\mathbb{VASS}^3$ and $\mathrm{Diag}\mathbb{VASS}^3$ are far apart as it were that it is difficult to see how to invoke an algorithm for the latter to solve the former.
A key step in our approach is to introduce an intermediate problem $\mathrm{Pump}\mathbb{VASS}^3$ such that $\mathrm{Diag}\mathbb{VASS}^3\subsetneq\mathrm{Pump}\mathbb{VASS}^3\subsetneq\mathrm{Seq}\mathbb{VASS}^3$.
The class $\mathrm{Pump}\mathbb{VASS}^3$ contains all the {\em pumpable} instances in $\mathrm{Seq}\mathbb{VASS}^3$, whose definition will be given in~\cref{sec:overview}. 
We will define a reduction from $\mathrm{Pump}\mathbb{VASS}^3$ to $\mathrm{Diag}\mathbb{VASS}^3$ to show that the inequality~(\ref{2026-06-26-airport}) is valid for the pumpable VASS, using the fact that it is valid for the diagonal VASS.
Similarly, we will define a reduction from $\mathrm{Seq}\mathbb{VASS}^3$ to $\mathrm{Pump}\mathbb{VASS}^3$ to show that the inequality~(\ref{2026-06-26-airport}) is valid for the sequential VASS, using the fact that it is valid for the pumpable VASS.
The two reductions are studied in \cref{sec:par-pump-3-vass} and in \cref{sec:general-case} respectively.
We should mention that such a reduction is a self-reduction in that it always maps an instance of a class onto an instance of the same class. 
The latter instance may or may not be in the target class, but it always has fewer components.
In other words, the self-reduction on the source class is controlled by the target class.
\end{enumerate}

\paragraph*{Organization.}
\cref{Sec-Preliminaries} fixes the notations, terminologies, and the necessary lemmas.
\cref{sec:overview} gives an overview of the proof strategy.
\cref{sec:repr-2vass} studies the efficient representation of reachability sets.
\cref{sec:DP3VASS}, \cref{sec:par-pump-3-vass} and
\cref{sec:general-case} derive the new complexity upper bound of $\mathrm{Diag}\mathbb{VASS}^3$, $\mathrm{Pump}\mathbb{VASS}^3$ and $\mathrm{Seq}\mathbb{VASS}^3$, respectively.
\cref{sec:conclusion} concludes with some open problems.

\section{Preliminaries}\label{Sec-Preliminaries}

The terminologies and results referred to in the rest of the paper are introduced here. 

\paragraph*{Vector Space.}
Let $\QQ$ denote the set of rational numbers. We use subscripts to indicate restrictions to $\NN,\ZZ,\QQ$. For instance, $\QQ_{\ge 0}$ denotes the set of the non-negative rationals. 
Vectors are denoted in bold, e.g., $\Vec{u}, \Vec{v}, \Vec{x}, \Vec{y}$. 
All vectors are column vectors by default.
We write $\Vec{x}^{\dag}$ for the transpose of the vector $\Vec{x}$. 
The $i$-th value of $\Vec{x}$ is denoted by $\Vec{x}(i)$ or $\Vec{x}_i$.
We denote by $\Vec{1}$ and $\Vec{0}$ the all-ones and all-zeros vectors, respectively.
For $n\in\NN$, let $[n]:=\{1,\ldots,n\}$.
We write $\Vec{e}^{(d)}_i$ or $\Vec{e}_i$ for the unit vector satisfying $\Vec{e}_i(i) = 1$ and $\Vec{e}_i(j) = 0$ for all $j \in[d]\setminus\{i\}$. 
A vector $\Vec{x} \in \mathbb{Q}^d$ is \emph{positive} if $\Vec{x}(i) > 0$ for all $i \in [d]$.
A vector $\Vec{x}$ is {\em negative} if $-\Vec{x}$ is positive.
Vector comparisons are understood component-wise: $\Vec{x} \le \Vec{y}$ means $\Vec{x}(i) \le \Vec{y}(i)$ for all $i \in [d]$. 
We write $\Vec{x} < \Vec{y}$ if $\Vec{x} \le \Vec{y}$ and $\Vec{x} \ne \Vec{y}$. 
The {\em (max-)norm} of a vector $\Vec{x}$ in $\QQ^d$ is $\norm{\Vec{x}} := \max_{i \in [d]} |\Vec{x}(i)|$. 
The inner product of $\Vec{x},\Vec{y} \in \QQ^d$ is $\langle \Vec{x}, \Vec{y} \rangle := \sum_{i \in [d]} \Vec{x}(i) \cdot \Vec{y}(i)$. 
For a finite set $X = \{\Vec{x}_1,\ldots,\Vec{x}_n\} \subseteq \QQ^d$, the {\em span} of $X$ is
$span(X) := \{\lambda_1 \Vec{x}_1 + \cdots + \lambda_n \Vec{x}_n \mid \lambda_1, \ldots, \lambda_n \in \QQ\}$.

\paragraph*{VASS.}
Suppose $V=(Q,T)$ is a $d$-VASS.
The \emph{reverse} of $V$ is the VASS $\Reverse{V} = (Q, \Reverse{T})$ obtained by inverting all the transitions, where $\Reverse{T} := \{ (q, -\Vec{a}, p) \mid (p, \Vec{a}, q) \in T \}$. 
Let $\pi:=\tau_1\tau_2\dots \tau_k$ be a sequence of transitions, where $\tau_i = (p_i,\Vec{a}_i,q_i)$ for each $i \in [k]$.
We say that $\pi$ is a \emph{path} if $q_i=p_{i+1}$ for all $i \in [k-1]$. 
If additionally $q_k=p_1$, then $\pi$ is called a \emph{cycle}. 
A cycle consisting of a single transition is a \emph{self-loop}. 
An empty path is denoted by $\varepsilon$.
The length of $\pi$ is $|\pi| := k$.
The \emph{effect} of $\pi$ is defined as $\Eff{\pi} := \sum_{i = 1}^k \Vec{a}_i$. 
A cycle is \emph{positive}, respectively {\em negative}, if its effect is positive, respectively negative.
Given two paths $\pi_1$ and $\pi_2$, their concatenation $\pi_1\pi_2$ is defined whenever the last state of $\pi_1$ coincides with the first state of $\pi_2$. If $\pi$ is a cycle and $m \in \NN$, we denote by $\pi^m$ its $m$-fold concatenation of $\pi$. 
The norm of a ($\ZZ$-)configuration $p(\Vec{x})$ is $\norm{p(\Vec{x})}:=\norm{\Vec{x}}$.

\paragraph*{Size Function.}
The \emph{unary size} of a VASS $V$ is $\Size(V) := |Q|+ d \cdot  |T| \cdot (\norm{T} + 1)$, where $\norm{T} := \max_{(p, \Vec{a}, q) \in T}\|\Vec{a}\|$.
The size $\Size(V, s, t)$ of $(V, s, t)$ is defined by $\Size(V) + d\cdot (\norm{s} + \norm{t} + 1)$, and $\Size(V, s)$ is defined by $\Size(V) + d\cdot (\norm{s} + 1)$. 
The {\em binary size} of $(V, s, t)$ is
\begin{equation*}
    \Size_{\textsc{bin}}(V,s,t) := |Q| + d \cdot |T| \cdot \log (\norm{T} + 1) + d \cdot \log (\norm{s} + 1) + d\cdot \log (\norm{t} + 1).
\end{equation*}
In the rest of the paper, when we say size, we mean the unary size. 

\paragraph*{Drop Function.}
The \emph{drop} of a path $\pi = \tau_1\tau_2\dots \tau_k$, denoted by $\Drop(\pi)$, is a vector in $\NN^d$ that captures the maximal decrease incurred along $\pi$. Formally, for each coordinate $i \in [d]$, let
\begin{equation*}
    \Drop(\pi)(i):= -\min(0,\min_{j\in [k]}\sum_{\ell \in [j]}\Eff{\tau_\ell}(i)).
\end{equation*}
Assume that the source and target states of $\pi$ are $p$ and $q$, respectively.
It follows immediately that for any $\Vec{z}\ge \Drop(\pi)$, $p(\Vec{z}) \xrightarrow{\pi} q(\Vec{z}+\Eff{\pi})$, and in this case, we say that $\pi$ is \emph{fireable} at $p(\Vec{z})$.

\paragraph*{Complexity Function.}
We write $\PolyF{n}$ and $\bigO(n)$ for a polynomial and a linear function on natural numbers respectively without being specific about constants. 
We also assume, without loss of generality, that every function $f:\NN \to \NN$ used as an upper bound in this paper is \emph{non-decreasing} in the sense that $f(m) \ge f(n)$ whenever $m \ge n$ and $f(n) \ge n + 1$ holds for each $n \ge 1$.

\paragraph*{Sequential Cone.}
For a non-empty finite set $X \subseteq \QQ^d$, let $\norm{X} := \max_{\Vec{x} \in X} \norm{\Vec{x}}$. 
For $X, Y \subseteq \QQ^d$, let $X + Y := \{\Vec{x} + \Vec{y} \mid \Vec{x} \in X, \Vec{y} \in Y\}$, and $X - Y$ is defined similarly. Also, $-X$ is $\{\Vec{0}\} - X$. 
For a finite set $X = \{\Vec{x}_1,\ldots,\Vec{x}_n\} \subseteq \QQ^d$, we denote by 
\begin{equation*}
    \Cone(X) := \{\lambda_1 \Vec{x}_1 + \cdots + \lambda_n \Vec{x}_n \mid \lambda_1, \ldots, \lambda_n \in \QQ_{\ge 0}\}
\end{equation*}
the {\em (closed rational) cone} generated by $X$. We also consider sequential cones over a sequence of cones. Let $C_1, \ldots, C_k \subseteq \QQ^d$ be cones, the \emph{sequential cone} over them is
\begin{equation*}
    \SeqCone(C_1, \ldots, C_k) := \left\{
        \Vec{v}_1 + \cdots + \Vec{v}_k \mid \text{for all } i \in [k], \Vec{v}_i \in C_i \text{ and } \Vec{v}_1 + \cdots + \Vec{v}_i \ge \Vec{0}
    \right\}.
\end{equation*}
We sometimes consider the cone of $X$ under natural coefficients, which is denoted by
\begin{equation*}
    X^* := \{\lambda_1 \Vec{x}_1 + \cdots + \lambda_n \Vec{x}_n \mid \lambda_1, \ldots, \lambda_n \in \NN\}.
\end{equation*}

\paragraph*{Sequential Cone of VASS.}
Let $V$ be a $d$-VASS. Its \emph{cycle space} $\CycleSpace(V)\subseteq \QQ^d$ is the vector space spanned by the effects of simple cycles in $V$, i.e., $\CycleSpace(V) := \Span\{\Eff{\theta} \mid \theta \text{ is a simple cycle in } V\}$. The \emph{cycle-dimension} $\dimcyc(V) := \dim(\CycleSpace(V))$ is the dimension of the cycle space. This is known as \emph{geometric dimension} in \cite{DBLP:conf/icalp/FuYZ24,DBLP:conf/concur/Zheng25}. 
The (closed) cone of $V$ is the cone generated by the effects of simple cycles in $V$, i.e., $\Cone(V) := \Cone(\{\Eff{\theta} \mid \theta \text{ is a simple cycle in } V\})$. In particular, if $V$ is a sequential VASS $V = (V_1) u_1 (V_2) u_2 \ldots u_{k-1} (V_k)$, we define the (closed) \emph{sequential cone} of $V$ by $\SeqCone(V) := \SeqCone(\Cone(V_1), \Cone(V_2), \ldots, \Cone(V_k))$.
By definition, if $V$ has a single component, then $\SeqCone(V) = \Cone(V)$. 
We remark that both cones and sequential cones considered here are closed, which is different from those defined in~\cite{DBLP:conf/icalp/CzerwinskiJ0O25}. 
Observe that $\Cone(\Reverse{V}) = -\Cone(V)$ and $\SeqCone(\Reverse{V}) = \SeqCone(-\Cone(V_k), -\Cone(V_{k-1}), \ldots, -\Cone(V_1))$.

\paragraph*{Reachability Set.} The set of the configurations reachable from $s$ is denoted by $\ReachVs{V}{s}$. 
If the target state is $q\in Q$, the notation $\ReachqVs{q}{V}{s}$ denotes the set of the vectors $\Vec{w}\in\NN^d$ satisfying $q(\Vec{w}) \in \ReachVs{V}{s}$. 
Suppose $C$ is a set of configurations. 
Let $\ReachqVs{q}{V}{C}$ be $\bigcup_{c \in C} \ReachqVs{q}{V}{c}$, and $\ReachVs{V}{C}$ be $\bigcup_{c \in C} \ReachVs{V}{c}$. 
For $S \subseteq \NN^d$, let $p(S):=\{p(\Vec{s}) \mid \Vec{s} \in S\}$. 
The notation $p(S) \xrightarrow{*} t$ indicates the fact that $p(\Vec{x}) \xrightarrow{*} t$ for some $\Vec{x} \in S$.

\paragraph*{Length Set.}
Given a VASS $(V, s, t)$, the \emph{length set} $\Len(V, s, t)$ is the set of the lengths of runs from $s$ to $t$ in $V$. 
Notice that every run induces from $V$ a sequential VASS. We may collect all such sequential VASS to form a finite family $\mathcal{V}$ of sequential VASS such that $\Len(V,s,t) = \bigcup_{V'\in\mathcal{V}}\Len(V',s,t)$ and $\Size(V')\le \Size(V)$ for every $V' \in \mathcal{V}$.

\paragraph*{Integer Programming.}
Given a matrix $A \in \QQ^{m\times n}$, let $\norm{A}$ denote the maximum of the absolute values of all the entries of $A$.
Various bounds on integer solutions to linear inequalities have been studied in the literature \cite{pottier1991minimal, vonzurGathen1978, DBLP:conf/icalp/ChistikovH16}. 
We shall refer to the results stated below.

\begin{lemma}[{\cite[Prop.\ 4]{DBLP:conf/icalp/ChistikovH16}}]
    \label{lem:lin-eq-decomp}
    Let $A \in \ZZ^{m \times n}$ be a matrix and $\Vec{b} \in \ZZ^m$ be a vector. Define $X = \{\Vec{x} \in \NN^n \mid A\Vec{x} = \Vec{b}\}$ and $X_0 = \{\Vec{x} \in \NN^n \mid A\Vec{x} = \Vec{0} \}$.
    Any vector in $X$ can be decomposed as the sum of a vector $\Vec{x} \in X$ and a finite sum of vectors $\Vec{x}_0 \in X_0$ where 
    \begin{equation*}
        \norm{\Vec{x}} \le ((n + 1)\norm{A} + \norm{\Vec{b}} + 1)^m\ \text{and}\  \norm{\Vec{x}_0} \le (n\norm{A} + 1)^m.
    \end{equation*}
\end{lemma}
\noindent

\begin{corollary}
    \label{cor:ilp-decomp}
    Let $A \in \ZZ^{m \times n}$, $C \in \ZZ^{k \times n}$ be two matrices and $\Vec{b} \in \ZZ^m$, $\Vec{d} \in \ZZ^{k}$ be two vectors. Define $X = \{\Vec{x} \in \NN^n \mid A\Vec{x} \le \Vec{b}\text{ and }C\Vec{x} = \Vec{d}\}$ and $X_0 = \{\Vec{x} \in \NN^n \mid A\Vec{x} \le \Vec{0} \text{ and }C\Vec{x}=\Vec{0}\}$. Let $M:=\max\{\norm{A},\norm{C}\}$ and $N:=\max\{\norm{\Vec{b}},\norm{\Vec{d}}\}$. Then any vector in $X$ can be decomposed as the sum of a vector $\Vec{x} \in X$ and a finite sum of vectors $\Vec{x}_0 \in X_0$ where 
    \begin{equation*}
        \norm{\Vec{x}} \le ((n + m + k + 1)M + N + 1)^{m+k}\ \text{and}\ \norm{\Vec{x}_0} \le ((n + m + k)M + 1)^{m+k}.
    \end{equation*}
\end{corollary}

\begin{proof}
    By introducing slack variables $\Vec{y} \in \NN^m$, we can replace the inequality $A \Vec{x} \le \Vec{b}$ by the equation $A \Vec{x} + \Vec{y} = \Vec{b}$. Then apply~\cref{lem:lin-eq-decomp}.
\end{proof}

\begin{corollary}
    \label{lem:pottier}
    Let $A \in \ZZ^{m \times n}$ be a matrix and $\Vec{b} \in \ZZ^m$ be a vector. If there exists a non-negative vector $\Vec{x} \in \NN^n$ such that $A\Vec{x} = \Vec{b}$, then there is also one satisfying $\norm{\Vec{x}} \le \bigO(nN)^m$.
\end{corollary}

\section{Overview of the Proof}
\label{sec:overview}

The goal of the rest of the paper is to prove \cref{prop:3-VASS-lenbound-intro}.
We shall describe in this section the key aspects of the proof.
We begin by comparing our proof to the proof in~\cite{DBLP:conf/icalp/CzerwinskiJ0O25}.
The algorithm  in~\cite{DBLP:conf/icalp/CzerwinskiJ0O25} is based on a self-reduction on $\seqVASSPro$, confer~\cref{subfig:previous_work} where the black arrows indicate self-reduction and the orange dashed arrows indicate the invocation dependency.
In our proof, the class $\seqVASSPro$ is refined to a sequence of subclasses, giving rise to the following problem hierarchy. 
\begin{equation}\label{2026-07-01-0732}
\diagVASSPro \subsetneq \semidiagVASSPro \subsetneq \pumpVASSPro \subsetneq \semipumpVASSPro \subsetneq \seqVASSPro.
\end{equation}
Our idea is to establish the inequality~\cref{2026-06-26-airport} for all the problems in~\cref{2026-07-01-0732}.
The self-reductions must respect the hierarchy in the sense that a reachability instance in a class should not be reduced to any super class. 
For each class, the reduction relies on the length bound already established for the more restrictive subclass. 
For example, when reducing from $\semidiagVASSPro$, the length bound for $\diagVASSPro$ is utilized to constrain the size blow-up. 
This allows us to transfer sharper bounds from smaller subclasses to larger ones, removing the issue of nested occurrence completely. 
See~\cref{subfig:our_reduction}.
Now ${\rm Seq}\mathbb{VASS}^d$ has been defined.
The other problems are defined below.
\begin{definition}
    Let $(V,s,t)\in{\rm Seq}\mathbb{VASS}^d$. 
    \begin{itemize}
    \item 
    $(V,s,t)$ is {\em forward diagonal} if there exists a cycle $\theta$ fireable at $s$ such that $\Eff{\theta} \geq \Vec{1}$. 
    It is {\em backward diagonal} if $(\Reverse{V},t,s)$ is forward diagonal.
    $(V,s,t)$ is {\em diagonal} if it is both forward and backward diagonal.
    \item 
    $(V,s,t)$ is {\em forward pumpable} if there exists a cycle $\theta$ fireable at $s$ such that $\Eff{\theta} > \Vec{0}$. 
    It is {\em backward pumpable} if $(\Reverse{V},t,s)$ is forward pumpable.
    $(V,s,t)$ is {\em pumpable} if it is both forward and backward pumpable.
    \end{itemize}
    For convenience, we also say that $s$ is diagonal/pumpable in $V$ without explicitly referring to $t$, if $(V,s,t)$ is forward diagonal/pumpable.
\end{definition}

${\rm Diag}\mathbb{VASS}^d$ contains all the diagonal instances in ${\rm Seq}\mathbb{VASS}^d$, and ${\rm Pump}\mathbb{VASS}^d$ contains all the pumpable instances in ${\rm Seq}\mathbb{VASS}^d$. 
By definition, a diagonal VASS is also pumpable.
${\rm Semi}$-${\rm diag}\mathbb{VASS}^d$ contains all the instances in ${\rm Pump}\mathbb{VASS}^d$ that are forward or backward diagonal. 
${\rm Semi}$-${\rm pump}\mathbb{VASS}^d$ contains all the instances in ${\rm Seq}\mathbb{VASS}^d$ that are forward or backward pumpable. 
We remark that all these instances are reachable, as they can be viewed as subproblems of ${\rm Seq}\mathbb{VASS}^d$.
Therefore, $\min \Len(V,s,t)$ is well-defined whenever $(V,s,t)$ belongs to any of these classes.
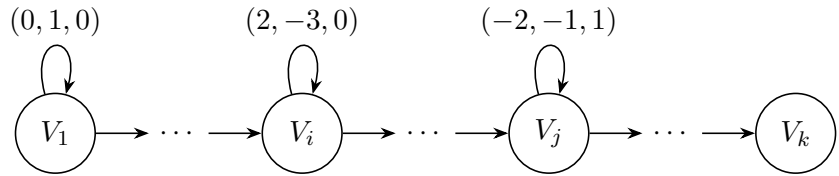
\begin{figure}[h]
\centering
\begin{tikzpicture}[
    node distance=0.7cm,
    auto,
    semithick,
    >=Stealth,
    every state/.style={inner sep=1pt, minimum size=30pt},
    every edge/.style={draw, ->},
    every loop/.style={looseness=8, min distance=6mm}
]
\node[state] (p) {$V_1$};
\node (mid1) [right=of p] {$\cdots$};
\node[state] (r) [right=of mid1] {$V_i$};
\node (mid2) [right=of r] {$\cdots$};
\node[state] (q) [right=of mid2] {$V_j$};
\node (mid3) [right=of q] {$\cdots$};
\node[state] (s) [right=of mid3] {$V_k$};
\coordinate [right=of s] (end);

\path
(p) edge (mid1)
(mid1) edge (r)
(r) edge (mid2)
(mid2) edge (q)
(q) edge (mid3)
(mid3) edge (s)

(p) edge[loop above] 
node[above]
{$\left(0,1,0\right)$}
(p) 
(r)  edge[loop above] 
node[above]
{$\left(2,-3,0\right)$}
(r) 
(q)  edge[loop above] 
node[above]
{$\left(-2,-1,1\right)$}
(q); 
\end{tikzpicture}
\caption{A 3-VASS instance $(V,p(\mathbf{0}),q(\mathbf{0}))$ that is forward pumpable but not forward diagonal.}
\label{fig:3-vass-pumpable-not-diagonal}
\end{figure}

We must explain why a finer classification of pumpability is needed. While $\diagVASSPro$ has the properties we are looking for, diagonality is generally too restrictive for our self-reductions. The broader class $\pumpVASSPro$ still retains useful structure through fireable non-negative cycles; see~\cref{fig:3-vass-pumpable-not-diagonal}. 
\emph{Joint diagonality} plays an important role in the present paper. 
As illustrated in~\cref{fig:3-vass-pumpable-not-diagonal}, the cycle with effect $(0,1,0)$ in $V_1$, $(2,-3,0)$ in $V_i$ and $(-2,-1,1)$ in $V_j$ can jointly provide the effect of diagonality. 
Roughly speaking, by treating the prefix up to $V_j$ as one component, the resulting instance behaves like a diagonal instance. This explains why pumpable subclasses are natural intermediate classes and are crucial for proving the result for $\seqVASSPro$.

We next highlight two key ingredients of the proof.

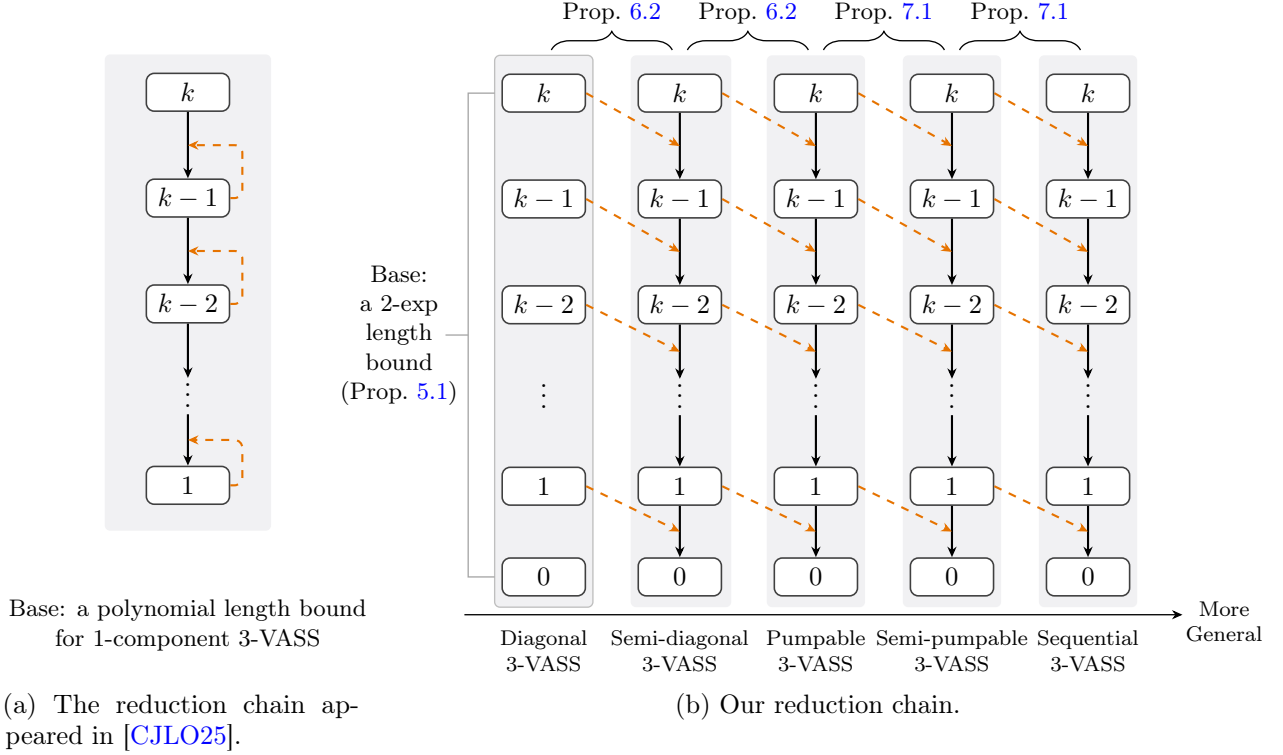
\begin{figure}[t]
    \centering
    \raisebox{-0.41cm} {
    \begin{subfigure}[b]{0.28\textwidth}
        \centering
        \begin{tikzpicture}[
            scale=1,
            transform shape,
            font=\rmfamily\small,
            >={Stealth[length=4pt,width=3.5pt]},
            box/.style={
                draw=black!75,
                rectangle,
                rounded corners=3pt,
                line width=0.6pt,
                minimum width=1.1cm,
                minimum height=0.5cm,
                fill=white,
                inner sep=3pt,
                align=center
            },
            prevbg/.style={fill=blue!6!gray!10}
        ]
            \useasboundingbox (-0.5,-0.9) rectangle (2.2,8.1);

            \fill[prevbg, rounded corners=2pt] (-0.15,1) rectangle (2.05,7.3);

            \node[box] (old-k)    at (0.95,6.8) {$k$};
            \node[box] (old-k1)   at (0.95,5.4) {$k-1$};
            \node[box] (old-k2)   at (0.95,4.0) {$k-2$};
            \node      (old-dots) at (0.95,2.9) {$\vdots$};
            \node[box] (old-1)    at (0.95,1.6) {$1$};
            \node[font=\rmfamily\footnotesize, align=center, text=black] (old-0) at (0.95,-0.25)
             {Base: a polynomial length bound \\for 1-component 3-VASS};

            \draw[->, line width=0.8pt] (old-k) -- (old-k1) coordinate[pos=0.5] (old-mk);
            \draw[->, line width=0.8pt] (old-k1) -- (old-k2) coordinate[pos=0.5] (old-mk1);
            \draw[->, line width=0.8pt] (old-k2) -- ($(old-dots)+(0,0.12)$) coordinate[pos=0.5] (old-mk2);
            \draw[->, line width=0.8pt] ($(old-dots)+(0,-0.36)$) -- (old-1) coordinate[pos=0.5] (old-mdots);

            \coordinate (old-r1) at ($(old-mk)+(0.72,0)$);
            \coordinate (old-s1) at (old-k1.east -| old-r1);
            \draw[->, dashed, orange!90!black, line width=0.8pt, rounded corners=3pt]
                (old-k1.east) -- (old-s1) -- (old-r1) -- (old-mk);

            \coordinate (old-r2) at ($(old-mk1)+(0.72,0)$);
            \coordinate (old-s2) at (old-k2.east -| old-r2);
            \draw[->, dashed, orange!90!black, line width=0.8pt, rounded corners=3pt]
                (old-k2.east) -- (old-s2) -- (old-r2) -- (old-mk1);

            \coordinate (old-r3) at ($(old-mdots)+(0.72,0)$);
            \coordinate (old-s3) at (old-1.east -| old-r3);
            \draw[->, dashed, orange!90!black, line width=0.8pt, rounded corners=3pt]
                (old-1.east) -- (old-s3) -- (old-r3) -- (old-mdots);
        \end{tikzpicture}
        \caption{The reduction chain appeared in \cite{DBLP:conf/icalp/CzerwinskiJ0O25}.}
        \label{subfig:previous_work}
    \end{subfigure}
    }
    \hfill 
    \begin{subfigure}[b]{0.68\textwidth}
        \centering
        \begin{tikzpicture}[
            scale=1,
            transform shape,
            font=\rmfamily\small,
            >={Stealth[length=4pt,width=3.5pt]},
            box/.style={
                draw=black!75,
                rectangle,
                rounded corners=3pt,
                line width=0.6pt,
                minimum width=1.1cm,
                minimum height=0.5cm,
                fill=white,
                inner sep=3pt,
                align=center
            },
            propbrace/.style={
                decorate,
                decoration={brace, amplitude=7pt, raise=2pt},
                line width=0.5pt
            },
            col1/.style={fill=blue!6!gray!10},
            col2/.style={fill=blue!6!gray!10},
            col3/.style={fill=blue!6!gray!10},
            col4/.style={fill=blue!6!gray!10},
            col5/.style={fill=blue!6!gray!10}
        ]
            \useasboundingbox (-0.3,-0.9) rectangle (10.2,8.1);

            \filldraw[col1, draw=gray!50, line width=0.5pt, rounded corners=2pt] (0.65,0) rectangle (1.95,7.3);
            \fill[col2, rounded corners=2pt] (2.45,0) rectangle (3.78,7.3);
            \fill[col3, rounded corners=2pt] (4.25,0) rectangle (5.55,7.3);
            \fill[col4, rounded corners=2pt] (6.05,0) rectangle (7.35,7.3);
            \fill[col5, rounded corners=2pt] (7.85,0) rectangle (9.15,7.3);

            \draw[->, line width=0.5pt] (0.25,-0.1) -- (9.75,-0.1);
            \node[above=2pt, align=center, font=\rmfamily\scriptsize] at (10.3,-0.65) {More\\General};

            \def\xA{1.30}
            \def\xB{3.10}
            \def\xC{4.90}
            \def\xD{6.70}
            \def\xE{8.50}

            \foreach \x/\name in {\xA/A,\xB/B,\xC/C,\xD/D,\xE/E}{
                \node[box] (b-\name-k)  at (\x,6.8) {$k$};
                \node[box] (b-\name-k1) at (\x,5.4) {$k-1$};
                \node[box] (b-\name-k2) at (\x,4.0) {$k-2$};
                \node[box] (b-\name-1)  at (\x,1.6) {$1$};
                \node[box] (b-\name-0)  at (\x,0.4) {$0$};
                \node (b-\name-dots) at (\x,2.9) {$\vdots$};
            }

            \foreach \x/\name in {\xB/B,\xC/C,\xD/D,\xE/E}{
                \coordinate (d-\name-top) at ($(b-\name-dots)+(0,0.12)$);
                \coordinate (d-\name-bot) at ($(b-\name-dots)+(0,-0.36)$);

                \draw[->, line width=0.8pt]
                    (b-\name-k.south) -- coordinate[pos=0.5] (m-\name-k) (b-\name-k1.north);
                \draw[->, line width=0.8pt]
                    (b-\name-k1.south) -- coordinate[pos=0.5] (m-\name-k1) (b-\name-k2.north);
                \draw[->, line width=0.8pt]
                    (b-\name-k2.south) -- coordinate[pos=0.5] (m-\name-k2) (d-\name-top);
                \draw[->, line width=0.8pt]
                    (d-\name-bot) -- coordinate[pos=0.5] (m-\name-dots) (b-\name-1.north);
                \draw[->, line width=0.8pt]
                    (b-\name-1.south) -- coordinate[pos=0.5] (m-\name-1) (b-\name-0.north);
            }

            \foreach \from/\targ in {A/B,B/C,C/D,D/E} {
                \draw[->, dashed, orange!90!black, line width=0.8pt] (b-\from-k.east) -- (m-\targ-k);
                \draw[->, dashed, orange!90!black, line width=0.8pt] (b-\from-k1.east) -- (m-\targ-k1);
                \draw[->, dashed, orange!90!black, line width=0.8pt] (b-\from-k2.east) -- (m-\targ-k2);
                \draw[->, dashed, orange!90!black, line width=0.8pt] (b-\from-1.east) -- (m-\targ-1);
            }

            \node[below=5pt, align=center, font=\rmfamily\scriptsize] at (\xA,0) {Diagonal\\3-VASS};
            \node[below=5pt, align=center, font=\rmfamily\scriptsize] at (\xB,0) {Semi-diagonal\\3-VASS};
            \node[below=5pt, align=center, font=\rmfamily\scriptsize] at (\xC,0) {Pumpable\\3-VASS};
            \node[below=5pt, align=center, font=\rmfamily\scriptsize] at (\xD,0) {Semi-pumpable\\3-VASS};
            \node[below=5pt, align=center, font=\rmfamily\scriptsize] at (\xE,0) {Sequential\\3-VASS};

            \draw[propbrace] ({\xA+0.1},7.3) -- ({\xB-0.1},7.3) node[midway, above=8pt, font=\rmfamily\footnotesize] {Prop.~\ref{prop:2-exp-lcr-for-par-pump-3-vass}};
            \draw[propbrace] ({\xB+0.1},7.3) -- ({\xC-0.1},7.3) node[midway, above=8pt, font=\rmfamily\footnotesize] {Prop.~\ref{prop:2-exp-lcr-for-par-pump-3-vass}};
            \draw[propbrace] ({\xC+0.1},7.3) -- ({\xD-0.1},7.3) node[midway, above=8pt, font=\rmfamily\footnotesize] {Prop.~\ref{prop:2-exp-lcr-for-gen-3-vass}};
            \draw[propbrace] ({\xD+0.1},7.3) -- ({\xE-0.1},7.3) node[midway, above=8pt, font=\rmfamily\footnotesize] {Prop.~\ref{prop:2-exp-lcr-for-gen-3-vass}};

            \draw[draw=gray!70, line width=0.5pt] (0.65,6.8) -- (0.3,6.8) -- (0.3,0.4) -- (0.65,0.4);
            \draw[draw=gray!70, line width=0.5pt] (0.3,3.6) -- (0.0,3.6);
            \node[font=\rmfamily\footnotesize,left, align=center] at (0.3,3.6) {Base: \\a $2$-exp\\length \\bound \\(Prop.~\ref{prop:diagonal-case-lenbound})};
        \end{tikzpicture}
        \caption{Our reduction chain.}
        \label{subfig:our_reduction}
    \end{subfigure}

    \caption{Comparison of two reduction chains.
    }
    \label{fig:overall_reduction_comparison}
\end{figure}

\paragraph*{The Class $\diagVASSPro$.} 
To prove that the inequality~\cref{2026-06-26-airport} is valid for $\diagVASSPro$, confer \cref{prop:diagonal-case-lenbound}, we proceed with a case analysis like in~\cite{DBLP:conf/icalp/CzerwinskiJ0O25}, distinguishing between the \emph{wide} case and the \emph{non-wide} case.
Our treatment of the wide instances is basically the same as in~\cite{DBLP:conf/icalp/CzerwinskiJ0O25}.
Our improvement is for the non-wide instances.
The key insight is that every non-wide instance can be viewed as a sequence of geometrically 2-dimensional 3-VASS. 
For a short explanation, consider a non-wide diagonal instance $(V,s,t)$ with $V = (V_1) u_1 (V_2) u_2 \dots u_{k-1} (V_k) $. 
We will argue that the sequential cone of $(V_1)u_1\dots (V_i)$ and the sequential cone of $\Reverse{((V_i)u_i\cdots (V_k))}$ are separated by a plane, and the reachable configurations should be very close to both sequential cones. 
It can be shown then that every reachable configuration in $V_i$, for every $i\in[k]$, is constrained within a pair of parallel 2-dimensional planes roughly determined by the start and the target configurations.
Therefore, every $V_i$ is transformed simultaneously into a geometrically 2-dimensional VASS of exponential size.
We are done by applying the efficient representation technique developed in~\cite{DBLP:conf/icalp/CzerwinskiJ0O25} and additionally in \cref{sec:repr-2vass}. 
The details are in \cref{sec:DP3VASS}.

\paragraph*{Length Controlled Self-Reduction.}
The proof that~\cref{2026-06-26-airport} is valid for $\pumpVASSPro$ and $\seqVASSPro$, confer \cref{prop:len-par-pump-3-vass} and \cref{prop:3-VASS-lenbound-intro}, makes use of a different technique.
The tool is called \emph{length-controlled self-reduction}. 
\begin{restatable}{definition}{DefinitionLengthControlledReduction}
    \label{def:length-controlled-reduction}
    Let $\mathcal{L} = (L_k)_{k \in \NN}$ be a family of functions from $\NN$ to $\NN$.
    Suppose $\mathcal{V} \subseteq \mathbb{VASS}^d$.
    The class $\mathcal{V}$ is \emph{$\mathcal{L}$-controlled self-reducible} if for every $k \ge 1$ and every $k$-component VASS $(V,s,t) \in \mathcal{V}$, there exists a $k'$-component VASS $(V',s',t')\in \mathcal{V}$ for some $k'\in\{0, \ldots,k - 1\}$ such that the following statements are valid for some $R(x) = x^{\bigO(k^2)}$. 
    \begin{enumerate}
        \item $\Size(V',s',t') \le R\left( L_k(\Size(V,s,t))\right)$, and
        \item\label{def:lcr-3} $\min\Len(V,s,t) \le R(\min \Len(V',s',t') + \Size(V,s,t))$.
    \end{enumerate}
\end{restatable}
In the above definition, $k'$ could be $0$.
A \emph{$0$-component} VASS $V=(Q, T)$ is a VASS with a single state and only non-negative transitions, i.e., $|Q|=1$ and $\Eff{u}\ge \Vec{0}$ for every $u\in T$.
The class of $0$-component VASS is meant to be simple.
\begin{restatable}{observation}{ObsZerodVASSReach}
\label{obs:0-com-vass-len}
Let $(V,s,t)$ be a $0$-component $d$-VASS.
If $s \xrightarrow{*}t$ in $V$, then $s \xrightarrow{\rho} t$ for some $\rho$ with $|\rho|\le d\cdot \norm{t}$.
\end{restatable}
By Condition 2 of \cref{def:length-controlled-reduction}, if $\mathcal{V}$ is $\mathcal{L}$-controlled self-reducible, we can estimate a bound on the length of the shortest runs of every $k$-component VASS in $\mathcal{V}$. 
By Condition 1, we anticipate that there is such a bound $G_k$ that is not very much larger than $L_k$.
It is important to notice that $\min \Len(V',s',t')$ occurs only once in the right-hand side of the inequality of Condition 2, implying that triple-exponential growth is avoided.
Let $\mathcal{G} = (G_k)_{k \in \NN}$.
We say that $\mathcal{V}$ \emph{admits $\mathcal{G}$-short runs} if, for every $k \in \NN$ and every $k$-component VASS $(V,s,t) \in \mathcal{V}$, it holds that 
\begin{equation}\label{2026-07-02-0630}
\min \Len(V,s,t) \le G_k(\Size(V,s,t)).
\end{equation}
Notice that~\cref{2026-07-02-0630} is precisely~\cref{2026-06-26-airport} when $G_k(x)=x^{2^{g(k)}}$ for some polynomial $g$.
When $G_k$ is such a doubly-exponential function, we will use the phrase {\em doubly-exponentially self-reducible} for $\mathcal{L}$-controlled self-reducible, and the phrase {\em doubly-exponentially short run} for $\mathcal{L}$-short run.

Once a length-controlled self-reducibility property of $\mathcal{V}$ is proved, the short run property of $\mathcal{V}$ comes for free. 
This is formally stated in the next lemma, whose proof (see~\cref{Sec-appendix-overview-of-the-algorithm}) is carried out by induction on $k$ and using the basic fact that $0$-component VASS admits short runs.
\begin{restatable}{lemma}{LemmaTwoEXPPreserveViaReduction}
    \label{lem:2-exp-preserve-via-reduction}
    If $\mathcal{V}$ is $(L_k)_{k \in \NN}$-controlled self-reducible, where $L_k(x) = x^{2^{f(k)}}$ for all $k \in \NN$, then $\mathcal{V}$ admits $(G_k)_{k \in \NN}$-short runs, where $G_k(x) = x^{2^{\bigO(f(k)\cdot k)}}$ for all $k \in \NN$.
\end{restatable}

In the light of \cref{lem:2-exp-preserve-via-reduction}, we are left to show that each subclass in~\cref{2026-07-01-0732} enjoys the doubly-exponential self-reducibility property.
This is done for $\diagVASSPro$ in \cref{sec:DP3VASS}, for $\semidiagVASSPro$ and $\pumpVASSPro$ in \cref{sec:par-pump-3-vass}, and for $\semipumpVASSPro$ and $\seqVASSPro$ in \cref{sec:general-case}.
See \cref{prop:2-exp-lcr-for-par-pump-3-vass} and \cref{seqdoubleexpoential}.

\section{Efficient Representation of Reachability Set}
\label{sec:repr-2vass}

Given a 2-VASS $V = (Q, T)$, it is well known \cite{DBLP:journals/tcs/HopcroftP79} that the reachability set $\Reach_q(V, s)$ of any source configuration $s$ and target state $q$ can be expressed as a finite union of \emph{linear sets} $\Vec{b} + P^*$ for some $\Vec{b} \in \NN^2$ and $P \subseteq \NN^2$. Moreover, \cite{DBLP:journals/jacm/BlondinEFGHLMT21} shows that the linear sets are constructed from the solution sets of a finite family of integer linear programs with norm $\PolyF{\Size(V)}$, each with $O(|Q|^2)$ variables and constraints. By \cref{lem:lin-eq-decomp}, every linear-set representation $\Vec{b} + P^*$ satisfies $\norm{\Vec{b}}, \norm{P} \le \PolyF{\Size(V)}^{|Q|^2}$. This bound is exponential and therefore too coarse for our purposes. Recently, \cite{DBLP:conf/icalp/CzerwinskiJ0O25} observed that some finite parts of such representations can be treated as {\em approximately linear sets} while still admitting polynomial-size descriptions. 
These sets are defined in terms of restricted periodic sets.

A \emph{restricted periodic set} is a triple $\langle P, \Vec{c}, K \rangle$ where $P = \{\Vec{p}_1, \dots, \Vec{p}_k\} \subseteq \NN^d \setminus \{\Vec{0}\}$ is a finite set of \emph{periods}, $\Vec{c} \in \NN^{k}$ is a vector, and $K \in \NN$ is a number. We abbreviate such a triple as $\ResPeriodic{P}{\Vec{c}}{K}$ and identify it with the following set of vectors:
\begin{equation*}
    \ResPeriodic{P}{\Vec{c}}{K} := \left\{ \sum_{i = 1}^{k} \Vec{x}(i) \cdot \Vec{p}_i \left|\, \Vec{x} \in \NN^k, \Inner{\Vec{c}, \Vec{x}} \le K \right.\right\}.
\end{equation*}
In particular, the set $\ResPeriodic{\{\Vec{p}\}}{\Vec{1}}{K}$ is abbreviated as $\{\Vec{p}\}^{\le K}$.
We define the \emph{norm} of this representation by $\NORM{\ResPeriodic{P}{\Vec{c}}{K}} := \max\{\norm{P}, \norm{\Vec{c}}\}$. 
Notice that the parameter $K$ does not contribute to the norm. 
Notice also that we do not require $\Vec{c} \in \NN_{>0}^d$ and hence periodic sets $P^*$ are a special kind of restricted periodic set with $\Vec{c} = \Vec{0}$.
A \emph{hybrid set} is a set of the form $S = \Vec{b} + \ResPeriodic{P}{\Vec{c}}{K}$, where $\Vec{b} \in \NN^d$ is the \emph{base}, $\ResPeriodic{P}{\Vec{c}}{K}$ is a restricted periodic set. Its norm is defined to be $\max\{\norm{\Vec{b}}, \NORM{\ResPeriodic{P}{\Vec{c}}{K}}\}$. We write $\NORM{S}$ for the norm of the representation of a hybrid set $S$. 
The following characterization is from~\cite{DBLP:conf/icalp/CzerwinskiJ0O25}. 

\begin{theorem}[{\cite[Lemma 35]{DBLP:conf/icalp/CzerwinskiJ0O25}}]
    \label{thm:2-vass-approximation}
    Let $V$ be a 2-VASS. For any configuration $s$ and state $q$, the reachability set $\Reach_q(V, s)$ is equal to a finite union of hybrid sets $S = \Vec{b} + \ResPeriodic{P}{\Vec{c}}{K}$, where $\NORM{S} \le \PolyF{\Size(V, s)}$.
\end{theorem}

We need to generalize the result of~\cref{thm:2-vass-approximation} in two accounts. 
Firstly, the single point source $s$ is generalized to a \emph{hybrid set} $S$, and secondly, similar to \cite[Lemma 20]{DBLP:conf/icalp/CzerwinskiJ0O25}, the $2$-VASS $V$ is generalized to a \emph{geometrically} 2-dimensional 3-VASS.  \cref{thm:hybrid-2vass-hybrid}  formalizes the generalization.
\begin{restatable}{theorem}{ThmHybridTwoVASSHybrid}
    \label{thm:hybrid-2vass-hybrid}
    Let $V$ be a geometrically 2-dimensional 3-VASS, $S = \Vec{s} + \ResPeriodic{P}{\Vec{c}}{K}$ be a hybrid set satisfying $P\subseteq \CycleSpace(V)$. 
    For all states $p,q$, the set $\ReachqVs{q}{V}{p(S)}$ is a finite union of hybrid sets $S' = \Vec{b} + \ResPeriodic{P'}{\Vec{c}'}{K'}$ with $\NORM{S'} \le \PolyF{\Size(V), \NORM{S}}$. Moreover, $P' \subseteq \Cone(P) + \Cone(V)$.
\end{restatable}
\noindent
The proof of~\cref{thm:hybrid-2vass-hybrid} is placed in~\cref{sec:appendix-repr-2vass}.
Most of the arguments in the proof follow the related arguments in~\cite[Appendix B]{DBLP:conf/icalp/CzerwinskiJ0O25} closely. 

Now we have represented reachability sets in geometrically 2-dimensional 3-VASS when the source belongs to a hybrid set. For our purpose, we need to know more about runs starting from a hybrid set.
The following lemma establishes a bound on the length and the norm of the source of such runs in a geometrically 2-dimensional 3-VASS.

\begin{lemma}
    \label{lem:len-target-in-hybrid-set}
    Let $V$ be a geometrically 2-dimensional 3-VASS, $t$ be a configuration and $S$ be a hybrid set. 
    Then for every state $p$, if $p(\Vec{x}) \xrightarrow{*} t$ in $V$ for some $\Vec{x} \in S$, then there is also a run $p(\Vec{x}') \xrightarrow{\rho} t$ in $V$ for some $\Vec{x}' \in S$ with
    \begin{equation*}
        \norm{\Vec{x}'}, |\rho| \le \PolyF{\Size(V,t),\NORM{S}}.
    \end{equation*}
\end{lemma}
To derive the length bound, we rely on a recent result about geometrically $2$-dimensional VASS from Chen et al.~\cite{chen2026improvingreachabilityvectoraddition}.

\begin{theorem}[{\cite[Corollary 6.22]{chen2026improvingreachabilityvectoraddition}}]
    \label{thm:blondin}
    There exists a polynomial $H$ such that, for every geometrically 2-dimensional VASS $(V,s,t)$, if $s\xrightarrow{*}t$, then $s \xrightarrow{\pi} t$ for some run $\pi$ with $|\pi| \le H(\Size(V,s,t))$.
\end{theorem}

\begin{proof}[Proof of~\cref{lem:len-target-in-hybrid-set}]
    Equivalently, we have $t \xrightarrow{*} p(\Vec{x})$ in $\Reverse{V}$.
    By~\cref{thm:hybrid-2vass-hybrid}, there exists a hybrid set $S'$ such that $\Vec{x} \in S' \subseteq \ReachqVs{p}{\Reverse{V}}{t}$ and $\NORM{S'} \le \PolyF{\Size(V,t)}$.
    Therefore, one has that $\Vec{x} \in S \cap S' \neq \emptyset$.
    Assume that 
    \begin{equation*}
        S = \Vec{b} + \ResPeriodic{P}{\Vec{c}}{K}\text{ and }S' = \Vec{b}' + \ResPeriodic{P'}{\Vec{c}'}{K'}.
    \end{equation*}
    We may interpret each periodic set as a matrix whose columns are the vectors contained in the set, e.g., $P' \in \NN^{d \times |P'|}$.
    Then the intersection of these two sets can be described as the following linear system:
    \begin{equation}
        \label{eq:intersect-of-hybrid-set}
        \begin{aligned}
            \begin{pmatrix}
                P &-P'
            \end{pmatrix}
            \Vec{\lambda} &= \Vec{b}' - \Vec{b}\\
            \begin{pmatrix}
                \Vec{c}^{\dag} & \Vec{0}\\
                \Vec{0} & (\Vec{c}')^{\dag}
            \end{pmatrix}\Vec{\lambda} &\le 
            \begin{pmatrix}
                K\\
                K'
            \end{pmatrix},
        \end{aligned}
    \end{equation}
    where $\Vec{\lambda} \in \NN^{|P|+|P'|}$, and the dimension $|P|+|P'|$ is bounded by 
    \begin{equation*}
        (\NORM{S}+1)^3+(\NORM{S'}+1)^3 \le \PolyF{\Size(V,t), \NORM{S}}.
    \end{equation*}
    Note that $\Vec{x}$ induces a solution $\Vec{\lambda}_{\Vec{x}}$ to this system.
    We omit the inequalities in \cref{eq:intersect-of-hybrid-set} and apply~\cref{lem:pottier} to obtain a solution $\Vec{\lambda}_0 \le \Vec{\lambda}_{\Vec{x}}$ with 
    \begin{equation*}
        \begin{pmatrix}
            P &-P'
        \end{pmatrix}
        \Vec{\lambda}_0 = \Vec{b}' - \Vec{b},
    \end{equation*}
    and $\norm{\Vec{\lambda}_0} \le \PolyF{\Size(V,t), \NORM{S}}$.
    However, since $\Vec{c},\Vec{c}' \ge \Vec{0}$ and $\Vec{\lambda}_0 \le \Vec{\lambda}_{\Vec{x}}$, one has that
    \begin{equation*}
        \begin{pmatrix}
            \Vec{c}^{\dag} & \Vec{0} \\
            \Vec{0} & (\Vec{c}')^{\dag}
        \end{pmatrix}\Vec{\lambda}_0 \le 
        \begin{pmatrix}
            \Vec{c}^{\dag} & \Vec{0}\\
            \Vec{0} & (\Vec{c}')^{\dag}
        \end{pmatrix}\Vec{\lambda}_{\Vec{x}} \le
        \begin{pmatrix}
            K\\
            K'
        \end{pmatrix},
    \end{equation*}
    which implies that $\Vec{\lambda}_0$ is also a solution to~\cref{eq:intersect-of-hybrid-set}.
    Let 
    \begin{equation*}
        \Vec{x}':=\Vec{b} + 
        \begin{pmatrix}
            P &\Vec{0}
        \end{pmatrix}
        \Vec{\lambda}_0.
    \end{equation*}
    It holds that $\Vec{x}' \in S \cap S'$ and $\norm{\Vec{x}'} \le \PolyF{\Size(V,t),\NORM{S}}$.
    Since $\Vec{x}' \in S'$, we obtain that $t\xrightarrow{*}p(\Vec{x}')$ in $\Reverse{V}$, which implies that $p(\Vec{x}')\xrightarrow{*}t$ in $V$.
    By applying~\cref{thm:blondin}, there exists a short run $\rho$ with 
    \begin{equation*}
        |\rho| \le H(\Size(V,t) + 3\cdot\norm{\Vec{x}'})\le \PolyF{\Size(V,t), \NORM{S}}.
    \end{equation*}
We are done. 
\end{proof}

\section{Diagonal 3-VASS}\label{sec:DP3VASS}

In this section, we focus on bounding the length of the shortest runs for diagonal 3-VASS instances, i.e., $(V, s, t)\in\diagVASSPro$. One may expect that the reachability from $s$ to $t$ is equivalent to the $\ZZ$-reachability, as every $\ZZ$-run from $s$ to $t$ may be lifted to a valid run using the pumping cycles. However, the effect of the pumping cycle on $s$ may differ from that of the dual pumping cycle on $t$, and it might not be the case that one can compensate for their difference by adjusting the $\ZZ$-run. Previous work \cite{DBLP:conf/icalp/CzerwinskiJ0O25} introduced the notion of \emph{wideness} to identify whether such a compensation can be achieved. Roughly speaking, $V$ is wide if every positive vector is contained in $\SeqCone(V)$. They showed that a wide instance in $\diagVASSPro$ allows the lifting of $\ZZ$-runs. The drawback lies in the handling of non-wide instances. In \cite{DBLP:conf/icalp/CzerwinskiJ0O25}, a non-wide sequential VASS is reduced to a sequential VASS of fewer components. Then, a length bound of the shortest runs in the original VASS is obtained by induction, which yields merely a triply-exponential bound. 

For diagonal 3-VASS, we improve the induction by introducing a structured preprocessing step that, in the non-wide case, transforms all components into geometrically 2-dimensional ones.
A doubly-exponential bound on the length of the shortest runs is then derived using the linear Diophantine system of geometrically 2-dimensional VASS~\cite{DBLP:conf/icalp/FuYZ24}. As a result, we prove the following proposition for $\diagVASSPro$.

\begin{restatable}{proposition}{propDiagonalCaseLenbound}
\label{prop:diagonal-case-lenbound}
If $(V, s, t)\in \diagVASSPro$ has $k$ components, $\min\Len(V,s,t) \leq \Size(V, s, t)^{2^\PolyF{k}}$.
\end{restatable}

We distinguish between the wide and the non-wide cases. Recall that $V$  is wide if $\QQ_{\ge 0}^{3}\subseteq{}\SeqCone(V)$ or $\QQ_{\ge 0}^{3}\subseteq{}\SeqCone(\Reverse{V})$; otherwise, it is non-wide. The length bound of the shortest runs for the wide case is provided by~\cref{prop:diagonal-wide-case-lenbound}, whereas that for the non-wide case is established in~\cref{prop:diagonal-nonwide-case-lenbound}. All omitted proofs are given in~\cref{sec:appendix-DP3VASS}.

\subsection{Wideness}\label{sec:wideness}

The wide case is relatively easy, as it can be reduced to an already established result. Note that here we use a more relaxed definition than the one adopted in \cite{DBLP:conf/icalp/CzerwinskiJ0O25}. 

\begin{proposition}
\label{prop:diagonal-wide-case-lenbound}
    Suppose $(V, s, t)\in \diagVASSPro$, $V$ has $k$ components and $V$ is wide. Then $\min\Len(V,s,t) \le \PolyF{\Size(V, s, t)}^{k}$.
\end{proposition}

A result analogous to \cref{prop:diagonal-wide-case-lenbound} is presented in~\cite[Lemma 24]{DBLP:conf/icalp/CzerwinskiJ0O25}. 
The difference is to do with the definition of wideness.
Our definition is slightly more permissive than that of \cite{DBLP:conf/icalp/CzerwinskiJ0O25}. 
Nevertheless, the two formulations are in fact equivalent in the presence of diagonality.
We recall the definition of wideness in \cite{DBLP:conf/icalp/CzerwinskiJ0O25}. Let $X = \{\Vec{v}_1, \Vec{v}_2, \ldots, \Vec{v}_k\} \subseteq \QQ^d$ be a finite set of vectors. The \emph{open cone} generated by $X$ is 
\begin{equation*}
    \Cone_{>0}(X) := \left\{ \sum_{i \in [k]} \lambda_i \Vec{v}_i \mid \lambda_i \in \QQ_{>0} \text{ for all } i \in [k] \right\}.
\end{equation*}
Let $C_1, \ldots, C_k \subseteq \QQ^3$ be a sequence of (finitely generated) open cones, then their \emph{open sequential cone} is 
\begin{equation*}
    \SeqCone_{>0}(C_1, \ldots, C_k) := \left\{ \sum_{i \in [k]} \Vec{v}_i \mid \Vec{v}_i \in C_i, \sum_{j = 1}^i \Vec{v}_i \in \QQ_{>0}^d \text{ for all } i \in [k] \right\}.
\end{equation*}
Accordingly, we define the open cone and the open sequential cone of a sequential VASS. A sequential $3$-VASS $V = (V_1)u_1(V_2)\ldots u_{k-1}(V_k)$ is said to be \emph{openly wide} if either $\QQ_{>0}^3 \subseteq \SeqCone_{>0}(V)$ or $\QQ_{>0}^3 \subseteq \SeqCone_{>0}(\Reverse{V})$. 
We are now able to invoke the following result.
\begin{lemma}[{\cite[Lemma 24]{DBLP:conf/icalp/CzerwinskiJ0O25}}]
    Suppose $(V, s, t)\in \diagVASSPro$, $V$ has $k$ components and $V$ is openly wide. Then $\min\Len(V,s,t) \le \PolyF{\Size(V, s, t)}^{k}$.
\end{lemma}
\noindent We are done by the following observation.
\begin{restatable}{lemma}{LemWideIFFOpenWide}
    \label{lem:wide-iff-open-wide}
    Suppose $(V, s, t)\in\diagVASSPro$. Then $V$ is wide if and only if $V$ is openly wide.
\end{restatable}

\subsection{Non-Wideness}\label{sec:non-wideness}

In this subsection, we consider the diagonal but non-wide case, which is more intricate. The main result is the following proposition.
\begin{proposition}
    \label{prop:diagonal-nonwide-case-lenbound}
    Suppose $(V, s, t)\in \diagVASSPro$, $V$ is non-wide and has $k$ components. Then $\min\Len(V,s,t)\le \Size(V, s, t)^{2^{\PolyF{k}}}$.
\end{proposition}
We fix $V = (V_1)u_1\ldots u_{k-1}(V_k)$. 
We first lay down the necessary technical groundwork, and then establish the doubly-exponential bound.

\paragraph{Separating Sequential Cones at a Component}\label{sec:non-wide-separatingPlane}
Here we generalize the following properties to $\diagVASSProd$. First, we show the existence of separating hyperplanes of the sequential cones of $V$ and $\Reverse{V}$. 
For that purpose we split the sequential VASS $V= (V_1)u_1\ldots u_{k-1}(V_k)$ at a component $V_i$, and denote the head by $\mathcal{H}_i$ and the tail by $\mathcal{T}_i$. 
Formally,
\begin{align*}
\mathcal{H}_i := &(V_1)u_1\ldots u_{i-1}(V_i),\\
\mathcal{T}_i := &(V_i)u_i\ldots u_{k-1}(V_k). 
\end{align*}
Let $C \subseteq \QQ^d$ be a set of vectors, and $\Vec{n} \in \QQ^d$ be a vector. We write for example $\Inner{\Vec{n}, C} \ge 0$ to denote that $\Inner{\Vec{n}, \Vec{c}} \ge 0$ for all $\Vec{c} \in C$. The following lemma shows that in the non-wide case, $\SeqCone(\mathcal{H}_i)$ and $\SeqCone(\Reverse{\mathcal{T}_i})$ can be separated by some vector in $\ZZ^d$. 

\begin{restatable}{lemma}{LemSepratingPlaneNonWide}
\label{lem:separating-plane-non-wide}
Suppose $(V,s,t)\in \diagVASSProd$ is non-wide and has $k$ components.
For each $i\in[k]$, there exists a non-zero vector $\Vec{n} \in \ZZ^d$ such that $\Inner{\Vec{n}, \SeqCone(\mathcal{H}_i)} \ge 0$ and $\Inner{\Vec{n}, \SeqCone(\Reverse{\mathcal{T}_i})} \le 0$. 
Moreover, $\norm{\Vec{n}} \le \PolyF{\Size(V)}^{k^3d^{5}}$.
\end{restatable}

\paragraph{Bounded Distance from Sequential Cone}\label{sec:non-wide-boundedDist}

Next, we show that any configuration reachable from $s$ has bounded distance from the sequential cone of $V$.

\begin{restatable}{lemma}{LemmaReachEqSeqConePlusBoundedOfDiag}
    \label{lem:reach-eq-seq-cone-plus-bounded-of-diag}
    Suppose $(V,s,t)\in \diagVASSProd$ is non-wide and has $k$ components. Let $s = p(\Vec{u})$. For any configuration $t = q(\Vec{v})$ in $V_k$ that is reachable from $s$, we have $\Vec{v} \in \SeqCone(V) + \Vec{h}$ and $\norm{\Vec{h}} \le \PolyF{\Size(V)}^{d+1}\cdot\norm{\Vec{u}}$ for some $\Vec{h} \in \ZZ^d$.
\end{restatable}

It should be pointed out that a similar result already appears in \cite{DBLP:conf/icalp/CzerwinskiJ0O25}. 
Their proof is built on a slightly different definition of sequential cones based on open cones, and the closeness is measured by Euclidean distance. 

\paragraph{Inner Products Encoded in States}\label{sec:non-wide-dimRed}

Consider a run $s \xrightarrow{\pi} t$ in $V$. Let $c_i$ be a configuration along $\pi$ whose state belongs to $V_i$. By \cref{lem:reach-eq-seq-cone-plus-bounded-of-diag} we know that the counter value of $c_i$ cannot be too far from both $\SeqCone(\mathcal{H}_i)$ and $\SeqCone(\Reverse{\mathcal{T}_i})$. Moreover, these two sequential cones are separated by a plane by \cref{lem:separating-plane-non-wide}. We deduce that $c_i$ must be close to this plane.  
We now formally justify this intuition. 
Let $\Vec{n}_i \in \ZZ^3$ be a normal vector given by \cref{lem:separating-plane-non-wide} satisfying $\norm{\Vec{n}_i} \le \PolyF{\Size(V)}^{k^3}$.

\begin{restatable}{lemma}{LemmaNonWideBoundedInnerProd}
\label{lem:non-wide-bounded-inner-prod}
Suppose $(V,s,t)\in\diagVASSPro$ is non-wide and has $k$ components. Let $c = p(\Vec{z})$ be a configuration in $V_i$ such that $s \xrightarrow{*} c \xrightarrow{*}t$. Then $|\Inner{\Vec{n}_i, \Vec{z}}| \le B:= \PolyF{\Size(V)}^{k^3} \cdot (\norm{s} + \norm{t})$.
\end{restatable}

Now we can encode the inner product with $\Vec{n}_i$ into states in each component $V_i$. The reachability set of each component can be expressed as a union of hybrid sets, one by one, thanks to \cref{thm:hybrid-2vass-hybrid}. To be explicit, recall that we are working with $V = (V_1)u_1\ldots u_{k-1}(V_k)$. For consistency, we assume that there is a dummy bridge $u_k$ after the last component $V_k$ before the state of $t$ whose effect is zero.
We have for each component $V_i$ a vector $\Vec{n}_i$ given by \cref{lem:separating-plane-non-wide}. 
Let $B_n$ be the bound for all $\norm{\Vec{n}_i}$. Then $B_n \le \PolyF{\Size(V)}^{k^3}$. 
Also, let $B_h$ be the bound given by \cref{lem:non-wide-bounded-inner-prod}. We have $B_h \le \PolyF{\Size(V)}^{k^3}\cdot (\norm{s} + \norm{t})$.
For each component $V_i = (Q_i, T_i)$ and the bridge $u_i = r \xrightarrow{\Vec{a}_i} r'$, 
we shall define a geometrically 2-dimensional $3$-VASS $V_i' = (Q_i', T_i')$ where $Q_i' = \{q_v \mid q\in Q_i, v \in [-B_h, B_h]\} \cup \{r'\}$, 
and $T_i$ contains all the transitions of the form $p_u \xrightarrow{\Vec{a}} q_v$ satisfying $p\xrightarrow{\Vec{a}}q \in T_i$ and $\Inner{\Vec{n}_i, \Vec{a}} = v - u$; 
moreover, we add transitions $r_v \xrightarrow{\Vec{a}_i} r'$ for all $v \in [-B_h, B_h]$. 
By \cref{lem:non-wide-bounded-inner-prod}, every configuration $p(\Vec{x})$ in $V_i$ on some run from $s$ to $t$ is mapped to a configuration $p_u(\Vec{x})$ where $u = \Inner{\Vec{n}_i, \Vec{x}} \in [-B_h, B_h]$. 
Notice that $\Size(V_i') \le \bigO(B_h) \cdot \Size(V_i) \le \PolyF{\Size(V)}^{k^3}\cdot (\norm{s} + \norm{t})$. We remark that each $V_i'$ is constructed independently and should not be treated as a component of a new VASS.

In the following, we will refer to the polynomials named below.
\begin{itemize}
    \item $L(x, y)$ from \cref{lem:len-target-in-hybrid-set} for bounding the length of a run in a geometrically 2-dimensional 3-VASS whose source is a hybrid set and whose target is a fixed configuration;
    \item $H(x,y)$ from \cref{thm:hybrid-2vass-hybrid} for bounding the size of the representation of a hybrid set of a reachability set of geometrically 2-dimensional 3-VASS.
\end{itemize}
For example, the conclusion of~\cref{lem:len-target-in-hybrid-set} is now written as ``$\norm{\Vec{x}'}, |\rho| \le L({\Size(V,t),\NORM{S}})$''. We then define a length function by induction (from $k$ to $1$).
\begin{eqnarray*}
    f_{k}(x) &:=& L\left(\Size(V_{k}', t), x\right), \\
    f_{j}(x) &:=& 2 L\left(\Size(V_j') + f_{j+1}\left(C \cdot H(\Size(V_j'), x)^2\right), x\right), \text{ for } 0< j < k,
\end{eqnarray*}
where $C := B_h +3B_n+ 1$.
Intuitively, for each $j\in[k]$, $f_j$ represents an upper bound on the minimal distance to the target $t$ when starting from some hybrid set in component $j$. The variable $x$ represents the norm of the set. We show that these functions $f_j$ are doubly-exponentially bounded. 
The bounding function $f_1(x)$ is what we are looking for. The following lemma provides an upper bound for $f_1(x)$.

\begin{restatable}{lemma}{LemmaBoundFone}
\label{lem:bound-f1}
    $f_1(x) \le (\Size(V, s, t) + x)^{2^{\PolyF{k}}}$.
\end{restatable}

We now have all the tools to prove \cref{prop:diagonal-nonwide-case-lenbound}.
The strategy of bounding $\Len(V, s, t)$ is to perform an induction over the components of $V$.
Starting from $s$, we may split a run to $t$ into a prefix $\sigma$ in the first component and the first bridge $(V_1)u_1$, followed by a suffix $\rho$ in $\mathcal{T}_2 = (V_2)u_2\ldots u_{k-1}(V_k)$. The target of $\sigma$, which is also the source of $\rho$, falls in some reachability set in the geometrically 2-dimensional 3-VASS $V_1'$, which we can efficiently represent as a hybrid set $S$ of norm $H(\Size(V_1'), \norm{s})$. We then handle the suffix $\rho$ inductively with is source $S$ in mind, deriving both a bound on its length and a bound on the norm of its source, i.e.\ the target of $\sigma$. After this, the length of $\sigma$ can be bounded in terms of the polynomial $L(x, y)$ using \autoref{lem:len-target-in-hybrid-set}. To realize this inductive argument, we prove the following lemma that strengthens \cref{prop:diagonal-nonwide-case-lenbound} in a way that is favorable for induction.
\cref{prop:diagonal-nonwide-case-lenbound} then follows immediately by applying~\cref{lem:bound-f1} with $x = \norm{{s}}$.

\begin{lemma}
\label{lem:non-wide-inductive}
Suppose $V\in\diagVASSPro$ is non-wide and has $k$ components.
    Let $j \in [k]$, $p$ be a state in $V_j$, $S = \Vec{b} + \ResPeriodic{P}{\Vec{c}}{K}$ be a hybrid set such that $\Inner{\Vec{n}_j, P} = 0$, $P \subseteq \SeqCone(\mathcal{H}_j)$ and $S \subseteq \ReachqVs{p}{\mathcal{H}_j}{s}$. If $p(S) \xrightarrow{*} t$ in $\mathcal{T}_j$, then there exists $\Vec{s} \in S$ and a run $\pi$ with $p(\Vec{s}) \xrightarrow{\pi} t$ and $\norm{\Vec{s}}, |\pi| \le f_j(\NORM{S})$.
\end{lemma}

\begin{proof}
    Notice that any run from $p(S)$ to $t$ must use the bridge $u_j$. Let $u_j$ be the transition $r\xrightarrow{\Vec{a}}r'$ where $r$ is a state in $V_j$ and $r'$ is a state in $V_{j+1}$ (or $r'$ is the state of $t$ when $j = k$).
    As configurations in $p(S)$ is reachable from $s$, and $t$ is reachable from $p(S)$, by \cref{lem:non-wide-bounded-inner-prod}  our construction of $V_j'$ guarantees that $t$ is also reachable (in the original VASS) from $r'(\ReachqVs{r'}{V_j'}{p_u(S)})$ where $u = \Inner{\Vec{n}_j, \Vec{b}}$. 
    To justify this, consider any run $p(\Vec{s}) \xrightarrow{\pi_j} r'(\Vec{s}') \xrightarrow{*} t$ where $\Vec{s} \in S$. By \cref{lem:non-wide-bounded-inner-prod} every configuration $q(\Vec{z})$ on $\pi_j$ satisfies $|\Inner{\Vec{n}_j, \Vec{z}}| \le B_h$. Hence $\pi_j$ can be easily projected into a run in $V_j'$ from $p_u(\Vec{s})$ to $r'(\Vec{s}')$.
    Hence, the set $r'(\ReachqVs{r'}{V_j'}{p_u(S)})$ must contain all configurations that occurs on some run from $p(S)$ to $t$ after the bridge $u_j$.
    Notice that $\CycleSpace(V_j')$ is orthogonal to $\Vec{n}_j$. We may further assume that $\CycleSpace(V_j')$ is the whole plane orthogonal to $\Vec{n}_j$ by adding isolated cycles to $V_j'$ if necessary, hence the condition $\Inner{\Vec{n}_j, P} = 0$ guarantees $P\subseteq \CycleSpace(V_j')$.

    We prove the lemma by backward induction on $j$. The base case $j = k$ is exactly \cref{lem:len-target-in-hybrid-set} applied to $V_k'$. Now we proceed with $j < k$. By \cref{thm:hybrid-2vass-hybrid}, $\ReachqVs{r'}{V_j'}{p_u(S)}$ is a finite union of hybrid sets $S' = \Vec{b}' + \ResPeriodic{P'}{\Vec{c}'}{K'}$ with $\NORM{S'} \le H(\Size(V_j'), \NORM{S})$.
    In the following, we let $\overline{V_j}$ be the VASS obtained from $V_j$ by adding the bridge transition $u_j$ from $r$ to $r'$, as we need to consider the sub-runs before $\mathcal{T}_{j+1}$.
    By the above argument, for one of such set $S'$ we have $p(S) \xrightarrow{*} r'(S')$ in $\overline{V_j}$ and $r'(S') \xrightarrow{*} t$ in $\mathcal{T}_{j+1}$. We would like to apply the induction hypothesis to $r'(S')$ in $\mathcal{T}_{j+1}$.
    We are however not guaranteed that $\Inner{\Vec{n}_{j+1}, P'} = 0$, and it is not straightforward to see $P' \subseteq \SeqCone(\mathcal{H}_{j+1})$. To meet these properties, recall that by \cref{thm:hybrid-2vass-hybrid} we have $P' \subseteq\Cone(V_j') + \Cone(P) \subseteq \Cone(V_j) + \Cone(P)$. Together with the condition $P \subseteq \SeqCone(\mathcal{H}_j)$, we have
    \begin{equation*}
        P' \subseteq(\Cone(P) + \Cone(V_j) ) \cap \NN^3 \subseteq (\SeqCone(\mathcal{H}_j) + \Cone(V_{j+1})) \cap \NN^3 = \SeqCone(\mathcal{H}_{j+1}).
    \end{equation*}
    Then by \cref{lem:separating-plane-non-wide} we have $\Inner{\Vec{n}_{j+1}, P'} \ge 0$. Extract from $P'$ the set of vectors $P_+'$ with positive inner products with $\Vec{n}_{j+1}$, we claim that it is enough to add these vectors for a bounded number of times to produce a vector $\Vec{s}'$ with $p(S) \xrightarrow{*} r'(\Vec{s}') \xrightarrow{*} t$, since vectors reachable at state $r'$ have bounded inner products with $\Vec{n}_{j+1}$ by \cref{lem:non-wide-bounded-inner-prod}. 
    To be specific, let $\Vec{s} \in S$ and $\Vec{s}' \in S'$ render true $p(\Vec{s}) \xrightarrow{*} r'(\Vec{s}') \xrightarrow{*} t$. Then $\Vec{s}' = \Vec{b}' + P'_+ \Vec{c}_+ + P'_0\Vec{c}_0$ for some coefficients $\Vec{c}_+$ and $\Vec{c}_0$, where $P'_0 = P' \setminus P_+'$ contains the vectors in $P'$ orthogonal to $\Vec{n}_{j+1}$. By \cref{lem:non-wide-bounded-inner-prod}, we have $|\Inner{\Vec{n}_{j+1}, \Vec{s}'}| = |\Inner{\Vec{n}_{j+1}, \Vec{b}' + P'_+\Vec{c}_+}| \le B_h$. 
    Hence $|\Inner{\Vec{n}_{j+1}, P_+'\Vec{c}_+}| \le B_h + |\Inner{\Vec{n}_{j+1}, \Vec{b}'}| \le B_h + 3B_n\cdot \NORM{S'}$.
    Clearly the norm of $\Vec{c}_+$ is bounded by $\norm{\Vec{c}_+}_1 \le B_h + 3B_n\cdot \NORM{S'}$. In other words, $\Vec{p}' := P_+'\Vec{c}_+$ is a sum of at most $B_h + 3B_n\cdot\NORM{S'}$ vectors from $P_+'$. We have $\norm{\Vec{p}'} \le (B_h + 3B_n)\cdot\NORM{S'}^2$.
    
    We then replace $S'$ by $S'' = \Vec{b}' + \Vec{p}' + \ResPeriodic{P_0'}{\Vec{c}''}{K''}$, where $\Vec{c}''$ is obtained by projecting $\Vec{c}'$ onto indices corresponding to $P_0'$, and $K''$ is obtained by subtracting from $K'$ the inner product induced by $\Vec{p}'$.
    It still holds that $r'(S'') \xrightarrow{*} t$ in $\mathcal{T}_{j+1}$. Moreover, $\NORM{S''} \le (B_h+3B_n + 1)\cdot \NORM{S'}^2$. Now we apply the induction hypothesis, and deduce that there exists $\Vec{s}'' \in S''$ and a run $r'(\Vec{s}'') \xrightarrow{\pi''} t$ in $\mathcal{T}_{k+1}$ such that 
    \begin{equation*}
    \begin{aligned}
        \norm{\Vec{s}''}, |\pi''| &\le f_{j+1}(\NORM{S''})\\
        &\le f_{j+1}((B_h + 3B_n + 1)\cdot H(\Size(V_j'), \NORM{S})^2)\\
        &= f_{j+1}(C\cdot H(\Size(V_j'), \NORM{S})^2).
    \end{aligned}
    \end{equation*}
    \noindent
    Moreover, as $\Vec{s}'' \in S'' \subseteq S' \subseteq \ReachqVs{r'}{V_j'}{p_u(S)} \subseteq \ReachqVs{r'}{\overline{V_j}}{p(S)}$, there exists $\Vec{s} \in S$ and a run 
    $\pi: p(\Vec{s}) \xrightarrow{\pi_1} r'(\Vec{s}'') \xrightarrow{\pi''} t$
    in $\mathcal{T}_j$. 
    By applying \cref{lem:len-target-in-hybrid-set} to $S$ and $p(\Vec{s}) \xrightarrow{*} r'(\Vec{s}'')$, we may assume that $\norm{\Vec{s}}$ and $|\pi_1|$ are bounded by 
    \begin{equation*}
    \begin{aligned}
        \norm{\Vec{s}},|\pi_1| &\leq L(\Size(V_j',\Vec{s}''), \NORM{S})\\
            &\leq L(\Size(V_j')+f_{j+1}(C\cdot H(\Size(V_j'), \NORM{S})^2).
    \end{aligned}
    \end{equation*}
     Therefore, $|\pi| = |\pi_1| + |\pi''| \le 2L(\Size(V_j')+f_{j+1}(C\cdot H(\Size(V_j'), \NORM{S})^2) = f_j(\NORM{S})$.
\end{proof}

Now we are done with~\cref{prop:diagonal-case-lenbound} by~\cref{prop:diagonal-wide-case-lenbound} and~\cref{prop:diagonal-nonwide-case-lenbound}.

\section{Pumpable 3-VASS}
\label{sec:par-pump-3-vass}

\Cref{sec:DP3VASS} is about diagonal 3-VASS.
In this section, we investigate upper bounds on the length of the shortest runs admitted by a pumpable 3‑VASS $(V, s, t)$, as shown in the following proposition.

\begin{restatable}{proposition}{lenParPumpThreeVASS}
    \label{prop:len-par-pump-3-vass}
    If $(V, s, t) \in \pumpVASSPro$ has $k$-components, $\min\Len(V,s,t) \le \Size(V, s, t)^{2^\PolyF{k}}$.
\end{restatable}

We first briefly recall and elaborate on the proof strategy outlined in~\cref{sec:overview} for this section.
By definition, in a pumpable VASS $(V,s,t)$, at least one coordinate along a given run can be increased to an arbitrarily large value.
One may therefore expect that, once this ``free'' coordinate is available, the behavior of $(V,s,t)$ closely resembles that of a 2-VASS.
Consequently, specialized pumping techniques for 2-VASS~\cite{CzerwinskiLLP19} (see~\cref{lem:par-pump-or-bounded-in-2-vass}) can be applied to identify additional pumping cycles, thereby extracting a jointly diagonal prefix from the original pumpable run.
By then efficiently representing the reachability set of this jointly diagonal prefix as a hybrid set, we obtain a semi-diagonal VASS instance.
This yields a self-reduction for $\pumpVASSPro$, controlled by the length function of $\semidiagVASSPro$.
A similar argument also applies to the class $\semidiagVASSPro$.
These two self-reductions are formalized in~\cref{prop:2-exp-lcr-for-par-pump-3-vass}.

\begin{proposition}
    \label{prop:2-exp-lcr-for-par-pump-3-vass}
(i) $\semidiagVASSPro$ is doubly-exponentially self-reducible, and admits doubly-exponentially short runs.
(ii) $\pumpVASSPro$ is doubly-exponentially self-reducible, and admits doubly-exponentially short runs.
\end{proposition}

Combined with \cref{lem:2-exp-preserve-via-reduction}, we get the length bound for $\pumpVASSPro$.

\begin{proof}[Proof of~\cref{prop:len-par-pump-3-vass}]
    Let $(V, s, t) \in \pumpVASSPro$ be a $k$-component 3-VASS.
    By~\cref{prop:2-exp-lcr-for-par-pump-3-vass} and \cref{lem:2-exp-preserve-via-reduction}, it holds that $\min \Len (V,s,t) \le \Size(V, s, t)^{2^\PolyF{k}}$.
\end{proof}

To establish~\cref{prop:2-exp-lcr-for-par-pump-3-vass}, we begin by formally introducing the notion of joint diagonality.

\paragraph{Joint Diagonality.}
Let $V = (Q, T)$ be a $d$-VASS. 
Alternatively, a sequence of $d$ cycles $\theta_1, \theta_2, \dots, \theta_d$ is jointly diagonal if the sequential cone generated by their effects contains a positive vector, i.e., 
\begin{equation}
    \SeqCone(\theta_1, \dots, \theta_{d}) \cap \NN_{>0}^d \neq \emptyset,
\end{equation}
where $\SeqCone(\theta_1, \dots, \theta_{d})$ is a short-hand for $\SeqCone(\Cone(\{\Eff{\theta_1}\}), \dots, \Cone(\{\Eff{\theta_d\}}))$.
The $d$ cycles need not be distinct.

We shall impose the \emph{non-redundancy} condition that, if $\Eff{\theta_i}(\iota) < 0$ for some $i, \iota \in [d]$, there must exist some $j < i$ such that $\Eff{\theta_j}(\iota) > 0$. Otherwise, the cycle $\theta_i$ could be removed and replaced by a repetition of a preceding cycle without changing the sequential cone, while maintaining a total of $d$ cycles.
In what follows, we assume this condition by default.
We say that a sequence of cycles $\theta_1, \dots, \theta_d$ is \emph{enabled} along $p_1(\Vec{x}_1), \dots, p_d(\Vec{x}_d)$ if for each $i \in [d]$, $\theta_i$ is fireable at $p_i(\Vec{x}_i+\Vec{a}_i)$ for some $\Vec{a}_i \in \SeqCone(\theta_1, \dots, \theta_{i-1})$.
If a run $\pi$ visits these configurations in that order, we also say that $\theta_1, \dots, \theta_d$ are enabled along $\pi$ and $\pi$ is jointly diagonal.

The proof below is divided into three parts, corresponding to the following steps.
In~\cref{subsec:find-prog-pump-cycles}, we present the core technique for extracting joint diagonality from pumpability.
In~\cref{subsec:prog-pump-to-diag-pump}, we use the efficient representation established in~\cref{sec:repr-2vass} to transform joint diagonality into diagonality.
Finally, in~\cref{subsec:len-par-pump-3-vass}, we establish the two doubly-exponential self-reductions.

\subsection{Identification of Joint Diagonality via Pumpability}
\label{subsec:find-prog-pump-cycles}

In this subsection, we show that, given a run $s \xrightarrow{\pi} t$ in a pumpable 3-VASS $(V,s,t)$, one can either identify a jointly diagonal prefix of $\pi$, or conclude that at least one coordinate along $\pi$ remains bounded. 
Let $\pi$ be a run in a $d$-VASS $V$ and let $B \in \NN$. 
We define $\BDCounters(\pi, B)$ as the set of indices $i \in [d]$ such that every configuration $q(\Vec{x})$ on $\pi$ satisfies $\Vec{x}(i) < B$.
In other words, $\BDCounters(\pi, B)$ consists of those coordinates remaining less than $B$ along $\pi$.
With this notion, the above observation can be formally stated as follows.

\begin{restatable}{lemma}{lemProgPumpOrBoundedInParThreeVASS}
    \label{lem:prog-pump-or-bounded-in-par-3-vass}
    Let $(V, s, t) \in \seqVASSPro$ be a $k$-component 3-VASS, $s$ is pumpable and $s \stackrel{\pi}{\longrightarrow}t$.
    There exists some $B=\PolyF{\Size(V,s)}^{k}$ such that $\pi$ can be partitioned into 
    \begin{equation*}
        s = c_1 \xrightarrow{\pi_1} c_2 \xrightarrow{\pi_2} c_3 \xrightarrow{\pi_3} t,
    \end{equation*}
    where $|\BDCounters(\pi_1, B)| \ge 2$ and $|\BDCounters(\pi_1\pi_2, B)| \ge 1$.
    If $\pi_3 \neq \varepsilon$, there is a sequence of non-redundant jointly diagonal cycles $\theta_1, \theta_2, \theta_3$, enabled along $c_1$, $c_2$, $c_3$, such that $|\theta_i| \le B$ for all $i \in [3]$ and the followings are valid. 
    \begin{enumerate}
        \item For each $\iota \in [3]$ with $\Eff{\theta_1}(\iota) = 0$, it holds that $\iota \in \BDCounters(\pi_1, B)$.
        \item For each $\iota \in [3]$ with $\Eff{\theta_1}(\iota) = \Eff{\theta_2}(\iota) = 0$, it holds that $\iota \in \BDCounters(\pi_2, B)$.
    \end{enumerate}
\end{restatable}

We remark that these segments can be empty.
In fact, the construction further guarantees that $\theta_2=\theta_1$ whenever $\pi_1=\varepsilon$, and $\theta_3=\theta_2$ whenever $\pi_2=\varepsilon$.
If $\pi_3 = \varepsilon$, then at least one coordinate remains bounded along $\pi$. 
To establish~\cref{lem:prog-pump-or-bounded-in-par-3-vass}, we need to identify along $\pi$ three short cycles, possibly with repetition, whose effects generate a sequential cone containing a strictly positive vector.
The pumpability of $s$ guarantees the existence of the initial non-negative cycle $\theta_1$.
If $\Eff{\theta_1}$ contains at least one zero entry, a second, distinct cycle is required.
Since $\Eff{\theta_1}$ is guaranteed to have at least one positive entry, we may ignore the corresponding coordinate and apply the following pumping technique for 2-VASS, known as the Non-Negative Cycle Lemma in~\cite[Lemma 3.2]{CzerwinskiLLP19}, to obtain the second short cycle.

\begin{lemma}[{\cite[Lemma 3.2]{CzerwinskiLLP19}}]
    \label{lem:par-pump-or-bounded-in-2-vass}
    There exists a polynomial $P$ such that, for any $k$-component 2-VASS $(V, s, t)$, if there is a run $s \xrightarrow{\pi} t$ such that $\norm{t} \ge P(\Size(V))^k \cdot (\norm{s} + 1)$, then $\pi$ contains a pumpable configuration enabling a cycle $\theta$ with $|\theta| \le P(\Size(V))$.
\end{lemma}
\begin{remark}
    The bound in the original statement is in the form of $P(\Size(V))^{|Q|}$. 
But if one looks into the proof in~\cite{CzerwinskiLLP19}, the exponent is bounded by the number of SCCs, which subsumes $|Q|$. 
\end{remark}

In the general scenario, applying \cref{lem:par-pump-or-bounded-in-2-vass} yields a cycle $\theta_2$ that,  together with $\theta_1$, allows one to pump up only two coordinates.
In this case, we ignore both the coordinates and view $V$ as a 1-VASS. 
\cref{lem:par-pump-or-bounded-in-2-vass} remains applicable and enables us to obtain a third cycle.
It is not difficult to see that, by an appropriate combination, these three cycles can be fired in sequence to increase all three coordinates. 
We also emphasize that if any application of~\cref{lem:par-pump-or-bounded-in-2-vass} fails, then $\pi_3=\varepsilon$ follows immediately; this is precisely where the dichotomy arises.

\subsection{Transformation from Joint Diagnality to Diagnality}
\label{subsec:prog-pump-to-diag-pump}

Recall that a $d$-dimensional hybrid set $S$ is a set of the form of $\Vec{b}+\ResPeriodic{P}{\Vec{c}}{K}$, where $\Vec{b}\in \NN^d$, $K \in \NN$, $P = \{\Vec{p}_1,\cdots, \Vec{p}_k\} \subseteq \NN^d\setminus \{\Vec{0}\}$ is a finite set, and $\Vec{c} \in \NN^{k}$.
The description size of this set is given by $\NORM{S}=\max\{\norm{\Vec{b}}, \norm{P}, \norm{\Vec{c}}\}$.
Notice that for each $i \in [k]$, if $\Vec{c}(i) = 0$, then $S + \{\Vec{p}_i\}^*\subseteq S$, meaning that there is no restriction on the number of times $\Vec{p}_i$ is repeated.
We may write $\Vec{b}+\ResPeriodic{P}{\Vec{c}}{K}$ as $\Vec{b} + \ResPeriodic{P_1}{\Vec{c}'}{K} + P_0^*$, where $P = P_1 \cup P_0$, to emphasize that the periods in $P_0$ may be used arbitrarily many times.

The main result of this part is the following proposition, which provides a characterization of the approximation properties of jointly diagonal configurations.

\begin{restatable}{proposition}{propProgPumpToDiagPump}
    \label{prop:prog-pump-to-diag-pump} 
    Let $(V, s, t) \in \seqVASSPro$ be a $k$-component 3-VASS and $s \stackrel{\pi}{\longrightarrow}t$ be a run in $V$.
    If $s$ is pumpable, then there exist a configuration $c = q(\Vec{y})$ in $\pi$ and a hybrid set $S = \Vec{b} + \ResPeriodic{P_1}{\Vec{c}}{K} + P_0^*$ such that $\Vec{y} \in  S \subseteq \ReachqVs{q}{V}{s}$ and the following statements are valid. 
    \begin{enumerate}
        \item \label{prop:prog-pump-to-diag-pump-1}  $\NORM{S} \le \Size(V, s)^{\bigO(k^2)}$, 
        \item \label{prop:prog-pump-to-diag-pump-2}  $\min \Len(V, s, q(\Vec{z})) \le \Size(V,s,q(\Vec{z}))^{\bigO(k^2)}$ for all $\Vec{z} \in S$, and
        \item \label{prop:prog-pump-to-diag-pump-3}  $P_0 \cap \NN_{>0}^3 \neq \emptyset$, unless $c = t$.
    \end{enumerate}
\end{restatable}

\cref{prop:prog-pump-to-diag-pump-1} ensures that the reachability set of a pumpable sequential 3-VASS at some given state $q$ can be approximated by a hybrid set $\Vec{b} + \ResPeriodic{P_1}{\Vec{c}}{K} + P_0^*$ of polynomial size.
\cref{prop:prog-pump-to-diag-pump-2} ensures that every configuration $q(\Vec{y})$ in this hybrid set admits a run from $s$ that is short in terms of $\Size(V,s,q(\Vec{z}))$. Concerning $P_0$, there are two cases. (i)
If $P_0 \cap \QQ_{>0}^3 \neq \emptyset$, we can add the periods in $P_1 \cup P_0$ as self-loops at the state $q$, and let $q(\Vec{b})$ be the new source configuration. 
\cref{prop:prog-pump-to-diag-pump-3}  guarantees that the new source $q(\Vec{b})$ is diagonal in the new VASS.
(ii) Otherwise, $P_0 \cap \QQ_{>0}^3 = \emptyset$, then $t=q(\Vec{y})$, where \cref{prop:prog-pump-to-diag-pump-2} implies that there is a run from $s$ to $t$ whose length is bounded by $\PolyF{\Size(V,s,t)}^{k^2}$.

We provide a high-level proof sketch of~\cref{prop:prog-pump-to-diag-pump}.
Firstly, \cref{lem:prog-pump-or-bounded-in-par-3-vass} guarantees that prior to the identification of a sequence of jointly diagonal cycles $\theta_1, \theta_2, \theta_3$, there exist some coordinates bounded throughout the prefix. 
Consequently, this bounded coordinate can be encoded into the states, reducing the system to a 2-VASS and permitting the application of the approximation technique.
The primary technical challenge is to ensure that the resulting period set $P_0$ contains a strictly positive vector. 
To address this, we partition the run into two segments at the configuration where $\theta_2$ becomes enabled. 
After applying the approximation technique to the first segment, we obtain a hybrid set whose period set can, without loss of generality, be assumed to contain $\Eff{\theta_1}$, by the following lemma.

\begin{restatable}{lemma}{lemReachClosedUnderPumpingCycle}
    \label{lem:reach-closed-under-pumping-cycle}
    Let $(V,s, q(\Vec{y}))$ be a d-VASS such that $s = p(\Vec{x}) \xrightarrow{*} q(\Vec{y})$ in $V$.
    Suppose $s \xrightarrow{\theta_1} p(\Vec{x}+\Vec{a})$ in $V$ for some cycle $\theta_1$ with $\Vec{a}:=\Eff{\theta_1} \ge \Vec{0}$ and $q(\Vec{y}+\lambda \cdot\Vec{a})\xrightarrow{\theta_2} q(\Vec{y}+\Vec{b})$ in $V$ for some $\lambda \in \NN$ and some cycle $\theta_2$ with $\Vec{b}:=\lambda\cdot\Vec{a}+\Eff{\theta_2} \ge \Vec{0}$.
    Then for each $\ell \in \NN$, it holds that $\Vec{y} + \ell\cdot \Vec{b} \in \ReachqVs{q}{V}{s}$ and
    \begin{equation*}
        \min\Len(V,s,q(\Vec{y}+ \ell \cdot \Vec{b})) \le \min \Len(V,s,q(\Vec{y})) + \ell \cdot \lambda \cdot |\theta_1| + \ell\cdot |\theta_2|.
    \end{equation*}
\end{restatable}

\cref{lem:reach-closed-under-pumping-cycle} allows one to enrich the period set with an appropriate linear combination of $\Eff{\theta_1}$ 
and $\Eff{\theta_2}$, denoted as $\Vec{c}$, which possesses at least two positive entries in the general case.
Subsequently, we perform \emph{a second round of 
approximation} and incorporate a suitable combination of $\Vec{c}$ and 
$\Eff{\theta_3}$ to obtain a strictly positive vector. 
It should be noted that the source of the second round of approximation is a \emph{hybrid set} rather than a single point. 
This distinction makes it necessary to invoke \cref{thm:hybrid-2vass-hybrid} and~\cref{lem:len-target-in-hybrid-set}.
Since the two constructions in the proof of \cref{thm:hybrid-2vass-hybrid} and~\cref{lem:len-target-in-hybrid-set} are fundamentally similar, we provide a uniform statement in the following lemma.

\begin{restatable}{lemma}{lemOneTurnApprox}
    \label{lem:one-turn-approx}
    Suppose $B \in \NN$, $(V, p(\Vec{x}), q(\Vec{y}))$ is a 3-VASS, and $p(\Vec{x}) \stackrel{\pi}{\longrightarrow}q(\Vec{y})$ for some $\pi$ such that $\BDCounters(\pi, B)\neq \emptyset$.
    For every nonempty set $I\subseteq \BDCounters(\pi,B)$ and every hybrid set $S_p$ satisfying (i) $\Vec{x} \in S_p$ and (ii) $\Vec{z}(\iota) = \Vec{x}(\iota)$ for all $\Vec{z} \in S_p$ and all $\iota \in I$, there exists a hybrid set $S_q$ 
    satisfying (i) $\Vec{y} \in S_q$ and (ii) $\Vec{z}(\iota) = \Vec{y}(\iota)$ for all $\Vec{z} \in S_q$ and all $\iota \in I$, and moreover the following statements are valid.  
    \begin{enumerate}
    \item $\NORM{S_q} \le \PolyF{B + \Size(V) +\NORM{S_p}}$.
        \item For each $\Vec{y}' \in S_q$, there exists some $\Vec{x}' \in S_p$ with a run $p(\Vec{x}')\stackrel{\rho}{\longrightarrow}q(\Vec{y}')$ in $V$ such that
        \begin{equation*}
            \norm{\Vec{x}'}, |\rho| \le \PolyF{B + \Size(V) + \NORM{S_p} + \norm{\Vec{y}'}}.
        \end{equation*}
    \end{enumerate}
\end{restatable}

To understand the lemma, one may think of it as an application of \cref{thm:hybrid-2vass-hybrid,lem:len-target-in-hybrid-set} to the geometrically $2$-VASS obtained by encoding the values in the $\iota$-counter, for all $\iota\in I$, into states. 
We remark that the case where the source is a single point $\Vec{x}$ is also covered by \cref{lem:one-turn-approx}.

We are now ready to prove \cref{prop:prog-pump-to-diag-pump}.
Since $B=\PolyF{\Size(V,s)}^{k}$ and \cref{lem:one-turn-approx} needs to be applied twice in the worst case, the $\bigO(k^2)$ exponent arises naturally in \cref{prop:prog-pump-to-diag-pump}.

\begin{proof}[Proof of~\cref{prop:prog-pump-to-diag-pump}]
    Assume that $V=(Q,T)$ and let $M:=\Size(V,s)$.
    Note that $M > 1$.
    By~\cref{lem:prog-pump-or-bounded-in-par-3-vass}, we obtain a partition of $\pi$ in the form of 
    \begin{equation*}
        s = c_1 \xrightarrow{\pi_1} c_2 \xrightarrow{\pi_2} c_3 \xrightarrow{\pi_3} t
    \end{equation*}
    with $|\BDCounters(\pi_1, B)| \ge 2$ and $|\BDCounters(\pi_1\pi_2, B)| \ge 1$ for some $B = \PolyF{M}^k$.
    We assume that $B \le M^{m\cdot k}$ for some $m \ge 1$ and let the polynomials given in~\cref{lem:one-turn-approx} be bounded by $F(x):=x^{m}$.
    There are two cases.
    \begin{description}
        \item 
        [Case 1. $\pi_3 = \varepsilon$.] Then $\BDCounters(\pi, B) \neq \emptyset$.
        Suppose $s=p(\Vec{x})$ and let $c = q(\Vec{y}) :=t$.
        Applying~\cref{lem:one-turn-approx} to the trivial hybrid set $\{\Vec{x}\}$, one obtains a hybrid set $S$ such that $\Vec{y} \in S \subseteq \ReachqVs{q}{V}{s}$
        and 
        \begin{equation*}
            \NORM{S} \le F(B + M) \le \left(M^{mk} + M\right)^m \le M ^{2m^2k}.
        \end{equation*}
        Moreover, for each $\Vec{z} \in S$, one has that $s \xrightarrow{\rho} q(\Vec{z})$ in $V$ with $|\rho| \le F(B+M+\norm{\Vec{z}})$, which implies that $\Vec{z} \in \ReachqVs{q}{V}{s}$ and
        \begin{equation*}
            \min \Len(V,s,q(\Vec{z})) \le F(B+ \Size(V,s,q(\Vec{z}))) \le \Size(V,s,q(\Vec{z})) ^{2m^2k}.
        \end{equation*}
        This justifies the first case.
        \item [Case 2. $\pi_3 \ne\varepsilon$.] Then by~\cref{lem:prog-pump-or-bounded-in-par-3-vass}, there exist jointly diagonal cycles $\theta_1$, $\theta_2$, $\theta_3$ of length bounded by $B$, enabled along $\pi$, witnessed by configurations $c_1$, $c_2$, and $c_3$.
        Assume that $c_2=r(\Vec{w})$ and $c_3=p(\Vec{v})$.
        For the first segment, we assume, w.l.o.g., that $\Delta(\theta_1)(1) > 0$ and $\{2,3\} \subseteq \BDCounters(\pi_1, B)$. 
        Using again~\cref{lem:one-turn-approx}, we get some hybrid set $S_r = \Vec{b}_r + \ResPeriodic{P_r}{\Vec{c}_r}{K_r}$ with 
        $\NORM{S_r} \le M ^{2m^2k}$ such that 
        \begin{itemize}
            \item $\Vec{w} \in S_r$ and $\Vec{z}(\iota) = \Vec{w}(\iota)$ for all $\Vec{z} \in S_r$ and all $\iota \in \{2,3\}$, and
            \item for each $\Vec{z} \in S_r$, there is a run $s\xrightarrow{\rho} r(\Vec{z})$ in $V$ such that
            $
                |\rho| \le \Size(V,s,q(\Vec{z}))^{2m^2k}.
            $
        \end{itemize}
        Now consider the second cycle $\theta_2$.
        Let $\lambda := BM + 1$ and
        \begin{equation*}
            \Vec{a} := \lambda \cdot \Eff{\theta_1} + \Eff{\theta_2}.
        \end{equation*}
        Since $|\theta_2| \le B$, we have $\Vec{a} \ge \Vec{0}$ and $\Vec{a}(\iota) > 0$ holds for every $\iota \in [3]$ such that $\Eff{\theta_1}(\iota) > 0$ or $\Eff{\theta_2}(\iota) > 0$.
        Let $S_r' := S_r+\{\Vec{a}\}^* = \Vec{b}_r + \ResPeriodic{P_r}{\Vec{c}_r}{K_r} + \{\Vec{a}\}^*$.
        Notice that $\norm{\Vec{a}} \le (\lambda+1)\cdot BM \le M^{2mk+3}$.
        Consequently
        \begin{equation*}
            \NORM{S_r'} \le M^{2m^2k} + \norm{\Vec{a}} \le M^{6m^2k}.
        \end{equation*}

        \begin{restatable}{claim}{claimLenFirstRoundApprox}
            \label{claim:len-first-round-approx}
            $\Vec{w}\in S_r'\subseteq \ReachqVs{r}{V}{s}$ and $\min \Len(V,s,r(\Vec{z})) \le \Size(V,s,r(\Vec{z}))^{6m^2k}$ holds for each $\Vec{z}\in S_r'$.
        \end{restatable}
     
        If $\Vec{a} \ge \Vec{1}$, we can finish the proof by taking $c:=r(\Vec{w})$.
        Otherwise, we continue with the second segment $r(\Vec{w})\xrightarrow{\pi_2} p(\Vec{v})$, which satisfies $|\BDCounters(\pi_2, B)| \ge 1$.
        Since $\Eff{\theta_1}(1) > 0$, we assume w.l.o.g. that $\Vec{a}(3) = 0$, which implies that $3 \in \BDCounters(\pi_1\pi_2, B)$ by~\cref{lem:prog-pump-or-bounded-in-par-3-vass}.
        For the hybrid set $S_r'$, all elements in $S_r'$ have $\Vec{w}(3)$ on their third entries.
        Applying~\cref{lem:one-turn-approx}, there exists a hybrid set $S_p = \Vec{b}_p + \ResPeriodic{P_p}{\Vec{c}_p}{K_p}$ with 
        \begin{equation*}
            \NORM{S_p} \le F(B+M+\NORM{S_r'}) \le \left(3 \cdot M^{6m^2k}\right)^m \le M^{8m^3k}.
        \end{equation*}
        Moreover, it renders true that
        \begin{itemize}
            \item $\Vec{v} \in S_p$ and $\Vec{z}(3) = \Vec{v}(3)$ holds for all $\Vec{z} \in S_p$, and
            \item  for each $\Vec{v}' \in S_p$, there exists some $\Vec{w}' \in S'_r$ with a run $r(\Vec{w}')\xrightarrow{\rho} p(\Vec{v}')$ in $V$ such that
            \begin{equation}
                \label{eq:norm-w'-len-rho}
                \norm{\Vec{w}'}, |\rho| \le F(B+M+\NORM{S_r'} + \norm{\Vec{v}'}) \le \Size(V,s,p(\Vec{v}'))^{8m^3k}.
            \end{equation}
        \end{itemize}
        Consider the last cycle $\theta_3$.
        Recall that $\lambda = BM + 1$.
        Let 
        $\Vec{p} := \lambda \cdot \Vec{a} + \Eff{\theta_3}$.
        Clearly $\Vec{p} \ge \Vec{1}$.
        Let $S_p' := \Vec{b}_p + \ResPeriodic{P_p}{\Vec{c}_p}{K_p} + \{\Vec{p}\}^*$.
        Then $\NORM{S_p'} \le \NORM{S_p} + \norm{\Vec{p}} \le 2\cdot M^{8m^3k} \le M^{9m^3k}$.
        Let $c := p(\Vec{v})$ and $S := S_p'$. 
        The following claim justifies the second case.

        \begin{restatable}{claim}{claimSecondRoundApprox}
            \label{claim:second-round-approx}
            $\Vec{v} \in S_p' \subseteq \ReachqVs{p}{V}{s}$ and $\min \Len(V, s, p(\Vec{z})) \le \Size(V,s,p(\Vec{z}))^{68m^5k^2}$ for all $\Vec{z} \in S_p'$.
        \end{restatable}

    \end{description}
    The polynomial $x^{68m^5}$ subsumes all the upper-bound polynomials mentioned in the proof.
\end{proof}

\subsection{Length-Controlled Self-Reduction for Pumpable 3-VASS}
\label{subsec:len-par-pump-3-vass}

In this part, we establish the doubly-exponential self-reductions for $\pumpVASSPro$ and $\semidiagVASSPro$, namely~\cref{prop:2-exp-lcr-for-par-pump-3-vass}.
Recall from~\cref{lem:prog-pump-or-bounded-in-par-3-vass} that a jointly diagonal prefix can be extracted from a pumpable run.
The key step of the proof is to apply~\cref{prop:prog-pump-to-diag-pump} to characterize the target configurations of this prefix by a hybrid set $S=\Vec{b} + \ResPeriodic{P_1}{\Vec{c}}{K} + P_0^*$, where $P_0$ contains a positive vector.
Adding these periods in $P_1 \cup P_0$ to the original VASS as self-loops yields a semi-diagonal or diagonal VASS $(V',p(\Vec{b}),t')$.
It remains to take the linear constraint $\Inner{\Vec{c},\cdot}\le K$ into account, which only requires comparing the size of $K$ with $\min\Len(V',p(\Vec{b}),t')$.
If $K$ is sufficiently large, the restriction in $\ResPeriodic{P_1}{\Vec{c}}{K}$ can be omitted in the sense that a short run in the new VASS $(V',s',t')$ also induces a run in $V$.
Otherwise, we know that $K$ is bounded and so is every vector in $\Vec{b} + \ResPeriodic{P_1}{\Vec{c}}{K}$.
In the latter case, we add the periods in $P_0$ to the original VASS and investigate the runs from different sources in the form of $p(\Vec{w})$ to $t$, where $\Vec{w} \in \Vec{b} + \ResPeriodic{P_1}{\Vec{c}}{K}$.
The size amplification during this reduction is thus controlled by the length function of a VASS class with a stronger pumpability property.
The formal proof is given next.

\begin{proof}[Proof of~\cref{prop:2-exp-lcr-for-par-pump-3-vass}]
    We define a family of functions $(L_k)_{k \in \NN}$ by letting $L_k(x) := x^{2^{f(k)}}$ for all $k \in \NN$, where $f$ is the polynomial given in~\cref{prop:diagonal-case-lenbound}.
    We then show that $\semidiagVASSPro$ is $(L_k)_{k \in \NN}$-controlled self-reducible.
    Let $g$ be the maximum of the two polynomials appearing in~\cref{prop:prog-pump-to-diag-pump}.
    Assume that $m \ge 1$ is such that $g(x) \le x^m$ for all $x >1$.
    Let $R(x) := x^{21mk^2}$.
    
    For every $k \ge 1$ and every $k$-component VASS $(V,s,t)\in \semidiagVASSPro$ with size $M:=\Size(V,s,t)$, we assume that $s \xrightarrow{\pi} t$ in $V$ for some run $\pi$.
    In the case where $(V,s,t)$ is diagonal, $\min\Len(V,s,t) \le L_k(M)$ by~\cref{prop:diagonal-case-lenbound}.
    Then $(V,s,t)$ is reduced to a $0$-component 3-VASS $(V',s',t')$ where $V'$ contains a single state $q$ and a single self-loop $(q,\Vec{1}, q)$, $s'=q(\Vec{0})$ and $t'=q(L_k(M)\cdot\Vec{1})$.
    \begin{restatable}{claim}{ClaimSemiDiagFirst}
        \label{claim:semi-diag-1}
        The VASS $(V',s',t')$ has $0$-component.
        It holds that $s'\xrightarrow{*}t'$ in $V'$, $\Size(V',s',t') \le  R(L_k(M))$, and $\min\Len(V,s,t) \le R(\min\Len(V',s',t')+M)$.
    \end{restatable}

    Suppose $s$ is not diagonal in $V$.
    By~\cref{prop:prog-pump-to-diag-pump}, there exist a configuration $c = q(\Vec{y})$ on $\pi$ and a hybrid set $S = \Vec{b} + \ResPeriodic{P_1}{\Vec{c}}{K} + P_0^*$ such that $\Vec{y} \in  S \subseteq \ReachqVs{q}{V}{s}$ and the following is valid.
    \begin{equation}
        \label{eq:2-exp-lcr-for-par-pump-3-vass-1}
       \NORM{S} \le g(M)^{k^2} \text{ and } \min \Len(V, s, q(\Vec{z})) \le g(\Size(V,s,q(\Vec{z})))^{k^2} \text{ for all } \Vec{z} \in S.
    \end{equation}
    Moreover $q(\Vec{y}) = t$ whenever  $P_0 \cap \QQ_{>0}^3 = \emptyset$.
    There are two cases.
    \begin{description}
        \item[Case 1. $P_0 \cap \QQ_{>0}^3 = \emptyset$.]
        Then $c = t$, and 
            $\min \Len(V, s, t) \le g(M)^{k^2}\le R(M)$.
        We are done.
        \item[Case 2. $P_0 \cap \QQ_{>0}^3 \neq \emptyset$.] 
        There exists some $\Vec{\delta} \in P_0$ with $\Vec{\delta}\ge \Vec{1}$.
        Assume that $P_1$ can be enumerated as $\Vec{p}_1,\dots, \Vec{p}_n$, $V = (V_1)u_1\dots u_{k-1}(V_k)$ and $q$ belongs to $V_\iota$.
        We claim that $\iota>1$.
        Otherwise, there would be a simple path from $q$ to $p$.
        Note that $\Vec{b} + \{\Vec{\delta}\}^* \subseteq \ReachqVs{q}{V}{s}$.
        We have $s\xrightarrow{*} q(\Vec{z})$, where $\Vec{z}:=\Vec{b} + (|\sigma| \cdot \Size(V) + 1) \cdot \Vec{\delta} \ge \Drop(\sigma)$.
        Hence, $s \xrightarrow{*} q(\Vec{z}) \xrightarrow{\sigma}p(\Vec{z} + \Eff{\sigma})$ and $\Vec{z} + \Eff{\sigma} \ge \Vec{\delta}\ge \Vec{1}$.
        This implies that $s$ is diagonal, which is a contradiction.
        Now construct a VASS $V_{S}$ from $(V_\iota)u_{\iota}\dots u_{k-1}(V_k)$ by adding a transition set $T_S$, where
        \begin{equation*}
            T_S := \{(q,\Vec{a},q)\mid \Vec{a} \in P_1 \cup P_0\}.
        \end{equation*}
        Notice that $V_S$ has at most $k - 1$ components, and $q(\Vec{b}) \xrightarrow{*} q(\Vec{y}) \xrightarrow{*} t$ in $V_S$.
        The size bound can be derived as follows.
        \begin{equation}
            \label{eq:2-exp-lcr-for-par-pump-3-vass-2}
            \begin{aligned}
                \Size(V_S, q(\Vec{b}),t) &\le M + 3\cdot\norm{\Vec{b}} + 3 \cdot(|T| + |T_S|) \cdot (\norm{T} + \norm{T_S} + 1)\\
                &\le M + 3\cdot M^{mk^2} + 3 \cdot(M + (M^{mk^2} + 1)^3) \cdot(M+M^{mk^2} + 1)\\
                &\le M + 3\cdot M^{mk^2} + 81 \cdot M^{4mk^2}\\
                &\le 85\cdot M^{4mk^2}.
            \end{aligned}
        \end{equation}
        Moreover, $(V_S, q(\Vec{b}),t)$ is diagonal as $(V,s,t)$ is backward diagonal and $q(\Vec{b}) \xrightarrow{*} q(\Vec{b} + \Vec{\delta})$.
        Let $\ell := \min\Len(V_S, q(\Vec{b}), t)$.
        We now invoke a result about the diagonal $3$-VASS.
        By~\cref{prop:diagonal-case-lenbound}, 
        \begin{equation*}
            \ell \le L_k(85\cdot M^{4mk^2}) \le M^{(4mk^2+7)\cdot 2^{f(k)}} = L_k(M)^{11mk^2}.
        \end{equation*}
        \begin{description}
            \item[Case I. $\norm{\Vec{c}}\cdot \ell \le K$.]
            The following claim implies that the construction from $(V,s,t)$ to the diagonal VASS $(V',s',t') := (V_S, q(\Vec{b}), t)$ is doubly-exponentially self-reducible.
            \begin{restatable}{claim}{ClaimSemiDiagSecond}
                \label{claim:semi-diag-2}
                The diagonal $(V',s',t')$ has at most $k-1$ components, $s'\xrightarrow{*}t'$ in $V'$, $\Size(V',s',t') \le  R(L_k(M))$, and $\min\Len(V,s,t) \le R(\min\Len(V',s',t')+M)$.
            \end{restatable}
            
            \item[Case II. $K < \norm{\Vec{c}} \cdot \ell$.] 
            We have 
            $K < \norm{\Vec{c}} \cdot \ell \le g(M)^{k^2}\cdot L_k(M)^{11mk^2} \le L_k(M)^{12mk^2}$.
            Since $\Vec{y} \in S$, we can write it as $\Vec{y} = \Vec{w} + \Vec{q}$, where $\Vec{q} \in P_0^*$ and $\Vec{w} = \Vec{b} + \sum_{i=1}^n \Vec{\lambda}(i)\cdot \Vec{p}_i$ for some $\Vec{\lambda} \in \NN^n$ satisfying $\Inner{\Vec{c}, \Vec{\lambda}} \le K$.
            We may assume that $\Vec{c} \ge \Vec{1}$, since each $\Vec{p}_i$ correspond to $\Vec{c}(i) = 0$ can be removed and inserted into $P_0$.
            Then
            \begin{equation*}
                \norm{\Vec{w}} \le \norm{\Vec{b}} + \Inner{\Vec{1}, \Vec{\lambda}}\cdot \norm{P} \le 2g(M)^{k^2} \cdot  \Inner{\Vec{c}, \Vec{\lambda}} \le 2M^{mk^2} \cdot K \le L_k(M)^{14mk^2}.
            \end{equation*}
            We define a diagonal $(k{-}1)$-component VASS $(V',s',t')$ from $(V_\iota)u_{\iota}\dots u_{k-1}(V_k)$ by setting  
            $s' := q(\Vec{w})\text{ and }t' := t\text{ and }T':=T\cup T_{\Vec{w}}$,
            where $T_{\Vec{w}} := \{(q,\Vec{a},q)\mid \Vec{a} \in P_0\}$.
            By the following claim, $(V',s',t')$ satisfies all the required properties.
            \begin{restatable}{claim}{ClaimSemiDiagThrid}
                \label{claim:semi-diag-3}
                The diagonal $(k{-}1)$-component $(V',s',t')$ renders true that $s'\xrightarrow{*}t'$ in $V'$, $\Size(V',s',t') \le  R(L_k(M))$ and $\min\Len(V,s,t) \le R(\min\Len(V',s',t')+M)$.
            \end{restatable}
        \end{description}
    \end{description}

    We conclude that $\semidiagVASSPro$ is $(L_k)_{k \in \NN}$-controlled self-reducible.
    By~\cref{lem:2-exp-preserve-via-reduction}, it admits $(G_k)_{k \in \NN}$-short runs, where $G_k(x) = x^{2^{\bigO({f(k)})k}}$. By the same token, we can show that $\pumpVASSPro$ is doubly-exponentially self-reducible and admits doubly-exponentially short runs.
\end{proof}

\section{Sequential $3$-VASS}
\label{sec:general-case}

In this section, we establish the short run property for the sequential VASS.

\propThreeVASSlenboundIntro*

We have established a similar bound for the pumpable VASS in~\cref{prop:len-par-pump-3-vass}.
Similarly, the reduction from sequential VASS to pumpable VASS is indirect.
We need to deal with semi-pumpability.
The proof strategy of~\cref{prop:3-VASS-lenbound-intro} closely follows the one in~\cref{subsec:len-par-pump-3-vass}.
By~\cref{lem:2-exp-preserve-via-reduction}, it suffices to establish the following length-bounded reductions.

\begin{proposition}\label{seqdoubleexpoential}
    \label{prop:2-exp-lcr-for-gen-3-vass}
    (i) $\semipumpVASSPro$ is doubly-exponentially self-reducible, and admits doubly-exponentially short runs.
    (ii) $\seqVASSPro$ is doubly-exponentially self-reducible, and admits doubly-exponentially short runs.
\end{proposition}

\begin{proof}[Proof of~\cref{prop:3-VASS-lenbound-intro}]
    Let $(V, s, t)$ be a $k$-component 3-VASS with $s \xrightarrow{*} t$ in $V$.
    By~\cref{prop:2-exp-lcr-for-gen-3-vass} and~\cref{lem:2-exp-preserve-via-reduction}, it holds that $\min \Len (V,s,t) \le \Size(V, s, t)^{2^\PolyF{k}}$.
\end{proof}

Now we proceed to prove~\cref{prop:2-exp-lcr-for-gen-3-vass}.
We first outline the proof idea.
Consider $(V,s,t)\in\seqVASSPro$. The most challenging situation we need to handle is when $(V,s,t)$ is neither forward pumpable nor backward pumpable.
In this case, we can apply an important technique called \emph{Rackoff's extraction}~\cite{DBLP:conf/lics/LerouxS19} to extract at least one bounded coordinate within the first SCC.
Recall that $\BDCounters(\pi, B)$ denotes the set of coordinates whose values do not exceed $B - 1$ along $\pi$, where $\pi$ is a run and $B \in \NN$ is a constant.
This technique is formally stated as follows.

\begin{lemma}[{\cite[Lemma A.1]{DBLP:conf/lics/LerouxS19}}]
    \label{lem:reach-unbounded-of-ever-unbounded}
    There exists a polynomial $R$ such that, for every 3-VASS $V$ and number $U \in \NN$, and for any run $\pi$ in $V$ with $\BDCounters(\pi, B)=\emptyset$, where $B=R(U,\Size(V))$, there is a run $\rho$ starting from the same source as $\pi$ with $|\rho| \le B$ and ending at $q(\Vec{x})$ satisfying $\Vec{x} \ge U \cdot \Vec{1}$.
\end{lemma}

An immediate consequence of~\cref{lem:reach-unbounded-of-ever-unbounded} is the  following. 

\begin{restatable}{corollary}{corDiagPumpOfUnbounded}
    \label{cor:diag-pump-of-unbounded}
    Let $(V,s)$ be a strongly connected 3-VASS.
    If there exists a run $\pi$ starting from $s$ with $\BDCounters(\pi, U) = \emptyset$ for some $U := \PolyF{\Size(V,s)}$, then $s$ is diagonal in $V$.
\end{restatable}

If $\BDCounters(\pi, U) \not= \emptyset$, the approximation technique for bounded runs developed in~\cref{lem:one-turn-approx} can be used to characterize the reachability set of the first SCC by a hybrid set.
The remainder of the proof follows closely the argument in the proof of~\cref{prop:2-exp-lcr-for-par-pump-3-vass}.
    
\begin{proof}[Proof of~\cref{prop:2-exp-lcr-for-gen-3-vass}]
Define a family of functions $(L_k)_{k \in \NN}$ by letting $L_k(x) := x^{2^{f(k)}}$ for each $k \in \NN$, where $f$ is the polynomial given in~\cref{prop:len-par-pump-3-vass}.
    We show that $\semipumpVASSPro$ is $(L_k)_{k \in \NN}$-controlled self-reducible.
    Let $g$ be the maximum of the two polynomials appearing in~\cref{lem:one-turn-approx}.
    Assume that $g(x) \le x^m$ holds for some $m \ge 1$ and all $x >1$.
    Let $R(x) := x^{30\cdot m^2}$.

    Suppose $(V,s,t)\in \semipumpVASSPro$, where $V=(V_1)u_1\dots u_{k-1}(V_k)$ and $t = q(\Vec{y})$.
    Let $M:=\Size(V,s,t)$.
    Suppose $s$ is not pumpable in $V$, and 
    \begin{equation*}
        s \xrightarrow{\pi_1} q_1(\Vec{y}_1) \xrightarrow{u_1} p_2(\Vec{x}_2) \xrightarrow{\pi'} t,
    \end{equation*}
    where $\pi'$ is a run in $(V_2)u_2\dots u_{k-1}(V_k)$.
    By~\cref{cor:diag-pump-of-unbounded}, the index set of bounded coordinates along $\pi_1$ is non-empty, i.e., $\BDCounters(\pi_1, B) \neq \emptyset$, where $B = \PolyF{\Size(V,s)}$.
    We may as well assume that $B \le \Size(V,s)^m \le M^m$ for the $m$ introduced in the above.
    Thus $I:=\BDCounters(\pi_1u_1, B') \supseteq \BDCounters(\pi_1, B) \neq \emptyset$, where $B' := B+M$.
    By applying~\cref{lem:one-turn-approx} to $(V,s,p_2(\Vec{x}_2))$ and $I$, one can obtain a hybrid set $S = \Vec{b} + \ResPeriodic{P_1}{\Vec{c}}{K}+P_0^*$ rendering true $\Vec{x}_2 \in  S \subseteq \ReachqVs{p_2}{V}{s}$ and the following inequalities.
    \begin{equation}\label{eq:2-exp-lcr-for-gen-3-vass-1}
    \begin{aligned}
    \NORM{S} &\le g(B'+M) \le M^{3m^2},\\    
        \min \Len(V, s, p_2(\Vec{z})) &\le g(B'+\Size(V,s,p_2(\Vec{z}))) \le \Size(V,s,p_2(\Vec{z}))^{3m^2}\ \text{for all}\ \Vec{z} \in S.
    \end{aligned}
    \end{equation}
  We may assume that $\Vec{c}\ge \Vec{1}$ if $P_1 \neq \emptyset$.
    We now construct a $(k-1)$-component VASS $(V_{S}, s_S,t)$ by removing the first SCC and turning the periods in $P_1\cup P_0$ into self-loops at the first state $p_2$ of $(V_2)u_2\dots u_{k-1}(V_k)$.
    Let $T_{tail}$ be the transition set of $(V_2)u_2\dots u_{k-1}(V_k)$ and $Q_{tail}$ be the state set of $(V_2)u_2\dots u_{k-1}(V_k)$. 
    The new 3-VASS instance $(V_{S}, s_S,t)=((Q_S,T_S),s_S,t)$ is given by
    \[
    s_S := p_2(\Vec{b}) \text{ and } Q_S=Q_{tail} \text{ and }
     T_S := T_{tail}\cup \{(p_2,\Vec{a},p_2)\mid \Vec{a} \in P_1 \cup P_0\}.
    \]
    Observe that $s_S \xrightarrow{*}p_2(\Vec{x}_2) \xrightarrow{\pi'} t$ in $V_S$ and
    \begin{equation}
        \label{eq:2-exp-lcr-for-gen-3-vass-2}
        \begin{aligned}
            \Size(V_S, s_S, t) &\le M + 3 \cdot \norm{\Vec{b}} + 3\cdot (M + |T_S|) \cdot (M + \norm{T_S} + 1)\\
            &\le M + 3\cdot M^{3m^2} + 3 \cdot (M + (M^{3m^2} + 1)^3) \cdot (M + M^{3m^2} + 1)\\
            &\le 85\cdot M^{12m^2}.
        \end{aligned}
    \end{equation}
    Suppose $P_1=\{\Vec{p}_1,\cdots,\Vec{p}_n\}$, where the order of $\Vec{p}_1,\cdots,\Vec{p}_n$ is consistent with the order of $\Vec{c}(1),\ldots,\Vec{c}(n)$.
    Let $\ell := \min\Len(V_S,s_S,t)$.
    There are two cases depending on whether $\norm{\Vec{c}}\cdot \ell\le K$.
    \begin{description}
        \item[Case 1. $\norm{\Vec{c}}\cdot \ell\le K$.] We let $(V',s',t'):=(V_S,s_S,t)$. It suffices to establish the following claim.
        \begin{restatable}{claim}{ClaimSemiParFirst}
            \label{claim:semi-par-1}
            The $(k{-}1)$-component 3-VASS $(V',s',t') \in \semipumpVASSPro$ and renders true the inequlities: $\Size(V',s',t') \le  R(L_k(M))$ and $\min\Len(V,s,t) \le R(\min\Len(V',s',t')+M)$.
        \end{restatable}
 
        \item[Case 2. $K < \norm{\Vec{c}}\cdot \ell$.] In this case, $P_1 \ne \emptyset$.
        By construction, $(V_S,s_S,t)$ is forward pumpable since $p_2(\Vec{x}_2)\xrightarrow{*}p_2(\Vec{x}_2+\Vec{\delta})$ for some $\Vec{\delta} \in P_1$ and $\Vec{\delta}>\Vec{0}$.
        So $(V_S,s_S,t)$ is pumpable.
        By \cref{prop:len-par-pump-3-vass}, 
        \begin{equation*}
            \begin{aligned}
                \ell &\le L_k(85\cdot M^{12m^2}) \le L_k(M^{19m^2}) = M^{19m^2\cdot 2^{f(k)}} \le L_k(M)^{19m^2}, \text{\ and}\\
                K &\le \norm{\Vec{c}} \cdot \ell \le M^{3m^2}\cdot L_k(M)^{19m^2} \le L_k(M)^{22m^2}.
            \end{aligned}
        \end{equation*} 
        Since $\Vec{x}_2 \in S$, we can write it as $\Vec{x}_2 = \Vec{w} + \Vec{q}$, where $\Vec{q} \in P_0^*$ and $\Vec{w} = \Vec{b} + \sum_{i=1}^n \Vec{\lambda}(i)\cdot \Vec{p}_i$ for some $\Vec{\lambda} \in \NN^n$ satisfying $\Inner{\Vec{c}, \Vec{\lambda}} \le K$.
        Hence, 
        \begin{equation*}
            \norm{\Vec{w}} \le \norm{\Vec{b}} + \Inner{\Vec{1}, \Vec{\lambda}}\cdot \norm{P_1} \le 2M^{3m^2} \cdot  \Inner{\Vec{c}, \Vec{\lambda}} \le 2M^{3m^2} \cdot K \le 2L_k(M)^{25m^2}.
        \end{equation*}
        We construct a VASS $V'$ from $(V_2)u_2\dots u_{k-1}(V_k)$ by adding some self-loops to represent $P_0$:
        \begin{equation*}
            T' := T_{tail}\cup\{(p_2,\Vec{a},p_2)\mid \Vec{a} \in P_0\}.
        \end{equation*}
        Let $s' := p_2(\Vec{w})$ and $t' := t$.
        We can prove the following claim for $(V',s',t')$.
        \begin{restatable}{claim}{ClaimSemiParSecond}
            \label{claim:semi-par-2}
            The $(k{-}1)$-component 3-VASS $(V',s',t')\in\semipumpVASSPro$ and renders true the inequlities: $\Size(V',s',t') \le  R( L_k(M))$ and $\min\Len(V,s,t) \le R(\min\Len(V',s',t')+M)$.
        \end{restatable}
    \end{description}
    
    We conclude from~\cref{claim:semi-par-1}, \cref{claim:semi-par-2} and \cref{lem:2-exp-preserve-via-reduction} that $\semipumpVASSPro$ is $(L_k)_{k \in \NN}$-controlled self-reducible and admits doubly-exponentially short runs.
    Similarly, $\seqVASSPro$ is doubly-exponentially self-reducible and admits doubly-exponentially short runs.
\end{proof}

\section{Future Work}
\label{sec:conclusion}

There are some obvious open problems for future studies. 
We mention three of them. 

\paragraph*{Is $\pumpVASSPro$ in $\PSPACE$?} 
It has been proved that the reachability problem of the diagonal and wide $3$-VASS is in $\PSPACE$, matching the lower bound for the $2$-VASS reachability problem. 
Since every $2$-VASS can be seen as a pumpable $3$-VASS, $\pumpVASSPro$ is $\PSPACE$-hard.
Therefore, it would be interesting to see if $\pumpVASSPro$ is $\PSPACE$-complete.
More generally, we may ask whether $\mathrm{Pump}\mathbb{VASS}^d$, where $d\ge3$, is strictly simpler than $\mathbb{VASS}^d$ in complexity-theoretical terms?
A positive answer would be significant to our understanding of the VASS reachability problem.

\paragraph*{What is the exact complexity of $\VASSPro$?}
Determining the exact complexity of $\VASSPro$ is the most natural thing to do next. 
The current lower bound is $\PSPACE$, while this paper shows that $\EXPSPACE$ is an upper bound. 
The precise picture is yet to be unveiled. We tend to think that the reduction from a general $3$-VASS to a pumpable $3$-VASS cannot be done in polynomial space, which might suggest that $\VASSPro$ is $\EXPSPACE$-hard. However, as pointed out in~\cite{DBLP:conf/icalp/CzerwinskiJ0O25}, we still do not know whether there exists a $3$-VASS that admits doubly exponential shortest runs. 
In the absence of such an example, it appears unlikely that we can come up with an $\EXPSPACE$-hardness proof.

\paragraph*{Is $\mathbb{VASS}^4$ elementary?} So far, all known upper bounds for $\mathbb{VASS}^d$ in dimension $d\geq 4$ are obtained within the KLM framework, where the bound cannot be smaller than   the size of the finite reachable set. 
Consequently, the best known upper bound for $\mathbb{VASS}^4$ remains far from elementary. 
On the other hand, $\mathsf{TOWER}$-hardness is currently known only for $\mathbb{VASS}^8$. This leaves a substantial gap in our understanding of the complexity of $\mathbb{VASS}^4$.
The work of~\cite{DBLP:conf/icalp/CzerwinskiJ0O25}, and this work as well, looks for an upper bound for shortest runs.
It ought to be fruitful to take a look at the $\mathbb{VASS}^4$ from this new perspective.

\bibliographystyle{alpha}
\bibliography{reference}

\appendix
\section{Proof of \cref{thm:3-vass-in-exp-intro}}
\label{sec:proofofmaintheorem}
\VASSTINEXP*

\begin{proof}
It suffices to consider the sequential 3-VASS. 
For each $k \ge 1$ and every $k$-component 3-VASS $(V,s,t)$, if $s\xrightarrow{*}t$ in $V$, by~\cref{prop:3-VASS-lenbound-intro}, there exists a short path $s\xrightarrow{\pi} t$ in $V$ with $|\pi| \le \Size(V,s,t)^{2^{\PolyF{k}}}$.
The reachability $s\xrightarrow{*}t$ is decided by inspecting all runs of length at most $\Size(V,s,t)^{2^{\PolyF{k}}}$.
Now
\begin{equation*}
        \log|\pi| = 2^{\PolyF{k}} \cdot \log\Size(V,s,t) \leq 2^{\PolyF{\Size_{\textsc{bin}}(V,s,t)}} \leq 2^{\PolyF{\Size(V,s,t)}}
\end{equation*}
where $\Size_{\textsc{bin}}(V,s,t)$ denotes the input size under binary encoding. That is, 
\begin{equation*}
    \Size_{\textsc{bin}}(V,s,t) = |Q| + d \cdot |T| \cdot \log (\norm{T} + 1) + d \cdot \log (\norm{s} + 1) + d\cdot \log (\norm{t} + 1).
\end{equation*} The nondeterministic traversal algorithm can be carried out in exponential space under both binary and unary encoding schemes. 
Hence, $\mathbb{VASS}^3$ is in \EXPSPACE{}.
\end{proof}
\section{Proofs for~\cref{sec:overview}}
\label{Sec-appendix-overview-of-the-algorithm}

\subsection{Proof of~\cref{obs:0-com-vass-len}}\label{Sec-appendix-obs:0-com-vass-len}
\ObsZerodVASSReach*
\begin{proof}
    Suppose that $V=(Q,T)$ and $s \xrightarrow{\rho}t$ in $V$. We can assume that $\rho$ contains no zero transitions. Since each transition increases at least one coordinate by at least $1$, it is immediate that $|\rho| \le d\cdot \norm{t}$.
\end{proof}

\subsection{Proof of~\cref{lem:2-exp-preserve-via-reduction}}\label{sec:appendix-2-exp-preserve-via-reduction}
\LemmaTwoEXPPreserveViaReduction*
\begin{proof}
Let $L_k(x) = x^{2^{f(k)}}$. By definition of length-controlled self-reduction, there exists a polynomial $f(x)$ satisfying the required conditions.
We choose a sufficiently large $c \ge 2$ such that $f(x) \le x^c$ for every $x > 1$.
Then $R(x) \le x^{ck^2}$ for each $x > 1$.
Define 
\begin{equation*}
        g(k) := \left(c^3\cdot (k+1)^4\cdot 2^{f(k)}\right)^{k} \le \left(c\cdot 2^{f(k)}\right)^{5k} \le 2^{\bigO(f(k))\cdot k}
\end{equation*}
and $G_k(x):= x^{g(k)}$ for each $k \in \NN$. Note that $G_0(x)=x$.
Let $(V,s,t) \in \mathcal{V}$ be a $k$-component VASS with $s \xrightarrow{*} t$ for some $\pi$ with $M := \Size(V, s, t) > 1$.
We prove by induction on $k \ge 0$.
In the base case, we have $k = 0$.
By~\cref{obs:0-com-vass-len}, we immediately obtain that 
\begin{equation*}
        \min\Len(V,s,t) \le d\cdot \norm{t} \le M \le G_0(M).
\end{equation*}
Then we consider the case where $k \ge 1$.
By definition, there exists a sequential VASS $(V',s',t')$ such that $s'\xrightarrow{*}t'$ in $V'$, $\min\Len(V,s,t) \le R(\ell + M) \le (\ell + M)^{ck^2}$, and
\begin{equation*}
        \Size(V',s',t')\le R(L_k(M)) \le M^{ck^2 \cdot 2^{f(k)}},
\end{equation*}
where $\ell := \min\Len(V',s',t')$.
Since $(V',s',t') \in \mathcal{V}$ has at most $k-1$ components, then by the induction hypothesis, we have $\ell \le M^{ck^2 \cdot 2^{f(k)} \cdot g(k-1)}$, $g(k-1)=\left(c^3\cdot k^4 \cdot 2^{f(k-1)}\right)^{k-1}$, and
\begin{equation*}
        \min \Len(V,s,t) \le \left(M^{ck^2\cdot 2^{f(k)} \cdot g(k-1)} + M\right) ^ {ck^2}
            \le M^{c^3\cdot (k+1)^4 \cdot 2^{f(k)}\cdot g(k-1)}\le M^{g(k)}.
\end{equation*}
We conclude that $\mathcal{V}$ admits $(G_k)_{k \in \NN}$-short runs.
\end{proof}

\section{Proofs for~\cref{sec:repr-2vass}}\label{sec:appendix-repr-2vass}

We first recall the relevant notions and results of~\cite{DBLP:conf/icalp/CzerwinskiJ0O25} in \cref{sec:repr-2vass-known}. We then generalize the results from the case of a  singleton source to that of a hybrid set in \cref{sec:repr-2vass-hybrid}. Finally, in \cref{sec:repr-2vass-Geo-hybrid}, we provide the proofs of \cref{thm:hybrid-2vass-hybrid}, establishing the corresponding results for geometrically 2-dimensional 3-VASS with source hybrid sets.

\subsection{2-SLPS and the Reachability Sets}\label{sec:repr-2vass-known}

A (2-dimensional) \emph{simple linear path scheme (SLPS)} is a regular expression of the form 
\begin{equation}
    \Lambda = \alpha_0 \beta_1^* \alpha_1\cdots \beta_k^* \alpha_k
    \label{eq:slps}
\end{equation}
where each $\alpha_i \in (\ZZ^2)^*$ is a sequence of vectors called a \emph{bridge}, and each $\beta_i \in \ZZ^2$ is called a \emph{loop}. Naturally, an SLPS $\Lambda$ of the form in \cref{eq:slps} defines a 2-VASS $V_\Lambda$ with states $q_0, \ldots, q_{k+1}$ such that there are paths from $q_i$ to $q_{i+1}$ whose transitions have effects specified by $\alpha_i$ for $i = 0, \ldots, k$, and there are self-loops $q_i\xrightarrow{\beta_i}q_i$ for $i = 1, \ldots, k$. For a bridge $\alpha_i = \Vec{a}_1\ldots\Vec{a}_m$, its effect is $\Eff{\alpha_i} = \sum_{j\in[m]}\Vec{a}_j$. We write $\Drop(\alpha_i) \in \NN^2$ for the drop of the corresponding path in $V_\Lambda$. 
ence, $\Drop(\alpha_i)(c) = -\min\{0, \min_{j\in[m]}\{\Vec{a}_1(c) + \cdots + \Vec{a}_j(c)\}\}$ for every $c\in[d]$. For two vectors $\Vec{x}, \Vec{y} \in \NN^2$, we write $\Vec{x} \xrightarrow{\Lambda} \Vec{y}$ if $q_0(\Vec{x}) \xrightarrow{*} q_{k+1}(\Vec{y})$ holds in $V_\Lambda$. We also define $\Size(\Lambda) := \Size(V_\Lambda)$ and similarly $\Cone(\Lambda) = \Cone(V_{\Lambda}) = \Cone\{\beta_1, \ldots, \beta_k\}$. The norm of $\Lambda$ is defined to be the maximum norm of bridges and loops: $\norm{\Lambda} = \max\{\norm{\alpha_0}, \norm{\beta_1}, \ldots, \norm{\beta_k}, \norm{\alpha_k}\}$.

\begin{lemma}[{\cite[Theorem 3.1 and Lemma 5.3]{DBLP:journals/jacm/BlondinEFGHLMT21}}]
    \label{lem:2vass-to-2slps}
    For every 2-VASS $V$ and two states $p, q$, there is a finite set $\mathcal{S}$ of SLPSs such that $p(\Vec{x}) \xrightarrow{*} q(\Vec{y})$ in $V$ if and only if $\Vec{x} \xrightarrow[]{\Lambda} \Vec{y}$ for some $\Lambda \in \mathcal{S}$. Moreover, for every $\Lambda \in \mathcal{S}$ we have $\Size(\Lambda) \le \PolyF{\Size(V)}$.
\end{lemma}

Simple linear path schemes admit further normalizations. We first introduce three special forms:
\begin{itemize}
    \item an SLPS consisting of at most three loops is called a \emph{short} SLPS;
    \item an SLPS of the form $\alpha_0 \beta_1^* \alpha_1 \beta_2^*$ such that $\beta_1 \in \NN_{>0} \times \NN_{<0}$ and $\beta_2 \in \NN_{<0} \times \NN_{>0}$ is called a \emph{one-turn} SLPS;
    \item a concatenation of one-turn SLPSs is called a \emph{zigzagging} SLPS.
\end{itemize}

Consider a run in a one-turn SLPS $\Lambda = \alpha_0 \beta_1^* \alpha_1 \beta_2^*$ given by
\begin{equation*}
    \pi: \Vec{x} \xrightarrow{\alpha_0} \Vec{x}_1 \xrightarrow{\beta_1^{n_1}} \Vec{y}_1 
    \xrightarrow[]{\alpha_1} \Vec{x}_2 \xrightarrow[]{\beta_2^{n_2}} \Vec{y}.
\end{equation*}
We say that $\pi$ is {\em $B$-vertical-returning} or simply {\em $B$-returning}, if $\Vec{x}, \Vec{y} \in [0, B] \times \NN$. Likewise, a run $\rho$ in a zigzagging SLPS is $B$-returning if $\rho$ is a sequence of $B$-returning one-turns. We write $\Vec{x} \BRetTo_{B}^{\Lambda}\; \Vec{y}$ to indicate that there exists a $B$-returning run in $\Lambda$ from $\Vec{x}$ to $\Vec{y}$. 
Moreover, for an SLPS $\Lambda = \alpha_0 \beta_1^* \alpha_1\cdots \beta_k^* \alpha_k$, we say that an SLPS $\Lambda'$ is a \emph{detailing} of $\Lambda$ if $\Lambda'$ is obtained from $\Lambda$ by replacing some $\beta_i^*$ by $\beta_i^{n_i}$ for some $n_i \in \NN$. Clearly the reachability relation of $\Lambda'$ is contained in that of $\Lambda$.

\begin{theorem}[{\cite[Theorem 36]{DBLP:conf/icalp/CzerwinskiJ0O25}}]
    \label{thm:filip}
    Let $\Lambda$ be a 2-dimensional SLPS. There exists $B \le \PolyF{\Size(\Lambda)}$ such that for every run $\Vec{x} \xrightarrow[]{\Lambda} \Vec{y}$ there exists a detailing $\Lambda' = \Lambda_1 \Lambda_2 \Lambda_3$ of $\Lambda$ with $\Size(\Lambda') \le \PolyF{\Size(\Lambda)}$, such that $\Lambda_1$ and $\Lambda_3$ are short, and
    \begin{itemize}
        \item either $\Lambda_2$ is empty, and $\Vec{x} \xrightarrow{\Lambda_1 \Lambda_3} \Vec{y}$;
        \item or $\Lambda_2$ is zigzagging, and $\Vec{x} \xrightarrow{\Lambda_1} \Vec{u} \BRetTo_B^{\Lambda_2} \Vec{v} \xrightarrow{\Lambda_3} \Vec{y}$ for some $\Vec{u}, \Vec{v} \in [0, B]\times \NN$.
    \end{itemize}
\end{theorem}

Notice that for each SLPS $\Lambda$ there are only finitely many detailings $\Lambda_1 \Lambda_2 \Lambda_3$ satisfying the size bound in \cref{thm:filip}.
Hence, in the remainder of this section, we focus on the representation of reachability sets of such normalized SLPSs.

For convenience, we denote by $\ReachVs{\Lambda}{\Vec{x}}$ the set of vectors $\Vec{y}\in \NN^2$ such that $\Vec{x} \xrightarrow{\Lambda} \Vec{y}$. Moreover, we sometimes investigate the reachability set $\ReachVs{\Lambda}{X}$ for some set $X\subseteq \NN^2$, i.e., $\ReachVs{\Lambda}{X} := \bigcup_{\Vec{x}\in X} \ReachVs{\Lambda}{\Vec{x}}$. Several useful tools for representing reachability sets in a 2-SLPS have been established in \cite{DBLP:conf/icalp/CzerwinskiJ0O25}; we recall them below. We define the \emph{$(a, p, K)$-arithmetic set}, where $a, p \in \NN$ and $K \in \NN_\infty$, to be the set $a + p^{\le K} := \{a + n \cdot p \mid n \le K\}$.

\begin{lemma}[{\cite[Claim 37]{DBLP:conf/icalp/CzerwinskiJ0O25}}]
    \label{lem:first-short-orig}
    For all short SLPS $\Lambda$, configuration $\Vec{s}$ and $x \in [0, B]$, the set 
    $\{\Vec{t}(2) \mid \Vec{s} \xrightarrow{\Lambda} \Vec{t}, \Vec{t}(1) = x\}$
    is a finite union of $(a, p, K)$-arithmetic sets where $a \le \PolyF{B, M}, p \le \PolyF{M}$ and $M = \Size(\Lambda, \Vec{s})$.
\end{lemma}

\begin{lemma}[{\cite[Claim 38]{DBLP:conf/icalp/CzerwinskiJ0O25}}]
    \label{lem:zigzag-orig}
    Let $\Lambda$ be a one-turn SLPS, $S = a + p^{\le K}$ be an arithmetic set, and $x, x' \in [0,B]$. The set 
    \begin{equation*}
        \left\{ y' \mid (x, y) \BRetTo_B^\Lambda (x', y')\text{ and } y \in S \right\}
    \end{equation*}
    is a finite union of $(a', p', K')$-arithmetic sets with $a' \le a + \PolyF{B, M, p}$, $p' \le \max\{p, \PolyF{M}\}$ and $M = \Size(\Lambda)$.
\end{lemma}

\begin{lemma}[{\cite[Claim 39]{DBLP:conf/icalp/CzerwinskiJ0O25}}]
    \label{lem:last-short-orig}
    Let $\Lambda$ be a short SLPS. The set $\ReachVs{\Lambda}{\Vec{s} + \{\Vec{p}\}^{\le K}}$ is a finite union of hybrid sets $S = \Vec{b} + \ResPeriodic{P}{\Vec{c}}{K'}$ such that $\NORM{S}$ is bounded by $\PolyF{\Size(\Lambda), \norm{\Vec{s}}, \norm{\Vec{p}}}$.
\end{lemma}

\subsection{Reachability Sets in 2-SLPS with Hybrid Set Source}\label{sec:repr-2vass-hybrid}

We intend to generalize \cref{thm:2-vass-approximation} from singleton sources to hybrid set sources.
It is not difficult to see that the key is to generalize \cref{lem:first-short-orig}.
This is the next lemma. 

\begin{lemma}
    \label{lem:first-short-new}
    Let $S:=\Vec{b} + \ResPeriodic{P}{\Vec{c}}{K}$.
    For every short SLPS $\Lambda$ and $x \in [0, B]$, the set 
    $\{\Vec{t}(2) \mid \Vec{s} \xrightarrow{\Lambda} \Vec{t}, \Vec{s} \in S, \Vec{t}(1) = x\}$
    is a finite union of $(a, p, K')$-arithmetic sets such that 
    $$a, p \le \PolyF{\Size(\Lambda), \NORM{S}, B}.$$
\end{lemma}

To prove~\cref{lem:first-short-new}, we prove a stronger lemma that handles the case where the source is a hybrid set and the target is any linear set.

\begin{lemma}
    \label{lem:hybrid-slps-linear}
    Let $\Lambda = \alpha_0\beta_1^*\alpha_1 \ldots \beta_k^*\alpha_k$.
    Let $S = \Vec{b} + \ResPeriodic{P}{\Vec{c}}{K}$ be a hybrid set, and let $T = \Vec{d} + Q^*$ be a linear set.
    Then $\ReachVs{\Lambda}{S} \cap T$ is a finite union of hybrid sets $S' = \Vec{b}' + \ResPeriodic{P'}{\Vec{c'}}{K'}$ where $\NORM{S'} \le \PolyF{\Size(\Lambda), \NORM{S}, \NORM{T}}^k$.
\end{lemma}

\begin{proof}
    It is well known that every run in $\Lambda$ starting from a vector in $S$ is captured by a system of linear inequalities \cite{DBLP:conf/lics/BlondinFGHM15}. We introduce variables $n_1, \ldots, n_k \in \NN$ for the number of repetitions of each loop, and $\Vec{a}_i, \Vec{b}_i$ where $i = 0, 1, \ldots, k$ for the intermediate configurations before and after the bridge $\alpha_i$. Also, we introduce a variable $\Vec{\ell} \in \NN^{|P|}$ to determine the source, and a variable $\Vec{r} \in \NN^{|Q|}$ to determine the target in $T$. The constraints are as follows, where $P$ and $Q$ are the matrices whose columns are vectors in $P$ and $Q$ respectively:
    \begin{align*}
    \begin{cases}
        \Vec{b} + P \cdot \Vec{\ell} & = \Vec{a}_0 \\
        \Vec{d} + Q \cdot \Vec{r} &= \Vec{b}_k \\
        \Vec{a}_i & \ge \Drop(\alpha_i) \qquad i = 0, 1, 2, \ldots, k\\
        \Vec{a}_i + \Eff{\alpha_i} & = \Vec{b}_i \qquad i = 0, 1, 2, \ldots, k\\
        \Vec{b}_{i-1} + \beta_i \cdot n_i & = \Vec{a}_{i} \qquad i = 1, 2, \ldots, k
    \end{cases}
    \end{align*}
    Observe that the non-negative solutions of this system are in one-to-one correspondence with runs from $\Vec{a}_0 = \Vec{b} + P \Vec{\ell}$ to $\Vec{b}_k = \Vec{d} + Q\Vec{r}$: the variables $\Vec{a}_i$ and $\Vec{b}_i$ record the configurations immediately before and after the bridge $\alpha_i$, respectively.
    Let $\mathcal{E}$ be the above system, and $\mathcal{E}_0$ be its homogeneous version. 
    For a solution $\Vec{h}$ of $\mathcal{E}$, we write for example $\Vec{\ell}^{\Vec{h}}$ for the value assigned to $\Vec{\ell}$ by $\Vec{h}$.
    In the following, we also refer to the variable $\Vec{b}_k$ by the name $\Vec{t}$ as it is the target of the whole run. Now we have
    \begin{equation*}
        \{\Vec{t}^{\Vec{h}} \mid \Vec{h} \text{ is a solution to } \mathcal{E}\} = \ReachVs{\Lambda}{\Vec{b} + P^*} \cap T.
    \end{equation*}
    The number of constraints of $\mathcal{E}$ is $m = 4 + (3k + 2) \times 2 = \bigO(k)$, and the number of variables is $n = |P| + |Q| + k + 2 \times (k+1) \times 2 = |P| + |Q| + \bigO(k)$. Notice that $|P| \le (\norm{P} + 1)^2$ and $|Q| \le (\norm{P} + 1)^2$, so we have $n \le \PolyF{\norm{P}, \norm{Q}, \Size(\Lambda)}$. Also, the norm of coefficients of $\mathcal{E}$ is bounded by $A := \norm{\Lambda} + \norm{P} + \norm{Q}$ and the norm of constant terms of $\mathcal{E}$ is bounded by $C := \norm{\Vec{b}} + \norm{\Vec{d}} + \Size(\Lambda)$. Towards applying \cref{cor:ilp-decomp}, define 
    \begin{align*}
        M & := ((m + n + 1) A + C + 1)^m = \PolyF{\Size(\Lambda), \NORM{S}, \NORM{T}}^k;\\
        M_0 & := ((m + n) A + 1)^m \le \PolyF{\Size{\Lambda}, \NORM{S}, \NORM{T}}^k.
    \end{align*}
    Denote by $X$ the solutions to $\mathcal{E}$ and by $X_0$ the solutions to $\mathcal{E}_0$. Also let $\overline{X} = \{\Vec{h} \in X \mid \norm{\Vec{h}} \le M\}$ and $\overline{X}_0 = \{\Vec{h}_0 \in X_0 \mid \norm{\Vec{h}_0} \le M_0\}$. We have $X = \overline{X} + (\overline{X}_0)^*$ by \cref{cor:ilp-decomp}. 

    Next, let
    \begin{align*}
        \overline{Y} := \{ \Vec{h} \in \overline{X} \mid \Inner{\Vec{c}, \Vec{\ell}^{\Vec{h}}} \le K \}.
    \end{align*}
    Also, enumerate $\overline{X}_0$ as $\{\Vec{g}_1, \ldots, \Vec{g}_h\}$. We define
    \begin{align*}
        P' := \{ \Vec{t}^{\Vec{g}_i} \mid i = 1, \ldots, h \}, \quad \Vec{c}' := (\Inner{\Vec{c}, \Vec{\ell}^{\Vec{g}_1}}, \dots, \Inner{\Vec{c}, \Vec{\ell}^{\Vec{g}_h}}).
    \end{align*}
    It is routine to check that
    \begin{align}
    \label{eq:reach-const-short}
        \bigcup_{\Vec{h} \in \overline{Y}} \Vec{t}^{\Vec{h}} + \ResPeriodic{P'}{\Vec{c}'}{K - \Inner{\Vec{c}, \Vec{\ell}^{\Vec{h}}}}
        =
        \ReachVs{\Lambda}{S} \cap T.
    \end{align}
    We are left to bound the norm of the left-hand side. Notice that 
    \begin{align*}
        \norm{\Vec{t}^{\Vec{h}}} &\le \norm{\overline{Y}} \le \norm{\overline{X}} \le M \le \PolyF{\Size(\Lambda), \NORM{S}, \NORM{T}}^k;\\
        \norm{P'} & \le \norm{\overline{X_0}} \le M_0 \le \PolyF{\Size(\Lambda), \NORM{S}, \NORM{T}}^k;\\
        \norm{\Vec{c}'} &\le \max_{i \in [h]}|\Inner{\Vec{c}, \Vec{\ell}^{\Vec{g}_i}}| \le 3 \norm{\Vec{c}} M_0 \le \PolyF{\Size(\Lambda), \NORM{S}, \NORM{T}}^k.
    \end{align*}
    Hence, we deduce that the norm of each resulting hybrid set in \cref{eq:reach-const-short} is also bounded by $\PolyF{\Size(\Lambda), \NORM{S}, \NORM{T}}^k$.
\end{proof}

Now we come back to the proof of \cref{lem:first-short-new}.

\begin{proof}[Proof of \cref{lem:first-short-new}]
    Let $T := (x, 0) + \{(0, 1)\}^*$ be the target set we are considering. Applying \cref{lem:hybrid-slps-linear} to $\Lambda$ where the number of loops is $k \le 3$, we represent $\ReachVs{\Lambda}{S} \cap T$ as a finite union of hybrid sets $S' = \Vec{b}' + \ResPeriodic{P'}{\Vec{c'}}{K'}$, where $\NORM{S'} \le \PolyF{\Size(\Lambda), \NORM{S}, B}$.
    Since $S' \subseteq T$, we must have $\Vec{b}' = (x, a_0)$ for some $a_0 \in \NN$, and $P' \subseteq \{(0, 1)\}^*$. Let $P'' = \{y \mid (0, y) \in P'\}$, the set
    \begin{equation*}
        \{\Vec{t}(2) \mid \Vec{s} \xrightarrow{\Lambda} \Vec{t}, \Vec{s} \in \Vec{b} + \ResPeriodic{P}{\Vec{c}}{K} , \Vec{t}(1) = x\}
    \end{equation*}
    is a finite union of hybrid sets $a_0 + \ResPeriodic{P''}{\Vec{c'}}{K'}$.
    It remains to represent these hybrid sets as a finite union of arithmetic sets.

    \begin{claim}
        \label{clm:one-dim-res-per-set-reduction}
        Let $a \in \NN$, $P \subseteq \NN_{>0}$, $\Vec{c} \in \NN^{|P|}$, and $K \in \NN$. If $|P| \ge 2$ then the set $a + \ResPeriodic{P}{\Vec{c}}{K}$ is a finite union of sets $a' + \ResPeriodic{P'}{\Vec{c}'}{K'}$ where $a' \le a + \norm{P}^2$, $P' \subsetneq P$, and $\norm{\Vec{c}'} \le \norm{\Vec{c}}$.
    \end{claim}
    \begin{claimproof}
        Suppose $P = \{p_1, \ldots, p_m\}$ with $m \ge 2$. Then any $v \in a + \ResPeriodic{P}{\Vec{c}}{K}$ can be written as
        \begin{equation*}
            v = a + n_1 p_1 + n_2 p_2 + \cdots + n_m p_m
        \end{equation*}
        satisfying $t := n_1 \Vec{c}(1) + n_2 \Vec{c}(2) + \cdots + n_m \Vec{c}(m) \le K$. We claim that there is always such a representation with either $n_1 < p_2$ or $n_2 < p_1$. Which bound holds depends on the sign of $-\Vec{c}(1) p_2 + \Vec{c}(2) p_1$.
        
        Assume $-\Vec{c}(1) p_2 + \Vec{c}(2) p_1 \le 0$ (the other case is symmetric). We claim that $n_1 < p_2$ can be achieved. Indeed, if $n_1 \ge p_2$ then observe that 
        \begin{equation*}
            v = a + (n_1 - p_2) p_1 + (n_2 + p_1) p_2 + n_3 p_3 + \cdots + n_m p_m
        \end{equation*}
        is another representation of $v$ with non-negative coefficients. Moreover, 
        \begin{equation*}
            (n_1 - p_2) \Vec{c}(1) + (n_2 + p_1) \Vec{c}(2) + \cdots + n_m \Vec{c}(m) = t -\Vec{c}(1) p_2 + \Vec{c}(2) p_1 \le t \le K.
        \end{equation*}
        As the choice of $v$ is arbitrary, it is justified that
        \begin{equation*}
            a + \ResPeriodic{P}{\Vec{c}}{K} = \bigcup_{n_1 < p_2} a + n_1 p_1 + \ResPeriodic{P'}{\Vec{c}'}{K - n_1 \Vec{c}(1)},
        \end{equation*}
        where $P' = \{p_2, \ldots, p_m\}$ and $\Vec{c}' = (\Vec{c}(2), \ldots, \Vec{c}(m))$. Finally, notice that $a + n_1p_1 \le a + p_2p_1 \le a + \norm{P}^2$.
    \end{claimproof}

    \begin{claim}
        \label{clm:one-dim-res-per-set-final}
        Let $a \in \NN$, $P \subseteq \NN_{>0}$, $\Vec{c} \in \NN^{|P|}$, and $K \in \NN$. The set $a + \ResPeriodic{P}{\Vec{c}}{K}$ is a finite union of $(a', p', K')$-arithmetic sets with $a' \le a + |P| \cdot \norm{P}^2$ and $p' \le \norm{P}$.
    \end{claim}
    \begin{claimproof}
        Apply \cref{clm:one-dim-res-per-set-reduction} inductively to $a + \ResPeriodic{P}{\Vec{c}}{K}$, we obtain a finite union of hybrid sets $S = \Vec{a}' + \ResPeriodic{\{p'\}}{c'}{K'}$ with $a' \le a + |P| \cdot \norm{P}^2$ and $p' \le \norm{P}$. Notice that $S$ can be represented as an $(a', p', \lfloor K'/c \rfloor)$-arithmetic set, where $K'/c$ is understood as $\infty$ when $c = 0$.
    \end{claimproof}

    Now the one-dimensional hybrid set $a_0 + \ResPeriodic{P''}{\Vec{c}'}{K'}$ can be represented as a finite union of $(a, p, K)$-arithmetic sets with 
    \begin{align*}
            a &\le a_0 + |P''| \norm{P''}^2 \le \norm{\Vec{b}'} + |P'| \norm{P'}^2 \le \PolyF{\Size(\Lambda), \NORM{S}, B},\\
            p & \le \norm{P'} \le \PolyF{\Size(\Lambda), \NORM{S}, B}.
    \end{align*}
The proof of~\cref{lem:first-short-new} is completed. 
\qedhere
\end{proof}

Now we have \cref{lem:first-short-new} together with \cref{lem:zigzag-orig,lem:last-short-orig}, which is enough to handle an SLPS that has a non-empty zigzagging infix (the second case of \cref{thm:filip}). For the other case, notice that the SLPS comprises a constant number of loops, so the bound in \cref{lem:hybrid-slps-linear} works fine.

\begin{lemma}
    \label{lem:2slps-rep-new}
    Let $\Lambda$ be a 2-dimensional SLPS. Let $S = \Vec{s} + \ResPeriodic{P}{\Vec{c}}{K}$ be a hybrid set. Then $\Reach(\Lambda, S)$ is a finite union of hybrid sets $S' = \Vec{b} + \ResPeriodic{P'}{\Vec{c}'}{K'}$ where $\NORM{S'} \le \PolyF{\Size(\Lambda), \NORM{S}}$. Moreover, $P' \subseteq \Cone(P) + \Cone(\Lambda)$.
\end{lemma}

\begin{proof}[Proof of \cref{lem:2slps-rep-new}]
    By \cref{thm:filip}, every run in $\Lambda$ can be represented by a detailing $\Lambda' = \Lambda_1\Lambda_2\Lambda_3$ where $\Lambda_1$ and $\Lambda_3$ are short, $\Lambda_2$ is zigzagging, and $\Size(\Lambda') \le \PolyF{\Size(\Lambda)}$. 
    If $\Lambda_2$ is empty, then we are done by applying \cref{lem:hybrid-slps-linear} to $\Lambda_1\Lambda_3$ with target set $T := \NN^2$.
    Otherwise, assume $\Lambda_2$ is nonempty.
    Any run starting from $\Vec{x}$ is then captured by a run in the form of
    \begin{align*}
        \Vec{x} \xrightarrow{\Lambda_1} \Vec{u} \BRetTo^{\Lambda_2}_{B} \Vec{v} \xrightarrow{\Lambda_3} \Vec{y}
    \end{align*}
    where $B \le \PolyF{\Size(\Lambda)}$ and $\Vec{u}, \Vec{v} \in [0, B] \times \NN$.
    
    Let $S = \Vec{s} + \ResPeriodic{P}{\Vec{c}}{K}$. For the detailing $\Lambda' = \Lambda_1\Lambda_2\Lambda_3$, we first apply \cref{lem:first-short-new} and obtain that $\Reach(\Lambda_1, S) \cap ([0, B] \times \NN)$ is a finite union of sets $\Vec{a}_1 + \{\Vec{p}_1\}^{\le K_1}$, where $\Vec{a}_1 = (x, a)$ and $\Vec{p}_1 = (0, p)$ satisfy: (i) $x \in [0, B]$, (ii) $a, p \le \PolyF{\Size(\Lambda), \NORM{S}}$.

    For each such hybrid set, we apply \cref{lem:zigzag-orig} to each one-turn segment of $\Lambda_2$. Since the number of one-turn segments in $\Lambda_2$ is bounded by $\Size(\Lambda_2) \le \PolyF{\Size(\Lambda)}$, we conclude that $\Reach(\Lambda_1\Lambda_2, S) \cap ([0, B] \times \NN)$ is a finite union of hybrid sets $\Vec{a}_2 + \{\Vec{p}_2\}^{\le K_2}$, where $\Vec{a}_2 = (x, a')$ and $\Vec{p}_2 = (0, p')$ satisfy: (i) $x \in [0, B]$, (ii) $a' \le a + \Size(\Lambda_2) \cdot \PolyF{B, \Size(\Lambda_2), p} \le \PolyF{\Size(\Lambda), \NORM{S}}$, and (iii) $p' \le \PolyF{\Size(\Lambda), \NORM{S}}$.

    Finally, we apply \cref{lem:last-short-orig} to $\Lambda_3$ with each of the hybrid sets $\Vec{a}_2 + \{\Vec{p}_2\}^{\le K_2}$ as a source. This shows that $\Reach\{\Lambda_1\Lambda_2\Lambda_3, S\}$ is a finite union of hybrid sets $S' = \Vec{b} + \ResPeriodic{P'}{\Vec{c}'}{K'}$, where $\NORM{S'} \le \PolyF{\Size(\Lambda), \norm{\Vec{a}_2}, \norm{\Vec{p}_2}} \le \PolyF{\Size(\Lambda), \NORM{S}}$.

    To see that the periods belong to $\Cone(P) + \Cone(\Lambda)$, one needs to go through in detail the proof of \cref{lem:first-short-new,lem:zigzag-orig,lem:last-short-orig,lem:hybrid-slps-linear}. It can be verified that periods constructed there correspond to some non-negative homogeneous solutions of the characterization system of $\Lambda$. Hence, they are non-negative combinations of vectors from $P$ and loops in $\Lambda$.
\end{proof}

\begin{lemma}
    \label{lem:2vass-repr}
    Let $V$ be a 2-VASS and $p, q$ be two states. Let $S = \Vec{s} + \ResPeriodic{P}{\Vec{c}}{K}$ be a hybrid set. Then $\ReachqVs{q}{V}{p(S)}$ is a finite union of hybrid sets $S' = \Vec{b} + \ResPeriodic{P'}{\Vec{c}'}{K'}$ where $\NORM{S'} \le \PolyF{\Size(V), \NORM{S}}$. Moreover, $P' \subseteq \Cone(P) + \Cone(V)$.
\end{lemma}

\begin{proof}
    By \cref{lem:2vass-to-2slps}, reachability in $V$ from $p$ to $q$ is characterized by finitely many SLPSs $\Lambda$ of size $\Size(\Lambda) \le \PolyF{\Size(V)}$. Hence, the polynomial-sized representation of $\ReachqVs{q}{V}{p(S)}$ can be derived from \cref{lem:2slps-rep-new}. In order to justify that $P' \subseteq \Cone(P) + \Cone(V)$, we need to show that every loop in $\Lambda$ belongs to $\Cone(V)$. This could be verified by looking carefully into the proof of \cref{lem:2vass-to-2slps}. Here, we also provide a black-box proof of this fact.

    Let $\Lambda = \alpha_0 \beta_1 \alpha_1 \ldots \beta_k \alpha_k$ be an SLPS satisfying that $\Vec{x} \xrightarrow{\Lambda} \Vec{y}$ implies $p(\Vec{x}) \xrightarrow{*} q(\Vec{y})$ in $V$. We show that $\beta_j \in \Cone(V)$ for any $j \in [k]$. Indeed, one can easily see that for all $\ell \in \NN$, there exists a vector $\Vec{s}_\ell \in \NN^2$ sufficiently large, so that
    \begin{equation*}
        \Vec{s}_\ell \xrightarrow{\alpha_0\alpha_1\ldots\alpha_{j-1}\beta_j^{\ell}\alpha_j\ldots\alpha_k} \Vec{t}_\ell = \Vec{s}_{\ell} + \ell \cdot {\beta_j} + \Eff{\alpha_0\alpha_1\ldots\alpha_k}.
    \end{equation*}
    Hence, there also exists a path $\pi_\ell$ in $V$ such that $p(\Vec{s}_\ell) \xrightarrow{\pi_\ell} q(\Vec{t}_\ell)$ for $\ell \in \NN$, where $\Eff{\pi_\ell} = \ell \cdot {\beta_j} + \Eff{\alpha_0\alpha_1\ldots\alpha_k}$. Denote $\mathcal{I}_{\pi_\ell}$ the \emph{Parikh image} of $\pi_\ell$, which is the function mapping each transition in $V$ to its number of occurrences along $\pi_\ell$. We may view $\mathcal{I}_{\pi_\ell}$ as a vector comparable by the well-quasi-order $\le$. So there exist $m < n \in \NN$ such that $\mathcal{I}_{\pi_m} \le \mathcal{I}_{\pi_n}$. Their difference $\mathcal{I} = \mathcal{I}_{\pi_n} - \mathcal{I}_{\pi_m}$ maps each transition to a non-negative number. Observe that for each state $s$ in $V$, $\mathcal{I}$ specifies the same number of transitions flowing in $s$ as those flowing out of $s$. Hence, $\mathcal{I}$ is the sum of Parikh images of finitely many cycles in $V$. Moreover, the ``effect'' of $\mathcal{I}$ is 
    \begin{equation*}
    \begin{aligned}
        \Eff{\mathcal{I}} := \sum_{\text{transition }u} \Eff{u} \cdot \mathcal{I}(u) & = \sum_{\text{transition }u} \Eff{u} \cdot \mathcal{I}_{\pi_n}(u) - \sum_{\text{transition }u} \Eff{u} \cdot \mathcal{I}_{\pi_m}(u)\\
        &= \Eff{\pi_n} - \Eff{\pi_m} = (n-m) \cdot {\beta_j}.
    \end{aligned}
    \end{equation*}
    This justifies that $\beta_j$ belongs to $\Cone(V)$.
\end{proof}

\subsection{Proof of \cref{thm:hybrid-2vass-hybrid}}\label{sec:repr-2vass-Geo-hybrid}

\ThmHybridTwoVASSHybrid*

\begin{proof}
Let $Q, T$ be the states and transitions of $V$. Let $U := \CycleSpace(V)$. A \emph{normal vector} of $U$ is a vector $\Vec{n} \in \ZZ^3$ such that $\Inner{\Vec{n}, \Vec{u}} = 0$ for all $\Vec{u} \in U$. We consider two cases.

\begin{description}
    \item[Case 1. $U$ has a non-zero normal vector $\Vec{n} \in \ZZ_{\ge0}^3$.]  Consider a run $p(\Vec{s} + \Vec{p}) \xrightarrow{\pi} q(\Vec{t})$ where $\Vec{p} \in \ResPeriodic{P}{\Vec{c}}{K}$. By extracting cycles from $\pi$ exhaustively, it can be justified that $\Eff{\pi}$ is decomposed into $\Vec{z} + \Vec{u}$ where $\Vec{u} \in U$ and $\Vec{z}$ is the effect of some simple path. Hence,
    \begin{align*}
        \Inner{\Vec{n}, \Vec{t}} & = \Inner{\Vec{n}, \Vec{s} + \Vec{p}} + \Inner{\Vec{n}, \Eff{\pi}} 
        = \Inner{\Vec{n}, \Vec{s}} + \Inner{\Vec{n}, \Vec{z}}
    \end{align*}
    where we use the fact that $\Vec{p} \in P^* \subseteq \CycleSpace(V)$. By \cref{lem:pottier}, we may assume a polynomial bound on the norm of the normal vector $\norm{\Vec{n}} \le \PolyF{\Size(V)}$. Thus,
    \begin{align*}
        \Inner{\Vec{n}, \Vec{t}} \le 3\norm{\Vec{n}} \cdot (\norm{\Vec{s}} + \norm{\Vec{z}}) \le \PolyF{\Size(V), \norm{\Vec{s}}}.
    \end{align*}
    As $\Vec{n} \in \ZZ_{\ge0}^3$, let $K = \Supp(\Vec{n})$, we deduce that $\norm{\Vec{t}|_K} \le \PolyF{\Size(V), \norm{\Vec{s}}} =: B$ as well. This shows that any configuration reachable from $p(S)$ has bounded counter values for counters in $K$. We may encode counters in $K$ into the states of a VASS. To be explicit, choose $i \in K$ and construct a 2-VASS $V_S = (Q_S, T_S)$ defined by 
    \begin{itemize}
        \item $Q_S = \{q_v \mid q \in Q, v \in [0, B]\}$;
        \item $T_S = \{p_u \xrightarrow{\Vec{a}|_{-i}} q_v \mid p \xrightarrow{\Vec{a}} q \in T, \Vec{a}(i) = v - u \}$.
    \end{itemize}
    Clearly $\Size(V_S) \le \PolyF{B,\Size(V)} \le \PolyF{\Size(V), \norm{\Vec{s}}}$. Let $p' := p_u$ where $u = \Vec{s}(i)$. Notice that for any $\Vec{p} \in P^*$ we must have $\Vec{p}(i) = 0$ as $\Inner{\Vec{n}, \Vec{p}} = 0$, and $P \subseteq \NN^3$, $\Vec{n} \ge \Vec{0}$. We deduce that 
    \begin{equation*}
        \ReachqVs{q}{V}{p(S)} = \bigcup_{v \in [0, B]} \{\Vec{t} \in \NN^3 \mid \Vec{t}(i) = v, \Vec{t}|_{-i} \in \ReachqVs{q_v}{V_S}{p_u(S|_{-i})}\}
    \end{equation*}
    By \cref{lem:2vass-repr}, each $\ReachqVs{q_v}{V_S}{p_u(S|_{-i})}$ is a finite union of hybrid sets $S' = \Vec{b'} + \ResPeriodic{P'}{\Vec{c}'}{K'}$. We may lift $S'$ to three dimensions by adding the $i$-th coordinate back to $\Vec{b}'$ with value $v$ and to $P'$ with value $0$. As $P' \subseteq \Cone(P|_{-i}) + \Cone(V_S)$, we deduce that the lifted version is also contained in $\Cone(P) + \Cone(V)$.

    \item[Case 2. Every normal vector of $U$ contains two different signs.] 
    Take a non-zero normal vector $\Vec{n} \in \ZZ^3$ with $\norm{\Vec{n}} \le \PolyF{\Size(V)}$. We may assume $\Vec{n} = (a, b, -c)$ for some $a, b, c \in \NN$ and $c\ne 0$. For a vector $\Vec{x} \in \NN^3$ we define its \emph{shift} as 
    \begin{equation*}
        \delta(\Vec{x}) = \Inner{\Vec{n}, \Vec{x}} - \Inner{\Vec{n}, \Vec{s}}
    \end{equation*}
    Similar to case 1, for any configuration $q(\Vec{t})$ reachable from $p(S)$, we have $|\delta(\Vec{t})| \le B$ for some $B \le \PolyF{\Size(V)}$. Thus, we encode shifts into states and ignore the third counter. To be explicit, we construct a 2-VASS $V_S = (Q_S, T_S)$ where 
    \begin{itemize}
        \item $Q_S$ contains $q_v, q_v'$ for each $q \in Q$ and $v \in [-B, B]$;
        \item $T_S$ contains $p_u \xrightarrow{\Vec{a}|_{-3}} q_v'$ if there is a transition $p \xrightarrow{\Vec{a}} q$ in $T$ such that $\delta(\Vec{a}) = v - u$.
    \end{itemize}
    A configuration $q_v(x, y)$ is supposed to represent the configuration $q(x, y, z)$ in $V$ such that $\delta(x, y, z) = v$. Let $z(x, y, v) = (ax + by - v - \Inner{\Vec{n}, \Vec{s}}) / c$, which is supposedly the missing counter value in $q_v(x, y)$. One can verify that if $z(x, y, v)$ is an integer, then every configuration $p_u(x', y')$ reachable from $q_v(x, y)$ renders $z(x', y', u)$ an integer.
    To check non-negativity of $z$, we add a gadget from each $q_v'$ to $q_v$ that tests the inequality $ax + by -v \ge \Inner{\Vec{n}, \Vec{s}}$. The minimal solutions to this inequality are finite and bounded by $\PolyF{\Size(V), \norm{\Vec{s}}}$, which creates a polynomial amplification in $\Size(V_S)$. Verify that
    \begin{equation*}
        \ReachqVs{q}{V}{p(S)} = \bigcup_{v \in [-B, B]} \{(x, y, z(x, y, v)) \mid (x, y) \in \ReachqVs{q_v}{V_S}{p_0(S|_{-3})}\}.
    \end{equation*}
    Again by \cref{lem:2vass-repr}, each $\ReachqVs{q_v}{V_S}{p_0(S|_{-3})}$ is a finite union of hybrid sets $\Vec{b}' + \ResPeriodic{P'}{\Vec{c}'}{K'}$. We lift it to three dimensions by mapping $\Vec{b}' = (x, y)$ to $\Vec{b}'' = (x, y, z(x, y, v))$, and mapping each $\Vec{p}' = (x, y) \in P'$ to $\Vec{p}'' = (x, y, (ax + by) / c)$. We need to justify that the lifted vectors are integral.
    Notice that for each $(x_s, y_s) \in S|_{-3}$ there exists $z_s \in \NN$ such that $(x_s, y_s, z_s) \in S$. We must have $z(x_s, y_s, 0) = z_s \in \NN$. As $q_v(\Vec{b}')$ is reachable from $p_0(x_s, y_s)$ for some $(x_s, y_s) \in S|_{-3}$, we deduce that $z(\Vec{b}'(1), \Vec{b}'(2), v) = \Vec{b}''(3)$ is integral as well. Similarly, for each $\Vec{p}' \in P'$, $q_v(\Vec{b}'+\Vec{p}')$ is reachable from $p_0(x_s, y_s)$ for some $(x_s, y_s) \in S|_{-3}$ (we may assume each period in $P'$ can be used at least once), so $\Vec{b}'' + \Vec{p}''$ is integral. Hence, so is $\Vec{p}''$ itself.
\end{description}
This completes the proof of~\cref{thm:hybrid-2vass-hybrid}.
\end{proof}

\section{Proofs for~\cref{sec:DP3VASS}}\label{sec:appendix-DP3VASS}

\subsection{Proof of \cref{lem:wide-iff-open-wide}}\label{lem:appendix-wide-iff-open-wide}
\LemWideIFFOpenWide*
\begin{proof}
Let $V = (V_1)u_1(V_2)\ldots u_{k-1}(V_k)$, with $C_i := \Cone_{>0}(V_i)$ and $\overline{C_i} := \Cone(V_i)$. As $(V,s)$ is fowrard diagonal, there exists a cycle $\pi_+$ in $V_1$ such that $\Vec{p} := \Eff{\pi_+} \in \QQ_{>0}^3$. Also, as $(V,t)$ is backward diagonal, there exists a cycle $\pi_-$ in $V_k$ such that $\Vec{q} := \Eff{\pi_-} \in \QQ_{<0}^3$
\begin{description}
    \item[($\implies$)] Assume w.l.o.g.\ that $\QQ_{\ge0}^3 \subseteq \SeqCone(V)$. We prove that $\QQ_{>0}^3 \subseteq \SeqCone_{>0}(V)$.
    Let $\Vec{x} \in \QQ_{>0}^3$. For each $C_i$, take a vector $\Vec{c}_i$ whose norm is small enough such that $\Vec{p} + \Vec{c}_1 + \cdots + \Vec{c}_j \in \QQ_{>0}^3$ for all $j \in [k]$. Let $\epsilon > 0$ be a rational number small enough so that 
    \begin{equation*}
        \Vec{x}' := \Vec{x} - \epsilon(\Vec{p} + \Vec{c}_1 + \Vec{c}_2 + \cdots + \Vec{c}_k) \in \QQ_{>0}^3 \subseteq \QQ_{\ge 0}^3.
        \end{equation*}
    Then $\Vec{x}' \in \SeqCone(V)$ and there exists $\Vec{v}_1 \in \overline{C_1}, \ldots, \Vec{v}_k \in \overline{C_k}$ such that $\Vec{x}' = \Vec{v}_1 + \cdots + \Vec{v}_k$ and $\Vec{v}_1 + \cdots + \Vec{v}_j \ge \Vec{0}$ for all $j \in [k]$. We have 
    \begin{equation*}
        \Vec{x} = (\Vec{v}_1 + \epsilon(\Vec{p} + \Vec{c}_1)) + (\Vec{v}_2 + \epsilon\Vec{c}_2) + \cdots + (\Vec{v}_k + \epsilon\Vec{c}_k),
    \end{equation*}
    which witnesses that $\Vec{x} \in \SeqCone_{>0}(V)$. As the choice of $\Vec{x}$ is arbitrary, we conclude that $\QQ_{>0}^3 \subseteq \SeqCone_{>0}(V)$.
    \item[($\impliedby$)] Assume w.l.o.g.\ that $\QQ_{>0}^3 \subseteq \SeqCone_{>0}(V)$. We prove that $\QQ_{\ge0}^3 \subseteq \SeqCone(V)$.
    As $\SeqCone(V)$ is a convex set, it suffices to show that each unit vector $\Vec{e}_i$ lies in $\SeqCone(V)$. Observe that $\Vec{x}_i := \Vec{e}_i - \Vec{q} \in \QQ_{>0}^3 \subseteq \SeqCone_{>0}(V) \subseteq \SeqCone(V)$. As $\Vec{q} \in \overline{C_k}$, it is straightforward to check that $\Vec{e}_i = \Vec{x}_i + \Vec{q} \in \SeqCone(V)$. 
    \qedhere
\end{description}
\end{proof}

\subsection{Proof of \cref{lem:separating-plane-non-wide}}\label{lem:appendix-separating-plane-non-wide}

\LemSepratingPlaneNonWide*
We need a few technical auxiliary lemmas.

\begin{lemma}
    \label{lem:diff-seq-cone}
    Let $\Vec{a} \in \SeqCone(C_1, \ldots, C_j)$ and $\Vec{a}' \in \SeqCone(D_1, \ldots, D_\ell)$. If $\Vec{a} - \Vec{a}' \ge \Vec{0}$ then $\Vec{a} - \Vec{a}' \in \SeqCone(C_1, \ldots, C_j, -D_\ell, \ldots, -D_1)$. In other words, 
    \begin{align*}
        (\SeqCone(C_1, \ldots, C_j) &- \SeqCone(D_1, \ldots, D_\ell)) \cap \QQ_{\ge 0}^d 
        \\
        &\subseteq \SeqCone(C_1, \ldots, C_j, -D_\ell, \ldots, -D_1)
    \end{align*}
\end{lemma}

\begin{proof}
    We express $\Vec{a} = \Vec{c}_1 + \cdots + \Vec{c}_j$ and $\Vec{a}' = \Vec{d}_1 + \cdots + \Vec{d}_\ell$ that are summations witnessing the membership $\Vec{a} \in \SeqCone(C_1, \ldots, C_j)$ and $\Vec{a}' \in \SeqCone(D_1, \ldots, D_\ell)$. Then 
    \begin{equation*}
        \Vec{a} - \Vec{a}' = \Vec{c}_1 + \cdots + \Vec{c}_j + (-\Vec{d}_\ell) + \cdots + (-\Vec{d}_1).
    \end{equation*}
    We need to verify for every $x \in [0, \ell - 1]$, that $\Vec{c}_1 + \cdots + \Vec{c}_j + (-\Vec{d}_{\ell - 0}) + \cdots + (-\Vec{d}_{\ell - x}) \ge \Vec{0}$. This can be justified as 
    \begin{align*}
        \Vec{c}_1 + \cdots + \Vec{c}_j + (-\Vec{d}_{\ell - 0}) + \cdots + (-\Vec{d}_{\ell - x}) &= \Vec{a} - (\Vec{a}' - (\Vec{d}_1 + \ldots + \Vec{d}_{\ell - x - 1}))\\
        &= (\Vec{a} - \Vec{a}') + (\Vec{d}_1 + \ldots + \Vec{d}_{\ell - x - 1}) \ge \Vec{0},
    \end{align*}
    where the last inequality is valid because every prefix sum of $\Vec{d}_1 + \cdots + \Vec{d}_\ell$ is non-negative. We conclude that $\Vec{a} - \Vec{a}' \in \SeqCone(C_1, \ldots, C_j, -D_\ell, \ldots, -D_1)$.
\end{proof}

\begin{corollary}
    \label{cor:seq-wide-iff-diff-rev-wide}
    Let $C_1, \ldots, C_k \subseteq \QQ^d$ be cones. Then $\SeqCone(C_1, \ldots, C_k) = \QQ_{\ge 0}^d$ if and only if there exist $i \le j$ with $\QQ_{\ge 0}^d \subseteq \SeqCone(C_1, \ldots, C_i) - \SeqCone(-C_k, \ldots, -C_j)$.
\end{corollary}

\begin{proof}
    The ``only if'' direction can be easily justified by taking $i = j = k$. For the ``if'' direction, by \cref{lem:diff-seq-cone} we have 
    \begin{align*}
        \QQ_{\ge0}^d &= \QQ_{\ge0}^d \cap (\SeqCone(C_1, \ldots, C_i) - \SeqCone(-C_k, \ldots, -C_j))\\
        &\subseteq \SeqCone(C_1, \ldots, C_i, C_j, \ldots, C_k)\\
        &\subseteq \SeqCone(C_1, \ldots, C_k).
    \end{align*}
    As any sequential cone is contained in $\QQ_{\ge0}^d$ we conclude $\SeqCone(C_1, \ldots, C_k) = \QQ_{\ge 0}^d$.
\end{proof}

As $V$ is not wide, by \cref{cor:seq-wide-iff-diff-rev-wide} the difference $\SeqCone(V_h) - \SeqCone(\Reverse{V_t})$ cannot contain the whole $\QQ_{\ge0}^3$. Consequently, the two sequential cones should be separated by a plane. Before we prove this property, we justify that sequential cones are cones of exponential norm.

\begin{lemma}
    \label{lem:seq-cone-eq-cone}
    Let $C_1, \ldots, C_k \subseteq \QQ^d$ be finitely generated cones whose generators have norm bounded by $M$. Then $\SeqCone(C_1, \ldots, C_k)$ is a finitely generated cone whose generators have norm bounded by $\PolyF{M}^{(kd)^3}$.
\end{lemma}

\begin{proof}
    Let $P_i$ be the $(d \times n_i)$-matrix whose columns are the generators of $C_i$, for $i \in [k]$. Notice that we have $n_i \le (M+1)^d$ as $\norm{P_i} \le M$. Let $n := \sum_{i} n_i$, then $n \le k(M+1)^d$. Consider the following linear system
    \begin{align}
        \label{eq:seq-cone-ge-0}
        \begin{pmatrix}
            P_1 & 0 & 0 & \cdots & 0\\
            P_1 & P_2 & 0 & \cdots & 0\\
            & & \vdots & &\\
            P_1 & P_2 & P_3 & \cdots & P_k\\
        \end{pmatrix} \cdot \begin{pmatrix}
            \Vec{x}_1 \\ \Vec{x}_2 \\ \vdots \\ \Vec{x}_k
        \end{pmatrix} \ge \Vec{0}
    \end{align}
    Each non-negative integer solution induces a vector $P_1\Vec{x}_1 + \cdots + P_k\Vec{x}_k$ in $\SeqCone(C_1, \ldots, C_k)$ and vice versa. Let $X$ be the set of the non-negative integer solutions to \cref{eq:seq-cone-ge-0}. By \cref{cor:ilp-decomp}, $X$ is a cone generated by some finite set $X_0$ where 
    \begin{align*}
        \norm{X_0} &\le ((n+kd+1)M+1)^{kd} \le ((k(M+1)^d + kd + 1)M+1)^{kd} \\&\le ((M+1)^{(2d+1)k})^{kd} \le \PolyF{M}^{k^2d^2}    
    \end{align*}
    Each solution $(\Vec{x}_1, \ldots, \Vec{x}_k) \in X_0$ yields a generator of $\SeqCone(P_1, \ldots, P_k)$, namely $\Vec{g} := P_1\Vec{x}_1 + \cdots + P_k\Vec{x}_k$. It holds that $\norm{\Vec{g}}\le n \cdot M \cdot \norm{X_0} \le \PolyF{M}^{(kd)^3}$.
\end{proof}

\begin{proof}[Proof of \cref{lem:separating-plane-non-wide}]
    Denote $C_h = \SeqCone(V_h)$ and $C_t = \SeqCone(\Reverse{V_t})$. Let $P_h = \{\Vec{u}_1, \ldots, \Vec{u}_\ell\}$ be the set of generators of $C_h$, and let $P_t = \{\Vec{v}_1, \ldots, \Vec{v}_m\}$ be the set of generators of $C_t$. By \cref{lem:seq-cone-eq-cone}, we assume $\norm{P_h}, \norm{P_t} \le \PolyF{\Size(V)}^{(kd)^3}$. Let $W = P_h \cup P_t$, we identify two cases:

    \begin{description}
        \item[Case 1. $\Span(W) \ne \QQ^d$.] In this case, both $C_h$ and $C_t$ are already contained in a hyperplane. Let $X = \{\Vec{x}_1, \ldots, \Vec{x}_r\}$ be a basis of $\Span(W)$ comprising vectors from $W$. Then any non-zero vector $\Vec{n} \in \ZZ^d$ orthogonal to $X$ ``separates'' $C_h$ and $C_t$, in the sense that $\Inner{\Vec{n}, C_h} \ge 0$ and $\Inner{\Vec{n}, C_t} \le 0$ (because they are both equal to zero). We may take $\Vec{n}$ as any non-zero integer solution of the equation $X^T\Vec{n} = \Vec{0}$. By \cref{lem:pottier}, such an $\Vec{n}$ exists with $\norm{\Vec{n}} \le \bigO(d\norm{W})^d \le \PolyF{\Size(V)}^{k^3d^5}$. Here we remark that \cref{lem:pottier} is stated for non-negative solutions. We may negate some columns of $X$ to assume that $\Vec{n} \ge \Vec{0}$, which does not affect the norm of $\Vec{n}$.

        \item[Case 2. $\Span(W) = \QQ^d$.] As $V$ is not wide, by \cref{cor:seq-wide-iff-diff-rev-wide}, there exists a non-zero vector $\Vec{p} \in \ZZ_{\ge0}^d$ such that $\Vec{p} \notin C_h - C_t$. Viewing $P_h$ and $P_t$ as matrices whose columns are their elements, the following equation has no non-negative solution:
        \begin{equation*}
            \begin{pmatrix}
                P_h & -P_t
            \end{pmatrix} \cdot \Vec{x} = \Vec{p}.
        \end{equation*}
        By the fundamental theorem of linear inequalities \cite[Theorem 7.1]{DBLP:books/daglib/0090562}, there exists a non-zero vector $\Vec{n}_0 \in \QQ^d$ such that $\Inner{\Vec{n}_0, \Vec{p}} < 0$ and
        \begin{itemize}
            \item $\Inner{\Vec{n}_0, \Vec{u}} \ge 0$ for $\Vec{u} \in P_h$, and $\Inner{\Vec{n}_0, -\Vec{v}} \ge 0$ for $\Vec{v} \in P_t$; hence $\Inner{\Vec{n}_0, \SeqCone(V_h)} \ge 0$, and $\Inner{\Vec{n}_0, \SeqCone(\Reverse{V_t})} \le 0$;
            \item the hyperplane $\{\Vec{x} \in \QQ^d \mid \Inner{\Vec{n}_0, \Vec{x}} = 0\}$ contains $d - 1$ linear independent vectors from $W$; here we let $X$ denote the set of these $d-1$ vectors.
        \end{itemize}
        The solutions to the equation $X^T \Vec{n} = 0$ is exactly $\Cone\{\Vec{n}_0, -\Vec{n}_0\}$, as the rank of $X$ is $d - 1$. Let $j \in [d]$ be such that $\Vec{n}_0(j) \ne 0$, we further assume $\Vec{n}_0(j) > 0$. The desired vector $\Vec{n}$ may be taken as any non-zero integer solution to the system $X^T \Vec{n} = 0, \Vec{n}(j) > 0$. By \cref{cor:ilp-decomp}, such an $\Vec{n}$ exists with $\norm{\Vec{n}} \le \bigO(d\norm{W})^{d+1} \le \PolyF{\Size(V)}^{k^3d^5}$.
        \qedhere
    \end{description}
\end{proof}

\subsection{Proof of \cref{lem:reach-eq-seq-cone-plus-bounded-of-diag}}\label{lem:appendix-reach-eq-seq-cone-plus-bounded-of-diag}

We need the following lemma to show \cref{lem:reach-eq-seq-cone-plus-bounded-of-diag}.

\begin{lemma}
    \label{lem:cone-inter-pos-nonempty}
    If $\Cone(V) \cap \QQ_{>0}^d \neq \emptyset$, then there exists a cycle $\theta$ with $|\theta| \le \PolyF{(\Size(V)}^d$ and $\Eff{\theta} \ge \Vec{1}$.
\end{lemma}

\begin{proof}
    Assume that $V=(Q, T)$, $s=p(\Vec{x})$, and $M:=\Size(V)$.
    For the first statement, let $\Vec{a} \in \Cone(V) \cap \QQ_{>0}^d$ be a positive vector.
    By Carathéodory’s Theorem~\cite[Corollary 7.1i]{DBLP:books/daglib/0090562}, there exists $d$ simple cycles $(\theta_i)_{i \in [d]}$ in $V$ such that 
    \begin{equation*}
        \sum_{i \in [d]}\lambda_i\cdot\Eff{\theta_i} = \Vec{a},
    \end{equation*}
    where $\lambda_i \in \NN$ for each $i \in [d]$. By scaling, we may assume that each $\lambda_i$ is an integer and $\Vec{a}$ an integer vector. It holds $\Vec{a}\geq \Vec{1}$. Consider the underlying Diophantine system $A\Vec{\lambda} \ge \Vec{1}$, where $\Vec{\lambda} \in \NN^d$ are variables and $A \in \ZZ^{d \times d}$ is the matrix whose columns are exactly $(\Eff{\theta_i})_{i \in [d]}$.
    Note that $\norm{A} \le |Q| \cdot \norm{T}$.
    Applying~\cref{cor:ilp-decomp}, we obtain that there exists some minimal solution $\Vec{\lambda}_0$ with $\norm{\Vec{\lambda}_0} \le \bigO(M)^{2d}$.
    One can construct a cycle accordingly: select the shortest cycle $\theta_0$ visiting all states, and then concatenate each of the $d$ simple cycles; we repeat the cycle $\theta_i$ exactly $(|\theta_0|\cdot \norm{T} + 1)\cdot \Vec{\lambda}_0(i)$ times for each $i \in [d]$, to obtain a single cycle $\theta$, whose size is 
    \begin{equation*}
        |\theta| \le |\theta_0|+(|\theta_0|\cdot \norm{T} + 1)\cdot d \cdot \norm{\Vec{\lambda}_0} \le \bigO(M)^{2d+3} \le \PolyF{M}^d.
    \end{equation*}
    and its effect is $\Eff{\theta} = \Eff{\theta_0} + (|\theta_0|\cdot \norm{T} + 1) \cdot A \Vec{\lambda}_0 \ge -|\theta_0|\cdot \norm{T}\cdot \Vec{1} + (|\theta_0|\cdot \norm{T} + 1) \cdot \Vec{1} \ge \Vec{1}$. 
\end{proof}

\LemmaReachEqSeqConePlusBoundedOfDiag*
\begin{proof}
    A run $\pi$ from $s$ to $t$ can be factored to $k$ segments, one for each component. The effect of each segment $\pi_i$ can be further decomposed into a sum of the effect $\Vec{z}_i$ of some simple path, and a vector $\Vec{c}_i \in \Cone(V_i)$. We have 
    \begin{equation}\label{2026-07-07-0828}
        \Vec{t}_i := \Vec{u} + (\Vec{z}_1 + \Vec{c}_1) + \cdots + (\Vec{z}_i + \Vec{c}_i) \ge \Vec{0}
    \end{equation}
    for all $i \in [k]$, as it is the vector of some configuration on $\pi$. Now let $\Vec{p}$ be the effect of a diagonal cycle in $V_1$ whose existence is guaranteed. We may assume $\norm{\Vec{p}} \le \PolyF{\Size(V)}^{d}$ by \cref{lem:cone-inter-pos-nonempty}. Let $m := \norm{\Vec{u}} + \sum_{i=1}^k \norm{\Vec{z}_i} \le \norm{s} + \PolyF{\Size(V)}$. We claim that $m\Vec{p} + \Vec{c}_1 + \cdots + \Vec{c}_k \in \SeqCone(V)$. Indeed, for $i \in [k]$,
    \begin{align*}
        m\Vec{p} + \Vec{c}_1 + \cdots + \Vec{c_i} &= 
            \Vec{t}_i + m\Vec{p} - (\Vec{u} + \Vec{z}_1 + \cdots + \Vec{z}_i) \ge \Vec{t}_i \ge \Vec{0}.
    \end{align*}
    Now $\Vec{v}=\Vec{t}_k$ by \cref{2026-07-07-0828}.
    It follows that the difference between $\Vec{v}$ and $\SeqCone(V)$ is bounded by 
    $\Vec{h} = (m\Vec{p} + \Vec{c}_1 + \cdots + \Vec{c_k})-\Vec{v} = m\Vec{p} - (\Vec{u} + \Vec{z}_1 + \cdots + \Vec{z}_k)$. Clearly $\norm{\Vec{h}} \le 2m\norm{\Vec{p}} \le \PolyF{\Size(V)}^{d+1}\cdot\norm{s}$.
\end{proof}

\subsection{Proof of \cref{lem:non-wide-bounded-inner-prod}}\label{lem:appendix-non-wide-bounded-inner-prod}

\LemmaNonWideBoundedInnerProd*
\begin{proof}
By \cref{lem:reach-eq-seq-cone-plus-bounded-of-diag}, there exist vectors $\Vec{h}_h, \Vec{h}_t \in \ZZ^3$ such that 
    $\Vec{z} = \Vec{s}_h + \Vec{h}_h = \Vec{s}_t + \Vec{h}_t$, where $\Vec{s}_h \in \SeqCone(\mathcal{H}_i)$ and $\Vec{s}_t \in \SeqCone(\Reverse{\mathcal{T}_i})$. Moreover, $\norm{\Vec{h}_h} \le \PolyF{\Size(V)} \cdot \norm{s}$ and $\norm{\Vec{h}_t} \le \PolyF{\Size(V)} \cdot \norm{t}$. By \cref{lem:separating-plane-non-wide}, the vector $\Vec{n}_i$ satisfies $\Inner{\Vec{n}_i, \Vec{s}_h} \ge 0$ and $\Inner{\Vec{n}_i, \Vec{s}_t} \le 0$. Hence,
    \begin{align*}
        \Inner{\Vec{n}_i, \Vec{z}} & = \Inner{\Vec{n}_i, \Vec{s}_h} + \Inner{\Vec{n}_i, \Vec{h}_h} \ge - 3\norm{\Vec{n}_i}\norm{\Vec{h}_h} \ge - \PolyF{\Size(V)}^{k^3}\norm{s},\\
        \Inner{\Vec{n}_i, \Vec{z}} & = \Inner{\Vec{n}_i, \Vec{s}_t} + \Inner{\Vec{n}_i, \Vec{h}_t} \le 3\norm{\Vec{n}_i}\norm{\Vec{h}_t} \le \PolyF{\Size(V)}^{k^3}\norm{t}.
    \end{align*}
    We conclude that $|\Inner{\Vec{n}_i, \Vec{z}}| \le \PolyF{\Size(V)}^{k^3} \cdot (\norm{s} + \norm{t})$.
\end{proof}

\subsection{Proof of \cref{lem:bound-f1}}\label{lem:appendix-bound-f1}
\LemmaBoundFone*
\begin{proof}
Let $M$ denote $\Size(V,s,t)$, and let $c$ be a constant satisfying $L(x, y) \le (x+y)^c$, $H(x, y)^2 \le (x+y)^c$, and $B_h, B_n, \Size(V_j') \le M^{ck^3}$. 
    We prove by induction that $f_j(x) \le (M + x)^{(4c^3k^3)^{k-j+1}}$. Firstly, for $j = k$ we have
    \begin{equation*}
        f_k(x) \le (\Size(V_k', t) + x)^c \le (M^{ck^3} + x)^c \le (M + x)^{c^2k^3}.
    \end{equation*}
    Secondly, for $j < k$, verify that 
    \begin{equation*}
        y := (B_h +3B_n+ 1)\cdot H(\Size(V_j'), x)^2 \le (M + x)^{2c^2k^3}.
    \end{equation*}
    By the induction hypothesis, $B_j(x) = f_{j+1}(y) \le (M + x)^{(2c^2k^3) \cdot (4c^3k^3)^{k-j}}$.
    \begin{align*}
        f_j(x) & = 2L(\Size(V_j')+B_j(x), x) \le 2(M^{ck^3} + (M + x)^{(2c^2k^3) \cdot (4c^3k^3)^{k-j}}+x)^c\\
        &\leq 2(M+x)^{(3c^3k^3)\cdot (4c^3k^3)^{k-j}}\\
        & \le (M + x)^{(4c^3k^3)^{k-j+1}}.
    \end{align*}
    Finally, we have $f_1(x) \le (\Size(V, s, t) + x)^{2^{\PolyF{k}}}$.
\end{proof}
\section{Proofs for~\cref{sec:par-pump-3-vass}}\label{sec:appendix-par-pump-3-vass}

\subsection{Proof of~\cref{lem:prog-pump-or-bounded-in-par-3-vass}}\label{lem:prog-pump-or-bounded-in-par-3-vass-proof}

\lemProgPumpOrBoundedInParThreeVASS*

\begin{proof}
    Assume that $s = p(\Vec{x})$ and $M:=\Size(V)$.
    Let $U := 2P(M)^{2k} \cdot (\norm{s} + 2)$, where $P(x)$ is the polynomial defined in~\cref{lem:par-pump-or-bounded-in-2-vass}.
    If $\BDCounters(\pi, U) \ne \emptyset$, then there exists a coordinate that remains bounded by $U$ throughout $\pi$, and we are done by taking $\pi_1=\pi_3:=\varepsilon$ and $\pi_2 := \pi$.
    We shall focus on the case $\BDCounters(\pi, U) = \emptyset$ and will show that there exists a sequence of jointly diagonal cycles $\theta_1, \theta_2, \theta_3$ enabled along $\pi$.
    Since $c_1:=s$ is pumpable in $V$, there exists some $\iota \in [d]$ such that $p(\Vec{x})\xrightarrow{*}p(\Vec{x}')$ in $V$ for some $\Vec{x}' \ge \Vec{x}+\Vec{e}_\iota$.
    Applying Rackoff's classical upper bound of coverability~\cite[Theorem 3.5]{DBLP:journals/tcs/Rackoff78}, we obtain a covering run $p(\Vec{x})\xrightarrow{\theta_1}p(\Vec{x}'')$ in $V$ of length 
    \begin{equation*}
        |\theta_1|\le \PolyF{M, \norm{\Vec{x}}+1} \le \PolyF{\Size(V,s)},
    \end{equation*}
    where $\Vec{x}'' > \Vec{x}$.
    This also implies that $\Eff{\theta_1}>\Vec{0}$.
    
    Now there are two subcases. 
    In the case where $\Eff{\theta_1} \ge \Vec{1}$, we are done by taking $\theta_2=\theta_3:=\theta_1$, and $c_2 = c_3:=c_1$.
    Otherwise, we assume, w.l.o.g., that $\Eff{\theta_1}(3) > 0$.
    Given a vector $\Vec{x} \in \QQ^3$, we use $\Vec{x}|_{-3} \in \QQ^2$ to denote the unique vector satisfying $\Vec{x}|_{-3}(i)=\Vec{x}(i)$ for $i \in [2]$.
    This subscription can be extended to configurations, runs, and VASS to denote the counterpart obtained by dropping the third coordinate.
    Now we omit the third coordinate of $(V, s, t)$ to obtain a new 2-VASS $(V',s',t')$ with size at most $M$.
    Also, $\pi$ induces a run $\pi':=\pi|_{-3}$ in $V'$ with $\BDCounters(\pi', U) = \emptyset$, which means that there exists a configuration along $\pi'$ whose norm exceeds $U$.
    Let $c$ be the first configuration with $\norm{c} \ge P(\Size(V'))^k\cdot(\norm{s'}+1)$.
    Note that $V$ and $V'$ have the same state graph; hence, $V'$ is also $k$-component.
    By~\cref{lem:par-pump-or-bounded-in-2-vass}, there is a cycle $\theta'_2$ such that $\Eff{\theta'_2} > \Vec{0}$ and $|\theta'_2| \le P(M)$, which is fireable at some $p_2(\Vec{x}'_2)$ before $c$.
    For the original VASS $V$, these induce a cycle $\theta_2$ with $|\theta_2|\le P(M)$ and $\Eff{\theta_2}|_{-3} = \Eff{\theta_2'}>\Vec{0}$.
    Moreover, there exists a configuration $c_2=p_2(\Vec{x}_2)$ such that $(\Vec{x}_2)|_{-3} = \Vec{x}_2' \ge \Drop(\theta_2)|_{-3}$.
    The third entry of $\Eff{\theta_3}$ may be negative.
    But considering the vector 
    \begin{equation*}
        \Vec{a}_2:=(\norm{\Drop(\theta_2)}+1) \cdot \Eff{\theta_1},
    \end{equation*}
    it is easy to see that $\Vec{a}_2 \in \SeqCone(\theta_1)$.
    As $\Vec{x}_2 + \Vec{a}_2 \ge \Drop(\theta_2)$, we conclude that $\theta_2$ is fireable at $p_2(\Vec{x}_2 + \Vec{a}_2)$.
    Let $\pi_1$ be the prefix from $s$ to $c_2$. 
    Since $c$ is the first configuration whose norm exceeds the specified bound, it follows that
    \begin{equation*}
        \Vec{z}(i) \le P(M)^k\cdot(\norm{s'}+1) + M \le P(M)^k\cdot(\norm{s}+2) \le U,
    \end{equation*}
    holds for every $i \in [2]$ and $q(\Vec{z})$ along $\pi_1$.
    Therefore, we have $|\BDCounters(\pi_1, U)|\ge 2$ and for each $\iota \in [3]$ with $\Eff{\theta_1}(\iota) = 0$, it holds that $\iota \in [2] \subseteq \BDCounters(\pi_1, U)$
    
    If $\Vec{a}_2 + \Eff{\theta_2} \in \QQ_{>0}^3$, then we are done by taking $\theta_3:=\theta_2$ and $c_3:=c_2$.
    Otherwise, we assume, w.l.o.g., that $\Eff{\theta_2}(2) > 0$ and $\Eff{\theta_2}(1) = 0$.
    Since 
    \begin{equation*}
        P(\Size(V))^k\cdot(\norm{c_2} + 1) \le P(M)^k \cdot (P(M)^k \cdot (\norm{s} + 2) + 1) \le U, 
    \end{equation*}
    following the same argument, we can obtain a cycle $\theta_3$ with $\Eff{\theta_3}(1) > 0$, together with some $c_3=p_3(\Vec{x}_3)$ with $\Vec{x}_3(1) \ge \Drop(\theta_3)(1)$.
    Let 
    \begin{equation*}
        \Vec{a}_3 := (\norm{\Drop(\theta_3)} + 1)\cdot \Vec{a}_2 \in \SeqCone(\theta_1, \theta_2).
    \end{equation*}
    Then $\Vec{a}_3 + \Eff{\theta_3} \ge \Vec{1}$, which implies that $\theta_1, \theta_2, \theta_3$ are jointly diagonal.
    Moreover, the cycle $\theta_3$ are fireable at $p_3(\Vec{x}_3+\Vec{a}_3)$, and consequently $\theta_1$, $\theta_2$, $\theta_3$ are enabled along $c_1$, $c_2$, and $c_3$.
    Let $\pi_2$ be the infix of $\pi$ from $c_2$ to $c_3$.
    Then $|\BDCounters(\pi_1\pi_2, U + M)| \ge 1$, and for each $\iota \in [3]$ with $\Eff{\theta_1}(\iota) = \Eff{\theta_2}(\iota) = 0$, it holds that $\iota \in \{1\} \subseteq \BDCounters(\pi_2, U)$.
    Finally, we complete the proof by letting $B$ be the maximum of all the bounds appearing in the above argument.
\end{proof}

\subsection{Proof of~\cref{lem:reach-closed-under-pumping-cycle}}\label{lem:appendix-reach-closed-under-pumping-cycle}

\lemReachClosedUnderPumpingCycle*

\begin{proof}
    Let $s \xrightarrow{\pi} q(\Vec{y})$ be the shortest run between these two configurations.
    It suffices to show that, for each $\ell \in \NN$, there exists a run $s\xrightarrow{\rho_i} q(\Vec{y}+\ell\cdot\Vec{a})$ with $|\rho_i| \le |\pi| + \ell \cdot \lambda \cdot |\theta_1| + \ell\cdot |\theta_2|$.
    Indeed, we can take $\rho_i := \theta_1^{\ell\lambda} \pi \theta_2^{\ell}$.
    Since $\Vec{a} \ge \Vec{0}$ and $s \xrightarrow{\theta_1} p(x+\Vec{a})$, it must hold that $s \xrightarrow{\theta_1^{\ell\lambda}} p(\Vec{x}+\ell\cdot\lambda\cdot\Vec{a})$.
    By the monotonicity of VASS, we have $p(\Vec{x}+\ell\cdot\lambda\cdot\Vec{a}) \xrightarrow{\pi} q(\Vec{y}+\ell\cdot\lambda\cdot\Vec{a})$.
    Therefore, it remains to justify that $q(\Vec{y}+\ell\cdot\lambda\cdot\Vec{b}) \xrightarrow{\theta_2^\ell} q(\Vec{y} + \ell \cdot \Vec{b})$, which is already a $\ZZ$-run.
    It suffices to show that the first iteration and the last iteration of $\theta_2$ are fireable, when $\ell \ge 1$.
    By assumption, we have $\Vec{y} + \ell\cdot\lambda\cdot\Vec{a} \ge \Vec{y} + \lambda\cdot\Vec{a} \ge \Drop(\theta_2)$, which implies that the first iteration is valid.
    As for the last iteration, just note that 
    \begin{equation*}
        \Vec{y} + \ell\cdot\lambda \cdot \Vec{a} + (\ell-1)\cdot \Eff{\theta_2} = \Vec{y} + (\ell - 1)\cdot\Vec{b} + \lambda\cdot \Vec{a} \ge \Vec{y} + \lambda\cdot\Vec{a} \ge \Drop(\theta_2).
    \end{equation*}
    Therefore, we have $s \xrightarrow{\rho_i} q(\Vec{y} + \ell \cdot \Vec{b})$, which completes the proof.
\end{proof}

\subsection{Proof of~\cref{lem:one-turn-approx}}\label{lem:appendix-one-turn-approx}

\lemOneTurnApprox*

\begin{proof}
    Note that we do not distinguish between $\QQ^{I}$ and $\QQ^{|I|}$. Let $M := \Size(V)$.
    For each vector $\Vec{z} \in \QQ^3$, we denote by $\Vec{z}|_{I}$ the unique vector in $\QQ^{I}$ such that $\Vec{z}|_{I}(\iota) = \Vec{z}(\iota)$ for each $\iota \in I$, and write $\Vec{z}|_{-I}:=\Vec{z}|_{[d]\setminus I}$. 
    Let $Q' := Q \times [0,B-1]^{I}$, and we write $p_{\Vec{u}}$ for the state $(p,\Vec{u}) \in Q'$.
    Define
    \begin{equation*}
        T':=\{(p_{\Vec{u}},\Vec{a}|_{-I},q_{\Vec{v}})\mid (p,\Vec{a},q) \in T, \Vec{u},\Vec{v} \in [0,B-1]^{I},\text{ and }\Vec{u}+\Vec{a}|_I=\Vec{v}\}.
    \end{equation*}
    Then we can construct a VASS $V' := (Q', T')$ with $\Size(V') \le B^{|I|}\cdot M\le B^3\cdot M$, which is intuitively obtained from $V$ by encoding all the coordinates in $I$ into states.
    We may add one isolated state and at most two self-loops to $V'$ to ensure that $\CycleSpace(V') = \QQ^{[3] \setminus I}$.
    Assume that $S_p = \Vec{b} + \ResPeriodic{P}{\Vec{c}}{K}$.
    We shall assume that each vector in $P$ can be used at least once.
    We first show that $\Vec{p}|_I = \Vec{0}$ for each $\Vec{p}\in P$.
    Otherwise, there exists some $\Vec{p} \in P$ with $\Vec{p}|_I > \Vec{0}$.
    Note that $\Vec{z} := \Vec{b} + \Vec{p} \in S$ while $\Vec{z}|_I > \Vec{b}|_I$, which contradicts with the fact that $\Vec{z}|_I = \Vec{x}|_I = \Vec{b}|_I$.
    Now we can obtain a hybrid set $S_p' = \Vec{b}' + \ResPeriodic{(P')}{\Vec{c}}{K}$ containing $\Vec{x}|_{-I}$ by projecting $S_p$ onto $\QQ^{[3]\setminus I}$, which renders true the following claim.
    \begin{claim}
        \label{claim:project-hybrid-set-recoverable}
        For each $\Vec{z}' \in S_p'$, there is a unique $\Vec{z} \in S_p$ such that $\Vec{z}|_{-I} = \Vec{z}'$ and $\Vec{z}|_I = \Vec{x}|_I$.
    \end{claim} 
    It is routine to check that $p_{\Vec{x}|_I}(\Vec{x}|_{-I}) \xrightarrow{*} q_{\Vec{y}|_I}(\Vec{y}|_{-I})$ in $V'$.
    Note that $V'$ is at most 2-dimensional, given that $I \neq \emptyset$.
    Since $P' \subseteq \CycleSpace(V')$, by~\cref{thm:hybrid-2vass-hybrid}, there exists a hybrid set in the form of $ S_q':=\Vec{b}_q' + \ResPeriodic{(P'_q)}{\Vec{c}_q}{K_q}$ with $\NORM{S_q'} \le \PolyF{B,M,\NORM{S_p}}$ and 
    \begin{equation*}
        \Vec{y}|_{-I} \in S_q' \subseteq \ReachqVs{q_{\Vec{y}|_I}}{V'}{p_{\Vec{x}|_I}(S_p')}.
    \end{equation*}
    Let $\Vec{b}_q \in \NN^3$ be the vector obtained by augmenting $\Vec{b}_q'$ with the values $\Vec{y}|_I$ at the coordinates in $I$. Similarly, let $P_q\subseteq \NN^3$ be the sets obtained by augmenting $P'$ with zero values at the coordinates in $I$.
    We claim that $S_q := \Vec{b}_q + \ResPeriodic{P_q}{\Vec{c}_q}{K_q}$ satisfies the requirements.
    \begin{itemize}
        \item Since $\norm{(\Vec{y}|_I)} \le B$, it holds that $\NORM{S_q} \le \NORM{S_q'} + B \le \PolyF{B + M + \NORM{S_p}}$.
        \item It is trivial that $\Vec{y} \in S_q$.
        For every $\Vec{z} \in S_q$, we have $\Vec{z} - \Vec{b}_q \in P_q^*$.
        It holds by construction that $(\Vec{z} - \Vec{b}_q)|_I = \Vec{0}$.
        Therefore, $\Vec{z}(\iota) = \Vec{b}_q(\iota) = \Vec{y}(\iota)$ for every $\iota \in I$.
        \item For each $\Vec{y}' \in S_q$, we have $\Vec{y}'|_{-I} \in S_q'$, which implies that there exists some $\Vec{z} \in S_p'$ with $p_{\Vec{x}|_I}(\Vec{z}) \xrightarrow{*} q_{\Vec{y}|_I}(\Vec{y}'|_{-I})$ in $V'$.
        By~\cref{lem:len-target-in-hybrid-set}, there is a run $p_{\Vec{x}|_I}(\Vec{z}') \xrightarrow{\rho'}q_{\Vec{y}|_I}(\Vec{y}'|_{-I})$ in $V'$ for some $\Vec{z}' \in S_p'$ with
        \begin{equation*}
            \norm{\Vec{z}'}, |\rho'| \le \PolyF{B, M,\NORM{S_p'}, \norm{\Vec{y}'}} \le \PolyF{B, M,\NORM{S_p}, \norm{\Vec{y}'}}.
        \end{equation*}
        Let $\Vec{x}' \in \NN^3$ be the vector obtained by augmenting $\Vec{z}'$ with the values $\Vec{x}|_I$ at the coordinates in $I$.
        Then by~\cref{claim:project-hybrid-set-recoverable}, $\Vec{x}' \in S_p$ and  $\norm{\Vec{x}'} \le  \PolyF{B + \Size(V) + \NORM{S_p} + \norm{\Vec{y}'}}$.
        It is straightforward that $\rho''$ induces a run $p(\Vec{x}')\xrightarrow{\rho} q(\Vec{y}')$ in $V$ with
        \begin{equation*}
            |\rho| = |\rho'| \le  \PolyF{B + \Size(V) + \NORM{S_p} + \norm{\Vec{y}'}}.
        \end{equation*}
    \end{itemize}
    The argument above completes the proof of the lemma.
\end{proof}

\subsection{Proof of~\cref{claim:len-first-round-approx}}\label{claim:appendix-len-first-round-approx}

\claimLenFirstRoundApprox*

\begin{proof}
    It is trivial that $\Vec{w} \in S_r \subseteq S_r'$.
    For any $\Vec{z} \in S_r'$, we can write it as 
    \begin{equation*}
        \Vec{z} =  \Vec{z}' + \ell\cdot\Vec{a} = \Vec{z}' + \ell\cdot \lambda\cdot \Eff{\theta_1} + \ell \cdot \Eff{\theta_2},
    \end{equation*}
    where $\Vec{z}' \in S_r \subseteq \ReachqVs{r}{V}{s}$ and $\ell \in \NN$.
    Note that $\Vec{z} \ge \Vec{z}'$.
    Assume that $\Vec{a}(\iota) > 0$ for some $\iota \in [3]$.
    Since $\Vec{z}(\iota) \ge \Vec{z}'(\iota) + \ell \ge \ell$, we have $\ell, \norm{\Vec{z}'} \le \norm{\Vec{z}}$.
    Then we have
    \begin{equation*}
        \min\Len(V,s,r(\Vec{z}')) \le \Size(V,s,r(\Vec{z}'))^{2m^2k} \le \Size(V,s,r(\Vec{z}))^{2m^2k}.
    \end{equation*}
    Since $\theta_1$, $\theta_2$, $\theta_3$ can be enabled, by definition, $\theta_1$ is fireable at $s$ and $\Vec{w} + \lambda'\cdot\Eff{\theta_1} \ge \Drop(\theta_2)$ for some $\lambda'\in \NN$.
    Note that $\Vec{z}'$ coincides with $\Vec{w}$ in the second and third coordinates.
    We have $\Vec{z}' + \lambda' \cdot \Eff{\theta_1} \ge \Drop(\theta_2)$.
    It is clear that choosing $\lambda = BM+1 \ge \norm{\Drop(\theta_2)}$ suffices, which ensures that $\Vec{z}' + \lambda\cdot \Eff{\theta_1}\ge \Drop(\theta_2)$, or equivalently, 
    \begin{equation*}
        r(\Vec{z}' + \lambda\cdot\Eff{\theta_1})\xrightarrow{\theta_2} r(\Vec{z}' + \Vec{a}).
    \end{equation*}
    By~\cref{lem:reach-closed-under-pumping-cycle}, we have $\Vec{z} \in \ReachqVs{r}{V}{s}$ and
    \begin{equation*}
        \begin{aligned}
            \min \Len(V, s, r(\Vec{z})) &\le \min\Len(V,s,r(\Vec{z}')) + \ell\cdot\lambda \cdot|\theta_1| + \ell \cdot |\theta_2|\\
            &\le \Size(V,s,r(\Vec{z}))^{2m^2k} + \norm{\Vec{z}} \cdot (BM + 1) \cdot B + \norm{\Vec{z}}\cdot B \\
            &\le \Size(V,s,r(\Vec{z}))^{6m^2k}.
        \end{aligned}
    \end{equation*}
    Note that $\NORM{S_r'} \le M^{2m^2k} + \norm{\Vec{a}} \le M^{6m^2k}$, which completes the proof.
\end{proof}

\subsection{Proof of~\cref{claim:second-round-approx}}\label{claim:appendix-second-round-approx}

\claimSecondRoundApprox*

\begin{proof}
    For any $\Vec{z} \in S_p'$, we can write it as 
    \begin{equation*}
        \Vec{z} =  \Vec{z}' + \ell\cdot\Vec{p} = \Vec{z}' + \ell\cdot \lambda\cdot \Vec{a} + \ell \cdot \Eff{\theta_3},
    \end{equation*}
    where $\ell \in \NN$ and $\Vec{z}' \in S_p$.
    Similarly, we have $\ell, \norm{\Vec{z}'} \le \norm{\Vec{z}}$.
    Then by~\cref{eq:norm-w'-len-rho}, there exists some $\Vec{w}' \in S'_r$ with a run $r(\Vec{w}')\xrightarrow{\rho} p(\Vec{z}')$ in $V$ such that
    \begin{equation*}
        \norm{\Vec{w}'}, |\rho| \le \Size(V,s,p(\Vec{z}'))^{8m^3k} \le  \Size(V,s,p(\Vec{z}))^{8m^3k}.
    \end{equation*}
    Note that $\Vec{w}'':=\Vec{w}' + \ell \cdot\lambda\cdot\Vec{a} \in S_r'$ holds by construction.
    Then
    \begin{equation*}
        \norm{\Vec{w}''} \le \norm{\Vec{w}'} + \norm{\Vec{z}} \cdot \lambda \cdot \norm{\Vec{a}} \le  2\cdot\Size(V,s,p(\Vec{z}))^{8m^3k} \le \Size(V,s,p(\Vec{z}))^{9m^3k}.
    \end{equation*}
    By~\cref{claim:len-first-round-approx}, there exists a run $s \xrightarrow{\sigma} r(\Vec{w}'')$ with 
    \begin{equation*}
        \begin{aligned}
            |\sigma| &\le (M + 3 \cdot \norm{\Vec{w}''})^{6m^2k} \\
            &\le \left(4 \cdot \Size(V,s,p(\Vec{z}))^{9m^3k}\right)^{6m^2k} \\
            &\le \Size(V,s,p(\Vec{z}))^{66m^5k^2}.
        \end{aligned}
    \end{equation*}
    Again by the monotonicity of VASS reachability, we have $r(\Vec{w}'')\xrightarrow{\rho}p(\Vec{z}'')$ in $V$, where $\Vec{z}'':=\Vec{z}' + \ell\cdot \lambda\cdot \Vec{a}$.
    It is routine to check that $p(\Vec{z}'') \xrightarrow{\theta_3^\ell} p(\Vec{z})$.
    Thus, we conclude that $s \xrightarrow{\sigma\rho\theta_3^\ell} p(\Vec{z})$ and
    \begin{equation*}    
        \begin{aligned}
            \min \Len(V, s, p(\Vec{z})) &\le |\sigma| + |\rho| + \ell\cdot |\theta_3|\\
            &\le 3\cdot \Size(V,s,p(\Vec{z}))^{66m^5k^2}\\
            &\le \Size(V,s,p(\Vec{z}))^{68m^5k^2}.
        \end{aligned}
    \end{equation*}
\end{proof}

\subsection{Proof of~\cref{claim:semi-diag-1}}\label{claim:appendix-semi-diag-1}
\ClaimSemiDiagFirst*
\begin{proof}
It is clear that $s'\xrightarrow{*} t'$ in $V'$.
Note that $\Size(V',s',t') \le  3L_k(M) + 10 \le 13 L_k(M)$, and $\min \Len(V, s, t) \le L_k(M) = \min \Len(V',s',t')$. 
\end{proof}

\subsection{Proof of~\cref{claim:semi-diag-2}}\label{claim:appendix-semi-diag-2}
\ClaimSemiDiagSecond*
\begin{proof}
The size bound is already given in~\cref{eq:2-exp-lcr-for-par-pump-3-vass-2}.
Let $\rho$ be the shortest run from $s'$ to $t'$ in $V'$.
The self-loops defined by the transitions in $T_S$ can be moved to the front of the run in $V_S$ while keeping the whole run valid, as their effects are non-negative.
So $\rho$ can be rearranged into $s' \xrightarrow{\rho_1} q(\Vec{z})\xrightarrow{\rho_2}t$, where $\Vec{z} \in \Vec{b} + (P_1 \cup P_0)^*$ and $\rho_2$ is also a run in $V$.
Let $\Vec{\lambda} \in \NN^{n}$ be the vector such that $\Vec{\lambda}(i)$ is the number of occurrences of the transition $(q,\Vec{p}_i, q)$ in $\rho_1$.
Then we have $\Vec{z} - \sum_{i = 1}^n \Vec{\lambda}(i) \cdot \Vec{p}_i \in \Vec{b} + P_0^*$ and
\begin{equation*}
    \Inner{\Vec{c},\Vec{\lambda}} \le \norm{\Vec{c}} \cdot \sum_{i=1}^n \Vec{\lambda}(i) \le \norm{\Vec{c}} \cdot |\rho_1| \le \norm{\Vec{c}} \cdot |\rho| = \norm{\Vec{c}} \cdot \ell \le K,
\end{equation*}
which implies that $\Vec{z} \in S$.
We can also estimate the norm of $\Vec{z}$ by
\begin{equation*}
    \norm{\Vec{z}} \le \norm{t} + \Size(V) \cdot |\rho_2| \le M \cdot (1 + \ell) \le (M+\ell)^2.
\end{equation*}
It follows from~\cref{eq:2-exp-lcr-for-par-pump-3-vass-1} that $s\xrightarrow{\sigma} q(\Vec{z})$ in $V$ for some $\sigma$ with 
\begin{equation*}
    |\sigma| \le g(\Size(V,s,q(\Vec{z})))^{k^2} \le (M + 3 \cdot \norm{\Vec{z}})^{mk^2} \le (M + \ell)^{4mk^2}.
\end{equation*}
Note that $s \xrightarrow{\sigma\rho_2} t$ in $V$, and hence, the length bound is given by
\begin{equation*}
    \begin{aligned}
        \min \Len(V, s, t) &\le |\sigma| + |\rho_2|\le 2\cdot (M + \ell)^{4mk^2} \le (M + \ell)^{6mk^2} \le R(M+\ell).
    \end{aligned}
\end{equation*}
The claim is proved. 
\end{proof}

\subsection{Proof of~\cref{claim:semi-diag-3}}
\ClaimSemiDiagThrid*
\begin{proof}
Clearly, $(V',s',t')$ is diagonal and $s' \xrightarrow{*} q(\Vec{y}) \xrightarrow{*} t'$ in $V'$, given that $\Vec{\delta} \in P_0$.
Recall that $R(x)=x^{21mk^2}$.
The size bound is derived as follows.
\begin{equation*}
    \begin{aligned}
        \Size(V',s',t') &\le M + 3\cdot\norm{\Vec{w}} + 3 \cdot(M + |T_{\Vec{w}}|) \cdot (\norm{T}+\norm{T_{\Vec{w}}} + 1)\\
        &\le M + 6\cdot L_k(M)^{14mk^2} + 81 \cdot M^{4mk^2}\\
        &\le 88L_k(M)^{14mk^2}\\
        &\le R(L_k(M)).
    \end{aligned}
\end{equation*}
We can similarly rearrange the shortest run $\rho$ from $s'$ to $t'$ into a run $s' \xrightarrow{\rho_1} q(\Vec{z})\xrightarrow{\rho_2}t'$, where $\Vec{z} \in \Vec{w} + P_0^* \subseteq S$ and $\rho_2$ is also a run from $q(\Vec{z})$ to $t$ in $V$.
We can also estimate the norm of $\Vec{z}$ by
\begin{equation*}
    \norm{\Vec{z}} \le \norm{t} + \Size(V) \cdot |\rho_2|\le (M + |\rho|)^2.
\end{equation*}
It follows from~\cref{eq:2-exp-lcr-for-par-pump-3-vass-1} that $s\xrightarrow{\sigma} q(\Vec{z})$ in $V$ for some $\sigma$ with 
\begin{equation*}
    |\sigma| \le g(\Size(V,s,q(\Vec{z})))^{k^2} \le \left(M +3 \cdot \norm{\Vec{z}}\right)^{mk^2} \le (M +|\rho|)^{4mk^2}.
\end{equation*}
Note that $s \xrightarrow{\sigma\rho_2} t$ in $V$, and consequently,
\begin{equation*}
    \min \Len(V, s, t) \le |\sigma| + |\rho_2| \le (M +|\rho|)^{6mk^2} \le R(M+|\rho|),
\end{equation*}
where $|\rho| = \min\Len(V',s',t')$.
\end{proof}

\section{Proofs for~\cref{sec:general-case}}\label{sec:appendix-general-case}

\subsection{Proof of~\cref{cor:diag-pump-of-unbounded}}

\corDiagPumpOfUnbounded*

\begin{proof}
    Let $U:= R(\norm{\Vec{x}} + 1 + |Q| \cdot \norm{T}, M)$, where $R$ is given by~\cref{lem:reach-unbounded-of-ever-unbounded}.
    If there is a run $\pi$ starting from $s$ with $\BDCounters(\pi, U) = \emptyset$, then by~\cref{lem:reach-unbounded-of-ever-unbounded}, there is a run $s \xrightarrow{\rho} q(\Vec{y})$ with 
    \begin{equation*}
        \Vec{y} \ge (\norm{\Vec{x}} + 1 + |Q| \cdot \norm{T}) \cdot \Vec{1} \ge \Vec{x} + (1 + |Q| \cdot \norm{T})\cdot\Vec{1}.
    \end{equation*}
    Choose any simple path $\sigma$ from $q$ back to $p$, and we have $s \xrightarrow{\rho\sigma} p(\Vec{y}')$, where 
    \begin{equation*}
        \Vec{y}':= \Vec{y} + \Eff{\sigma} \ge \Vec{x} + (1 + |Q| \cdot \norm{T})\cdot\Vec{1} - |\sigma|\cdot \norm{T}\cdot \Vec{1} \ge \Vec{x} + \Vec{1}.
    \end{equation*}
    Therefore $s$ is diagonal in $V$.
\end{proof}

\subsection{Proof of~\cref{claim:semi-par-1}}
\ClaimSemiParFirst*
\begin{proof}
It remains to show that $\min \Len(V,s,t) \leq R(\ell + M)$. Assume that the shortest run in $(V',s',t')$ is given by $s' \xrightarrow{\rho_1}p_2(\Vec{z})\xrightarrow{\rho_2}t' = t$ where $\rho_1$ contains only self-loops in $T_S$ and $\rho_2$ is a path of $V$.
Immediately, $|\rho_1|+|\rho_2| = \ell$, which implies that
\begin{equation*}
    \norm{\Vec{z}} \leq \norm{t} + \Size(V) \cdot \ell\le (M + \ell)^2.
\end{equation*}
Let $\Vec{\lambda} \in \NN^{n}$ be the vector such that $\Vec{\lambda}(i)$ is the number of occurrences of the transition $(q,\Vec{p}_i, q)$ in $\rho_1$.
Then we have $\Vec{z} - \sum_{i = 1}^n \Vec{\lambda}(i) \cdot \Vec{p}_i \in \Vec{b} + P_0^*$ and
\begin{equation*}
    \Inner{\Vec{c},\Vec{\lambda}} \le \norm{\Vec{c}} \cdot \sum_{i=1}^n \Vec{\lambda}(i) \le \norm{\Vec{c}} \cdot |\rho_1| \le \norm{\Vec{c}} \cdot \ell \le K,
\end{equation*}
which justifies that $\Vec{z} \in S$. 
This also implies that $s\xrightarrow{*}p_2(\Vec{z})$ in $V$.
The length of the shortest path for $(V,s,t)$ is bounded by
\begin{equation*}
    \begin{aligned}
        \min\Len(V,s,t) &\leq \min\Len(V,s,p_2(\Vec{z})) + \min\Len(V,p_2(\Vec{z}),t) \\
        &\leq \Size(V,s,p_2(\Vec{z}))^{3m^2} + \ell \\
        &\leq (M + 3\cdot \norm{\Vec{z}}) ^ {3m^2} + \ell \\
        &\leq (M + 3\cdot (M + \ell)^2) ^ {3m^2} + \ell \\
        &\leq (M + \ell) ^{13m^2} \\
        &\le R(M+\min\Len(V',s',t')),
    \end{aligned}
\end{equation*}
bearing in mind that $R(x)=x^{30m^2}$.
\end{proof}

\subsection{Proof of~\cref{claim:semi-par-2}}
\ClaimSemiParSecond*
\begin{proof}
Clearly, $(V',s',t')$ is a $(k{-}1)$-component forward pumpable VASS with $s' \xrightarrow{*} p_2(\Vec{x}_2) \xrightarrow{\pi'} t'$ in $V'$.
The size bound is derived as follows.
\begin{equation*}
    \begin{aligned}
        \Size(V',s',t') &\le M + 3\cdot\norm{\Vec{w}} + 3 \cdot(M + |T_{\Vec{w}}|) \cdot (\norm{T}+\norm{T_{\Vec{w}}} + 1)\\
        &\le M + 6\cdot L_k(M)^{25m^2} + 81 \cdot M^{12m^2}\\
        &\le 88\cdot L_k(M)^{25m^2}\\
        &\le R(L_k(M)).
    \end{aligned}
\end{equation*}
It is routine to rearrange the shortest run $\rho$ from $s'$ to $t'$ into a run $s' \xrightarrow{\rho_1} p_2(\Vec{z})\xrightarrow{\rho_2}t'$, where $\Vec{z} \in \Vec{w} + P_0^* \subseteq S$ and $\rho_2$ is also a run from $p_2(\Vec{z})$ to $t$ in $V$.
Notice that
\begin{equation*}
    \norm{\Vec{z}} \le \norm{t} + \Size(V) \cdot |\rho_2|\le (M + |\rho|)^2.
\end{equation*}
Thus, by~\cref{eq:2-exp-lcr-for-gen-3-vass-1}, it follows that $s\xrightarrow{\sigma} p_2(\Vec{z})$ in $V$ for some $\sigma$ with 
\begin{equation*}
    |\sigma| \le \Size(V,s,q(\Vec{z}))^{3m^2} \le \left(M +3 \cdot \norm{\Vec{z}}\right)^{3m^2} \le (M +|\rho|)^{12mk^2}.
\end{equation*}
It follows from $s \xrightarrow{\sigma\rho_2} t$ that
\begin{equation*}
    \min \Len(V, s, t) \le |\sigma| + |\rho_2| \le (M +|\rho|)^{14m^2} \le R(M+\min\Len(V',s',t')),
\end{equation*}
where $|\rho| = \min\Len(V',s',t')$ by definition.
\end{proof}

\end{document}